\documentclass[12pt]{article}
\usepackage{amsmath}
\usepackage{amsfonts}
\usepackage{amssymb}
\usepackage{amsthm}
\usepackage{graphicx}
\usepackage{fullpage}
\usepackage{setspace}
\usepackage{verbatim}
\usepackage{natbib}
\usepackage{appendix}
\usepackage{rotating}

\newtheorem{theorem}{Theorem}
\newtheorem{lemma}{Lemma}
\newtheorem{corollary}{Corollary}
\newtheorem{assumption}{Assumption}

\title{Asymptotically Exact Inference in Conditional Moment Inequality
  Models}

\author{Timothy B. Armstrong\thanks{email: timothya@stanford.edu.  Thanks
    to Han Hong and Joe Romano for guidance and many useful discussions,
    and to Liran Einav, Azeem Shaikh, Tim Bresnahan, Guido Imbens, Raj
    Chetty, Whitney Newey, Victor Chernozhukov, Jerry Hausman, Andres
    Santos, Elie Tamer, Vicky Zinde-Walsh, Alberto Abadie, Karim Chalak,
    Xu Cheng,
    Stefan Hoderlein,
    Don Andrews, Peter Phillips, Taisuke Otsu, Ed Vytlacil, Xiaohong Chen,
    Yuichi Kitamura
    and participants at seminars at Stanford and MIT
    for helpful comments and criticism.  All remaining errors are
    my own.  This paper was written with generous support from a
    fellowship from the endowment in memory of
    B.F. Haley and E.S. Shaw through the Stanford Institute for Economic
    Policy Research.
}
\\
Stanford University}
\date{November 10, 2011}

\begin{document}

\maketitle

\vspace{-.2in}

\begin{center}
JOB MARKET PAPER
\end{center}

\begin{abstract}
This paper derives the rate of convergence and asymptotic distribution for
a class of Kolmogorov-Smirnov style test statistics for conditional
moment inequality models %
for
parameters on the boundary of the identified set
under general conditions.
In contrast to other moment inequality settings, the rate of
convergence is faster than root-$n$, and the asymptotic distribution
depends entirely on nonbinding moments.
The results require the development of new techniques that draw
a connection between moment selection, irregular identification, bandwidth
selection and nonstandard M-estimation.
Using these results, I propose tests that are more powerful than existing
approaches for choosing critical values for this test statistic.
I
quantify the power improvement by
showing that the new tests
can detect alternatives that converge to points on the identified set
at a faster rate than those detected by existing approaches.
A monte carlo study
confirms that the tests and the asymptotic approximations they use
perform well in finite samples.  %
In an application to a regression of prescription drug expenditures on
income with interval data from the Health and Retirement Study, confidence
regions based on the new tests are substantially tighter than those based
on existing methods.

\end{abstract}

\section{Introduction}\label{introduction_sec}

Theoretical restrictions used for estimation of economic models often take
the form of moment inequalities.
Examples include models of consumer demand and strategic interactions
between firms,
bounds on treatment effects using instrumental variables restrictions,
and various forms of censored and
missing data
\citep[%
  see, among many others,][and papers cited
therein]{,manski_nonparametric_1990,manski_inference_2002,pakes_moment_2006,ciliberto_market_2009,chetty_bounds_2010}.
For these models, the restriction often takes the form
of moment inequalities conditional on some observed variable.  That is,
given a sample
$(X_1,W_1),\ldots (X_n,W_n)$, we are
interested in testing a null hypothesis of the form
$E(m(W_i,\theta)|X_i)\ge 0$ with probability one, where the inequality is
taken elementwise if $m(W_i,\theta)$ is a vector. 
Here, $m(W_i,\theta)$ is a known function of an observed random
variable $W_i$, which may include $X_i$, and 
a parameter $\theta\in\mathbb{R}^{d_\theta}$, and the moment inequality
defines the identified set $\Theta_0\equiv \{\theta|
E(m(W_i,\theta)|X_i)\ge 0 \text{ a.s.}\}$ of parameter values that cannot
be ruled out by the data and the restrictions of the model.

In this paper, I consider inference in models defined by conditional
moment inequalities.  I focus on %
test
statistics that exploit the equivalence between the null hypothesis
$E(m(W_i,\theta)|X_i)\ge 0$ almost surely and
$Em(W_i,\theta)I(s<X_i<s+t)\ge 0$ for all
$(s,t)$.  Thus, we can use $\inf_{s,t} \frac{1}{n}\sum_{i=1}^n
m(W_i,\theta)I(s<X_i<s+t)$, or the infimum of some
weighted version of the unconditional moments indexed by $(s,t)$.
Following the terminology commonly used in the literature, I refer to
these as Kolmogorov-Smirnov (KS) style test statistics.
The main contribution of this paper %
is to
derive the rate of convergence and asymptotic distribution of this test
statistic for parameters on the boundary of the identified set under a
general set of conditions.  The asymptotic distributions derived in this
paper and the methods used to derive them fall into a different category
than other asymptotic distributions derived in the conditional moment
inequalities and goodness-of-fit testing literatures.  Rather, the
asymptotic distributions and rates of convergence derived here
resemble more closely those of maximized objective functions
for nonstandard M-estimators
\citep[see, for example,][]{kim_cube_1990},
but
require new methods to derive.
The results draw a connection between moment selection, bandwidth
selection, irregular identification and nonstandard M-estimation.

While asymptotic distribution results are available for this statistic in
some cases \citep{andrews_inference_2009,kim_kyoo_il_set_2008}, the
existing results give only a conservative upper bound of
$\sqrt{n}$ on
the rate of convergence of this test statistic in a large class of
important cases.  For example, in the interval regression model, the
asymptotic distribution of this test statistic for parameters on the
boundary of the identified set and the proper scaling needed to achieve
it have so far been unknown in the generic case
(see Section \ref{overview_sec} for the definition of this model).
In these cases, results available in the literature do not give an
asymptotic distribution result, but state
only that the test statistic converges in probability to zero when scaled
up by $\sqrt{n}$.  This paper derives the scaling that leads to a
nondegenerate asymptotic distribution and characterizes this
distribution.
Existing results can be used for conservative inference in these cases
(along with tuning parameters to prevent the critical value from going to
zero),
but lose power relative to procedures that use the results derived in this
paper to choose critical values based on the asymptotic distribution of
the test statistic on the boundary of the identified set.

To quantify this power improvement, I
show that using the asymptotic distributions derived in this paper gives
power against %
sequences of parameter values that approach points on the boundary of the
identified set
at a faster rate than those detected
using root-$n$ convergence to a
degenerate distribution.
Since local power results have not been available for the conservative
approach based on root-$n$ approximations in this setting, %
making this comparison
involves
deriving new local power results for the existing tests in addition to
the new tests.
The increase in power is %
substantial.  In the leading case considered in Section \ref{inf_dist_sec},
I find that the methods developed in this paper give power against local
alternatives that approach the
identified set at a $n^{-2/(d_X+4)}$ rate (where $d_X$ is the dimension of
the conditioning variable), while using conservative $\sqrt{n}$
approximations only gives power against $n^{-1/(d_X+2)}$ alternatives.
The power improvements are not completely free, however, as the new tests
require smoothness conditions not needed for existing approaches.  In
another paper \citep{armstrong_weighted_2011}, I propose a modification of
this test statistic that achieves a similar power improvement (up to a
$\log n$ term) without sacrificing the robustness of the conservative
approach.  See Section \ref{discussion_sec} for more on these tradeoffs.

To examine how well these asymptotic approximations describe sample sizes
of practical importance, I perform a monte carlo study.
Confidence regions based on the tests proposed in this paper
have close to the nominal
coverage in the monte carlos, and shrink to the identified set at a faster
rate than those based on existing tests.
In addition, I provide an empirical illustration
examining the relationship between
out of pocket prescription spending and %
income in a data set
in which out of pocket prescription spending is sometimes missing or
reported as an interval.
Confidence regions for this application constructed using the methods in
this paper are
substantially tighter than those that use existing methods %
(these confidence regions are reported in Figures \ref{cr95_est_fig} and
\ref{cr95_cons_fig} and Table \ref{ci_table}; see Section
\ref{application_sec} for the details of the empirical illustration).

While the asymptotic distribution results in this paper are technical in
nature, the key insights can be described at an intuitive level.  I
provide a nontechnical
exposition of these ideas in Section
\ref{overview_sec}.  Together with the statements of the asymptotic
distribution results in Section \ref{inf_dist_sec} and the local power
results in Section \ref{local_alt_sec}, this provides a general picture of
the results of the paper.
The rest of this section discusses the relation of these results to the
rest of the literature, and introduces notation and definitions.
Section \ref{inf_dist_alpha_sec} generalizes
the asymptotic distribution results of Section \ref{inf_dist_sec}, and
Sections \ref{inference_sec} and \ref{rate_test_sec} deal with estimation
of the asymptotic distribution for feasible inference.
Section \ref{mc_sec} presents monte carlo results.
Section \ref{application_sec} presents the empirical illustration.
In Section \ref{discussion_sec}, I discuss some implications of these
results beyond the immediate application to constructing asymptotically
exact tests.  Section \ref{conclusion_sec} concludes.  Proofs are in the
appendix.

\subsection{Related Literature}

The results in this paper relate to recent work on testing conditional
moment inequalities, including papers by \citet{andrews_inference_2009},
\citet{kim_kyoo_il_set_2008}, \citet{khan_inference_2009},
\citet{chernozhukov_intersection_2009},
\citet{lee_testing_2011},
\citet{ponomareva_inference_2010},
\citet{menzel_estimation_2008} and \citet{armstrong_weighted_2011}.
The results on the local power of asymptotically exact and conservative KS
statistic
based procedures derived in this paper are useful for comparing confidence
regions based on KS statistics to other methods of inference on the
identified set proposed in these papers.  \citet{armstrong_weighted_2011}
derives local power
results for some common alternatives to the KS statistics based on
integrated moments considered in this paper (the confidence
regions considered in that paper satisfy the stronger criterion of
containing the entire identified set, rather than individual points, with
a prespecified probability).  I compare the local power calculations in
this paper with those results in Section \ref{discussion_sec}.

Out of these existing approaches to inference on conditional moment
inequalities, the papers that are most closely related
to this one are
those by \citet{andrews_inference_2009} and \citet{kim_kyoo_il_set_2008}, both
of which consider statistics based on integrating the conditional
inequality.  As discussed above, the main contributions of the present
paper relative to these papers are (1) deriving the rate of
convergence and nondegenerate asymptotic distribution of this statistic
for parameters on the boundary of the identified set in the common case
where the results in these papers reduce to a statement that the statistic
converges to zero at a root-$n$ scaling and (2) deriving local power
results that show how much power is gained by using critical values based
on these new results.
\citet{armstrong_weighted_2011} uses a statistic similar to the one
considered here, but proposes an increasing sequence of weightings ruled
out by the assumptions of the rest of the literature (including the
present paper).  This leads to %
almost the same power
improvement %
as the methods in this paper even
when conservative critical values are used.
\citet{khan_inference_2009} propose a statistic
similar to one considered here for a model defined by conditional moment
inequalities, but consider point estimates and confidence intervals based
on these estimates under conditions that lead to point identification.
\citet{galichon_test_2009} propose a similar statistic for a class of
partially identified models under a different setup.  Statistics based on
integrating conditional moments have been used widely in other contexts as
well, and go back at least to \citet{bierens_consistent_1982}.

The literature on models defined by finitely many unconditional moment
inequalities is more developed, but still recent.
Papers in this literature include
\citet{andrews_confidence_2004}, %
\citet{andrews_inference_2008}, %
\citet{andrews_validity_2009}, %
\citet{andrews_inference_2010}, %
\citet{chernozhukov_estimation_2007}, %
\citet{romano_inference_2010}, %
\citet{romano_inference_2008}, %
\citet{bugni_bootstrap_2010},
\citet{beresteanu_asymptotic_2008},
\citet{moon_bayesian_2009},
\citet{imbens_confidence_2004}
and \citet{stoye_more_2009}.
While most of this
literature does not apply directly to the problems considered in this
paper when the conditioning variable is continuous, ideas from these
papers have been used in the literature on conditional moment inequality
models and other problems involving inference on sets.  Indeed, some of
these results are stated in a broad enough way to apply to the general
problem of inference on partially identified models.

\subsection{Notation and Definitions}

Throughout this paper, I use the terms asymptotically exact and
asymptotically conservative to refer to the behavior of tests for a fixed
parameter value under a fixed probability distribution.  I refer to a test
as asymptotically exact for testing a parameter $\theta$ under a data
generating process $P$ such that the null hypothesis holds if
the probability of rejecting $\theta$ converges to the nominal
level as the number of observations increases to infinity under $P$.  I
refer to a test as asymptotically conservative for testing a parameter
$\theta$ under a data generating process $P$ if the probability of falsely
rejecting $\theta$ is asymptotically strictly less than the nominal level
under $P$.  While this contrasts with a definition where a test is
conservative only if the size of the test is less than the nominal size
taken as the supremum of the
probability of rejection over a composite null of
all possible values of $\theta$ and $P$ such that $\theta$ is in the
identified set under $P$, it facilitates discussion of results like the
ones in this paper (and other papers that deal with issues related to
moment selection) that characterize the behavior of tests for different
values of $\theta$ in the identified set.

I use the following notation in the rest of the paper.  For observations
$(X_1,W_1),\ldots,(X_n,W_n)$ and a
measurable function $h$ on the sample space, $E_nh(X_i,W_i)\equiv
\frac{1}{n}\sum_{i=1}^n h(X_i,W_i)$ denotes the
sample mean.  I use double
subscripts to denote elements of vector observations so that $X_{i,j}$
denotes the $j$th component of the $i$th observation $X_i$.
Inequalities
on Euclidean space refer to the partial ordering of elementwise
inequality.  For a vector valued function
$h:\mathbb{R}^\ell\to\mathbb{R}^m$, the infimum of $h$ over a set
$T$ is defined to be the vector consisting of the infimum of each element:
$\inf_{t\in T} h(t)\equiv (\inf_{t\in T} h_1(t),\ldots,\inf_{t\in T}
h_m(t))$.  I use $a\wedge b$ to denote the elementwise minimum and $a\vee
b$ to denote the elementwise maximum of $a$ and $b$.  The notation
$\lceil x\rceil$ denotes the least integer greater than or equal to $x$.

\section{Overview of Results}\label{overview_sec}

The asymptotic distributions derived in this paper arise when the
conditional moment inequality binds only on a probability zero set.
In contrast to inference with finitely many unconditional moment
inequalities, in which at least one moment inequality will bind on the
boundary of the identified set and limiting distributions of test
statistics are degenerate only on the interior of the identified set, this
lack of nondegenerate binding moments holds even on the boundary of the
identified set in typical applications.
This leads to a faster than
root-$n$ rate of convergence to an asymptotic distribution that depends
entirely on moments that are close to, but not quite binding.

To see why this case is typical in applications, consider an application
of moment inequalities to regression with interval data.  In the interval
regression model,
$E(W_i^*|X_i)=X_i'\beta$, and $W_i^*$ is unobserved, but known to be between
observed variables $W_i^H$ and $W_i^L$, so that $\beta$ satisfies the moment
inequalities
\begin{align*}
E(W_i^L|X_i)\le X_i'\beta\le E(W_i^H|X_i).
\end{align*}
Suppose that
the distribution of $X_i$ is absolutely continuous with respect to the
Lebesgue measure.  Then, to have one of these inequalities bind on a
positive probability set, $E(W_i^L|X_i)$ or $E(W_i^H|X_i)$ will have to be
linear on this set.  Even if this is the case, this only means that the
moment inequality will bind on this set for one value of $\beta$, and
the moment inequality will typically not bind when applied to
nearby values of $\beta$ on the boundary of the identified set.  Figures
\ref{int_reg_smooth_fig} and \ref{int_reg_rootn_fig} illustrate this for
the case where the conditioning variable is one dimensional.  Here, the
horizontal axis is the nonconstant part of $x$, and the vertical axis
plots the conditional mean of the $W_i^H$ along with regression functions
corresponding to points in the identified set.  Figure
\ref{int_reg_smooth_fig} shows a case where the KS statistic converges at
a faster than root-$n$ rate.  In Figure \ref{int_reg_rootn_fig}, the
parameter $\beta_1$ leads to convergence at exactly a root-$n$ rate, but
this is a knife edge case, since the KS statistic for testing
$\beta_2$ will converge at a faster rate.

This paper derives asymptotic distributions under conditions that
generalize these cases to arbitrary moment functions $m(W_i,\theta)$.
In this broader setting,
KS statistics converge at a faster than root-$n$ rate on the boundary of
the identified set under general conditions when the model is set
identified and at least one conditioning variable is continuously
distributed.
In interval
quantile regression, contact sets
for the conditional median translate to contact sets for the conditional
mean of the moment function, leading to faster than root-$n$ rates of
convergence in similar settings.  Bounds in selection models, such as
those
proposed by \citet{manski_nonparametric_1990}, lead to a similar
setup to the interval regression model, as do some of the structural
models considered by \citet{pakes_moment_2006}, with the intervals
depending on a first stage parameter estimate.
See \citet{armstrong_weighted_2011} for primitive conditions for a set of
high-level conditions similar to the ones used in this paper for some of
these models.

While the results hold more generally, the rest of this section describes
the results in the context of the interval regression example in a
particular case.
Consider deriving the rate of convergence and nondegenerate asymptotic
distribution of the KS statistic for a parameter $\beta$ like the one
shown in Figure \ref{int_reg_smooth_fig}, but with $X_i$ possibly
containing more than one covariate.
Since
the lower bound never binds, it is intuitively clear that the KS statistic
for the lower bound will converge to zero at a
faster rate than the KS statistic for the upper bound, so consider the KS
statistic for the upper bound given by $\inf_{s,t} E_nY_iI(s<X_i<s+t)$ where
$Y_i=W^H_i-X_i'\beta$.  If $E(W_i^H|X_i=x)$ is tangent to $x'\beta$ at a
single point $x_0$, and $E(W_i^H|X_i=x)$ has a positive second derivative
matrix $V$ at this
point, we will have $E(Y_i|X_i=x)\approx (x-x_0)'V(x-x_0)$ near $x_0$, so
that, for $s$ near $x_0$ and $t$ close to zero,
$EY_iI(s<X_i<s+t)\approx f_X(x_0)\int_{s_1}^{s_1+t_1}\cdots
\int_{s_{d_X}}^{s_{d_X}+t_{d_X}}
(x-x_0)'V(x-x_0) \, dx_{d_X}\cdots dx_1$ (here, if the regression contains
a constant, the conditioning variable $X_i$ is redefined to be the
nonconstant part of the regressor, so that $d_X$ refers to the dimension
of the nonconstant part of $X_i$).

Since $EY_iI(s<X_i<s+t)=0$ only when $Y_iI(s<X_i<s+t)$ is degenerate, the
asymptotic behavior of the KS statistic should depend on indices $(s,t)$
where the moment inequality is not quite binding, but close enough to
binding that sampling error makes $E_nY_iI(s<X_i<s+t)$ negative some of the
time.  To determine on which indices $(s,t)$ we should expect this to
happen, split up the process in the KS statistic into a mean zero process
and a drift term: $(E_n-E)Y_iI(s<X_i<s+t)+EY_iI(s<X_i<s+t)$.  In order for this to
be strictly negative some of the time, there must be non-negligible
probability that the mean zero process is greater in absolute value than
the drift term.  That is, we must have $sd((E_n-E)Y_iI(s<X_i<s+t))$ of at
least the same order of magnitude as $EY_iI(s<X_i<s+t)$.  The idea is similar
to rate of convergence arguments for M-estimators with possibly
nonstandard rates of convergence, such as those considered by
\citet{kim_cube_1990}.  We have
$sd((E_n-E)Y_iI(s<X_i<s+t))=\mathcal{O}(\sqrt{\prod_i
  t_i}/\sqrt{n})$ for small $t$, and some calculations show that, for $s$
close to $x_0$, $EY_iI(s<X_i<s+t)\approx f_X(x_0)\int_{s_1}^{s_1+t_1}\cdots
\int_{s_{d_X}}^{s_{d_X}+t_{d_X}} (x-x_0)'V(x-x_0) \, dx_{d_X}\cdots dx_1\ge C
\|(s-x_0,t)\|^2\prod_i t_i$ for some $C>0$.  Thus, we expect the asymptotic
distribution to depend on $(s,t)$ such that $\sqrt{\prod_i t_i}/\sqrt{n}$
is of the same or greater order of magnitude than $\|(s-x_0,t)\|^2\prod_i
t_i$, which corresponds to $\|(s-x_0,t)\|^2\sqrt{\prod_i t_i}$ less than
or equal to $\mathcal{O}(1/\sqrt{n})$.

To get the main intuition for the rate of convergence, let us first
suppose that $s-x_0$ is of the same order of magnitude as $t$, and the
components $t_i$ of $t$ are of the same order of magnitude, and show
separately that cases where components of $(s,t)$ converge at different
rates do not matter for the asymptotic distribution. %
If $s-x_0$ and
all components $t_i$ are to converge to zero at the same rate $h_n$, we
must have $\|(s-x_0,t)\|=\mathcal{O}(h_n)$ and $\prod_i
t_i=\mathcal{O}(h_n^{d_X})$, so that, if $\|(s-x_0,t)\|^2\sqrt{\prod_i t_i}\le
\mathcal{O}(1/\sqrt{n})$,  we will have $\mathcal{O}(1/\sqrt{n})\ge
h_n^2\sqrt{h_n^{d_X}}=h_n^{2+{d_X}/2}$ so that $h_n\le
\mathcal{O}(1/n^{1/(2(2+d_X/2))})=\mathcal{O}(n^{-1/(4+d_X)})$.
Then, for $(s,t)$ with $t$ in an $h_n$-neighborhood of zero, we will have
$(E_n-E)Y_iI(s<X_i<s+t)=\mathcal{O}_P(\sqrt{\prod_i t_i}/\sqrt{n})
=\mathcal{O}_P(n^{-({d_X}+2)/({d_X}+4)})$.

Next suppose that $s$ or converges to $x_0$ more slowly than
$h_n=n^{-1/({d_X}+4)}$ or that one of the components of $t$ converges to zero
more slowly than $h_n$.  In this case, we will have $\|(s-x_0,t)\|$
greater than some sequence $k_n$ with $k_n/h_n\to\infty$, so that, to have
$\|(s-x_0,t)\|^2\sqrt{\prod_i t_i}\le \mathcal{O}(1/\sqrt{n})$, we would
have to have $\sqrt{\prod_i t_i}\le \mathcal{O}(1/(k_n^2\sqrt{n}))$ so
that $(E_n-E)Y_i(s<X_i<s+t)$ will be of order less than $1/(k_n^2n)$, which
goes to zero at a faster rate than the $n^{-({d_X}+2)/({d_X}+4)}$ rate
that we get when the components of $(s,t)$ converge at the same rate.

Thus, we should expect that the values of $(s,t)$ that matter for the
asymptotic distribution of the KS statistic are those with $(s-x_0,t)$ of
order $n^{-1/({d_X}+4)}$, and that the KS statistic will converge in
distribution when scaled up by $n^{-({d_X}+2)/({d_X}+4)}$ to the infimum
of the limit of a sequence of local objective functions indexed by $(s,t)$
with $(s-x_0,t)$ in a sequence of $n^{-1/({d_X}+4)}$ neighborhoods of zero.
Formalizing this argument requires showing that this intuition holds
uniformly in $(s,t)$.  The formal proof uses a ``peeling'' argument along
the lines of \citet{kim_cube_1990},
but a different type of argument is needed for regions where, even though
$\|(s-x_0,t)\|$ is far from zero, some components of $t$ are small enough
that $E_nY_iI(s<X_i<s+t)$ may be slightly negative because the region
$\{s<X_i<s+t\}$ is small and happens to catch a few observations with
$Y_i<0$.  The proof formalizes the intuition that these regions cannot
matter for the asymptotic distribution,
since $\prod_i t_i$ must be much
smaller than when $s$ is close to $x_0$ and the components of $t$ are of
the same order of magnitude as each other.

These results can be used for inference once the asymptotic distribution
is estimated.
In Section \ref{inference_sec}, I
describe two procedures for estimating this asymptotic distribution.  The
first is a generic subsampling procedure that uses only the fact that the
statistic converges to a nondegenerate distribution at a known rate.  The
second is based on estimating a finite dimensional set of objects that
allows this distribution to be simulated.

Both procedures rely on the conditional mean having a positive definite
second derivative matrix near its minimum.  To form tests that are
asymptotically valid under more general conditions, I propose pre-tests
for these conditions, and embed these tests in a procedure that uses the
asymptotic approximation to the null distribution for which the pre-test
finds evidence.  I describe these pre-tests in Section
\ref{rate_test_sec}, but, before doing this, I extend the results of
Section \ref{inf_dist_sec} to a broader class of shapes of the conditional
mean in Section \ref{inf_dist_alpha_sec}.  These results are useful for
the pre-tests in Section \ref{subsamp_rate_subsec}, which adapt methods
from \citet{politis_subsampling_1999} for estimating rates of convergence
to this setting.  Section \ref{deriv_rate_subsec} describes another
pre-test for the conditions of Section \ref{inf_dist_sec}, this one based
on estimating the second derivative and testing for positive definiteness.
The pre-tests are valid under regularity conditions governing the
smoothness of the conditional mean.

One of the appealing features of using asymptotically exact critical
values over conservative ones is the potential for more power against
parameters outside of the identified set.
In Section \ref{local_alt_sec}, I consider power against local
alternatives.  I describe the intuition for the results in more detail in
that section, but the main idea is that, for a sequence of alternatives
$\theta_n$ converging to a point $\theta$ on the identified set that under
which the argument described above goes through, the drift process has an
additional term $E(m(W_i,\theta_n)-m(W_i,\theta))I(s<X<s+t)$, where
$s-x_0$ and $t$ are of order $h_n$.  The exact asymptotics will detect
$\theta_n$ when this term is of order $n^{-(d_X+2)/(d_X+4)}$, while
conservative asymptotics will have power only when $\theta_n$ is large
enough so that this term is of order $n^{-1/2}$.  This leads to power
against local alternatives of order $n^{-2/(d_X+4)}$ for the
asymptotically exact critical values, and $n^{-1/(d_X+2)}$ when the
conservative $\sqrt{n}$ approximation is used.

\section{Asymptotic Distribution of the KS Statistic}\label{inf_dist_sec}

Given iid observations $(X_1,W_1),\ldots,(X_n,W_n)$, of random variables
$X_i\in\mathbb{R}^{d_X}$, $W_i\in\mathbb{R}^{d_W}$, we wish to test the null
hypothesis that $E(m(W_i,\theta)|X_i)\ge 0$ almost surely, where
$m:\mathbb{R}^{d_W}\times \Theta\to\mathbb{R}^{d_Y}$ is a known measurable
function and $\theta\in\Theta\subseteq\mathbb{R}^{d_\theta}$ is a fixed
parameter value.  I use the notation $\bar m(\theta,x)$ to denote a
version of $E(m(W_i,\theta)|X=x)$ (it will be clear from context which
version is meant when this matters).
In some cases when it is clear which parameter value is
being tested, I will define $Y_i=m(W_i,\theta)$ for notational
convenience.  Defining $\Theta_0$ to
be the identified set of values of $\theta$ in $\Theta$ that satisfy
$E(m(W_i,\theta)|X_i)\ge 0$ almost surely, these tests can then be
inverted to obtain a confidence region that, for every
$\theta_0\in\Theta_0$, contains $\theta_0$ with a prespecified
probability \citep{imbens_confidence_2004}.  The tests considered here
will be based on asymptotic approximations, so that these statements
will only hold asymptotically.

The results in this paper allow for asymptotically
exact inference using KS style statistics in cases where the $\sqrt{n}$
approximations for these statistics are degenerate.  This includes the
case described in the introduction
in which one component of $E(m(W_i,\theta)|X_i)$
is tangent to zero at a single point and the rest are bounded away from
zero.
While this case captures the essential intuition for the results in this
paper, I state the results in a slightly more general way in order to make
them more broadly applicable.
I allow each component of $E(m(W_i,\theta)|X)$ to be tangent to zero at
finitely many points, which may be different for each component.  This
is relevant in the interval regression example for parameters for which
the regression line is tangent to $E(W^H_i|X)$ and $E(W^L_i|X)$ at
different points.  In the case of an interval regression on a scalar and a
constant, the points in the identified set corresponding to the largest
and smallest values of the slope parameter will typically have this
property.

I consider KS style statistics that are a function of
$\inf_{s,t}E_nm(W_i,\theta)I(s<X_i<s+t)
=(\inf_{s,t}E_nm_1(W_i,\theta)I(s<X_i<s+t),\ldots,
\inf_{s,t}E_nm_{d_Y}(W_i,\theta)I(s<X_i<s+t))$.
Fixing some function $S:\mathbb{R}^{d_Y}\to\mathbb{R}_+$,
we can then reject for large values of
$S(\inf_{s,t}E_nm(W_i,\theta)I(s<X_i<s+t))$ (which correspond to more
negative values of the components of
$\inf_{s,t}E_nm(W_i,\theta)I(s<X_i<s+t)$ for typical choices of $S$).
Note that this is different in
general than taking $\sup_{s,t} S(E_nm(W_i,\theta)I(s<X_i<s+t))$,
although similar ideas will apply here.
Also, the moments $E_nm(W_i,\theta)I(s<X_i<s+t)$ are not weighted, but the
results could be extended to allow for a weighting function $\omega(s,t)$,
so that the infimum is over $\omega(s,t)E_nm(W_i,\theta)I(s<X_i<s+t)$ as
long as $\omega(s,t)$ is smooth and bounded away from zero and infinity.
The condition that the weight function be bounded uniformly in the sample
size, which is also imposed by \citet{andrews_inference_2009} and
\citet{kim_kyoo_il_set_2008}, turns out to be important
\citep[see][]{armstrong_weighted_2011}.

I formalize the notion that $\theta$ is at a point in the identified set
such that one or more of the components of $E(m(W_i,\theta)|X_i)$ is
tangent to zero at a finite number of of points in the following
assumption.

\begin{assumption}\label{smoothness_assump_multi}
For some version of $E(m(W_i,\theta)|X_i)$, the conditional mean of each
element of $m(W_i,\theta)$ takes its minimum only on a finite set
$\{x|E(m_j(W_i,\theta)|X=x)=0 \text{ some
  $j$}\}=\mathcal{X}_0=\{x_1,\ldots,x_\ell\}$.  For each $k$ from $1$
to $\ell$, let $J(k)$ be the set of indices $j$ for which
$E(m_j(W_i,\theta)|X=x_k)=0$.
Assume that there exist neighborhoods $B(x_k)$ of each
$x_k\in\mathcal{X}_0$ such that, for each $k$ from $1$ to $\ell$,
the following assumptions hold.
\begin{itemize}
\item[i.)] %
  $E(m_j(W_i,\theta)|X_i)$ is bounded away
  from zero outside of $\cup_{k=1}^\ell B(x_k)$ for all $j$ and, for
  $j\notin J(k)$, $E(m_j(W_i,\theta)|X_i)$ is bounded away from zero on
  $B(x_k)$.

\item[ii.)] %
  For $j\in J(k)$,
  $x\mapsto E(m_j(W_i,\theta)|X=x)$ has continuous
  second derivatives inside of the closure of $B(x_k)$ and a
  positive definite second derivative matrix $V_j(x_k)$ at each $x_k$.

\item[iii.)] $X$ has a continuous density $f_X$ on $B(x_k)$.

\item[iv.)] %
  Defining $m_{J(k)}(W_i,\theta)$ to have $j$th component
  $m_j(W_i,\theta)$ if $j\in J(k)$ and $0$ otherwise,
  $x\mapsto E(m_{J(k)}(W_i,\theta)m_{J(k)}(W_i,\theta)'|X_i=x)$ is
  finite and continuous on $B(x_k)$ for some version of this conditional
  second moment matrix.
\end{itemize}
\end{assumption}

Assumption \ref{smoothness_assump_multi} is the main substantive
assumption distinguishing the case considered here from the case where the
KS statistic converges at a $\sqrt{n}$ rate.  In the $\sqrt{n}$ case, some
component of $E(m(W_i,\theta)|X_i)$ is equal to zero on a positive
probability set.  Assumption \ref{smoothness_assump_multi} states
that any component of $E(m(W_i,\theta)|X_i)$ is equal to zero only on a
finite set, and that $X_i$ has a density in a neighborhood of this set, so
that this finite set has probability zero.  Note that the assumption that
$X_i$ has a density at certain points means that the moment inequalities
must be defined so that $X_i$ does not contain a constant.  Thus, the
results stated below hold in the interval regression example with $d_X$
equal to the number of nonconstant regressors.

Unless otherwise stated, I assume that the contact set $\mathcal{X}_0$ in
Assumption \ref{smoothness_assump_multi} is nonempty.  If Assumption
\ref{smoothness_assump_multi} holds with $\mathcal{X}_0$ empty so that the
conditional mean $\bar m(\theta,x)$ is bounded from below away from zero,
$\theta$ will typically be on the interior of the identified set (as long
as the conditional mean stays bounded away from zero when $\theta$ is
moved a small amount).  For such values of $\theta$, KS statistics will
converge at a faster rate (see Lemma \ref{large_s_lemma2_multi} in the
appendix), leading
to conservative inference even if the rates of convergence derived under
Assumption \ref{smoothness_assump_multi}, which are faster than
$\sqrt{n}$, are used.

In addition to imposing that the minimum of the components of the
conditional mean $\bar m(\theta,x)$ over $x$ are taken on a probability
zero set, Assumption \ref{smoothness_assump_multi} requires that this set
be finite, and that $\bar m(\theta,x)$ behave quadratically in $x$ near
this set.  I state results under this condition first, since it is easy to
interpret as arising from a positive definite second derivative matrix at
the minimum, and is likely to provide a good description of many situations
encountered in practice.  In Section \ref{inf_dist_alpha_sec}, I
generalize these results to other shapes of the conditional mean.  This is
useful for the tests for rates of convergence in Section
\ref{rate_test_sec}, since the rates of convergence turn out to be well
behaved enough to be estimated using adaptations of existing methods.

The next assumption is a
regularity condition that bounds $m_j(W_i,\theta)$ by a nonrandom
constant.  This assumption will hold naturally in models based on
quantile restrictions.  In the interval regression example, it requires
that the data have finite support.  This assumption could be replaced with
an assumption that $m(W_i,\theta)$ has exponentially decreasing tails, or
even a finite $p$th moment for some potentially large $p$ that would
depend on $d_X$ without much modification of the proof, but the finite
support condition is simpler to state.

\begin{assumption}\label{bdd_y_assump_multi}
For some nonrandom $\overline Y<\infty$, $|m_j(W_i,\theta)|\le \overline
Y$ with probability one for each $j$.
\end{assumption}

Finally, I make the following assumption on the function $S$.  Part of
this assumption could be replaced by weaker smoothness conditions, but the
assumption covers $x\mapsto \|x\|_{-}\equiv \|x\wedge 0\|$ for any norm
$\|\cdot\|$ as stated, which should suffice for practical purposes.

\begin{assumption}\label{S_assump}
$S:\mathbb{R}^{d_Y}\to\mathbb{R}_+$ is continuous %
and satisfies $S(ax)=aS(x)$ for any nonnegative
scalar $a$.
\end{assumption}

The following theorem gives the asymptotic distribution and rate of
convergence for $\inf_{s,t}E_nm(W_i,\theta)I(s<X_i<s+t)$ under these
conditions.  The distribution of
$S(\inf_{s,t}E_nm(W_i,\theta)I(s<X_i<s+t))$ under mild conditions on $S$
then follows as an easy corollary.

\begin{theorem}\label{inf_dist_thm_multi}
Under Assumptions \ref{smoothness_assump_multi} and
\ref{bdd_y_assump_multi},
\begin{align*}
n^{(d_X+2)/(d_X+4)} \inf_{s,t} E_nm(W_i,\theta)I(s<X_i<s+t)
\stackrel{d}{\to} Z
\end{align*}
where $Z$ is a random vector on $\mathbb{R}^{d_Y}$ defined as follows.
Let $\mathbb{G}_{P,x_k}(s,t)$, $k=1,\ldots,\ell$ be independent mean
zero Gaussian processes with sample paths in the space
$C(\mathbb{R}^{2d_X},\mathbb{R}^{d_Y})$ of continuous functions from
$\mathbb{R}^{2d_X}$ to $\mathbb{R}^{d_Y}$ and covariance kernel
\begin{align*}
cov(\mathbb{G}_{P,x_k}(s,t),\mathbb{G}_{P,x_k}(s',t'))
=E(m_{J(k)}(W_i,\theta)m_{J(k)}(W_i,\theta)'|X_i=x_k)f_X(x_k)
\int_{s\vee s'<x<(s+t)\wedge(s'+t')}\, dx
\end{align*}
where $m_{J(k)}(W_i,\theta)$ is defined to have $j$th element equal to
$m_j(W_i,\theta)$ for $j\in J(k)$ and equal to zero for $j\notin J(k)$.
For $k=1,\ldots,\ell$, let $g_{P,x_k}:\mathbb{R}^{2d_X}\to\mathbb{R}^{d_Y}$
be defined by
\begin{align*}
g_{P,x_k,j}(s,t)=\frac{1}{2}f_X(x_k)\int_{s_1}^{s_1+t_1}\cdots
\int_{s_{d_X}}^{s_{d_X}+t_{d_X}} x'V_j(x_k)x\, dx_{d_X}\cdots dx_1
\end{align*}
for $j\in J(k)$ and $g_{x_k,j}(s,t)=0$ for $j\notin J(k)$.  Define $Z$ to
have $j$th element
\begin{align*}
Z_j=\min_{k \text{ s.t. } j\in J(k)} \inf_{(s,t)\in\mathbb{R}^{2d_X}}
      \mathbb{G}_{P,x_k,j}(s,t)+g_{P,x_k,j}(s,t).
\end{align*}
\end{theorem}

The asymptotic distribution of $S(\inf_{s,t}E_nm(W_i,\theta)I(s<X_i<s+t))$
follows immediately from this theorem.

\begin{corollary}\label{inf_dist_S_cor}
Under Assumptions \ref{smoothness_assump_multi}, \ref{bdd_y_assump_multi},
and \ref{S_assump},
\begin{align*}
n^{(d_X+2)/(d_X+4)} S(\inf_{s,t}E_nm(W_i,\theta)I(s<X_i<s+t))
\stackrel{d}{\to} S(Z)
\end{align*}
for a random variable $Z$ with the distribution given in Theorem
\ref{inf_dist_thm_multi}.
\end{corollary}

These results will be useful for constructing asymptotically exact level
$\alpha$ tests if the asymptotic distribution does not have an atom at the
$1-\alpha$ quantile, and if the quantiles of the asymptotic distribution
can be estimated.  In the next section, I show that the asymptotic
distribution is atomless under mild conditions and propose two methods for
estimating the asymptotic distribution.  The first is a generic
subsampling procedure.  The second is a procedure based on estimating a
finite dimensional set of objects that determine the asymptotic
distribution.  This provides feasible methods for constructing
asymptotically exact confidence intervals under Assumption
\ref{smoothness_assump_multi}.  However, while, in many cases, this
assumption characterizes the distribution of $(X_i,m(W_i,\theta))$ for
most or all values of $\theta$ on the boundary of the identified set, it
is not an assumption that one would want to impose a priori.  Thus, these
tests should be embedded in a procedure that tests between this case and
cases where $E(m(W_i,\theta)|X)=0$ on a positive probability set, or where
$E(m(W_i,\theta)|X)$ is still equal to $0$ only at finitely many points, but
behaves like $x^4$ or the absolute value function or something
else near these points rather than a quadratic function.
In Section \ref{inf_dist_alpha_sec}, I generalize Theorem
\ref{inf_dist_thm_multi} to handle a wider set of shapes of the
conditional mean, with different rates of convergence for different cases.
In Section
\ref{rate_test_sec},
I propose procedures for testing for Assumption
\ref{smoothness_assump_multi} under mild smoothness conditions.  Combining
one of these preliminary tests with inference that is valid in the
corresponding case gives a procedure that is asymptotically valid under
more general conditions.  These include tests based on estimating the rate
of convergence directly, which use the results of Section
\ref{inf_dist_alpha_sec}.

\section{Inference}\label{inference_sec}

To ensure that the asymptotic distribution is continuous, we need to
impose additional assumptions to rule out cases where components of
$m(W_j,\theta)$ are degenerate.  The next assumption rules out these
cases.

\begin{assumption}\label{inv_mat_assump}
For each $k$ from $1$ to $\ell$, letting $j_{k,1},\ldots,j_{k,|J(k)|}$ be the
elements in $J(k)$, the matrix with $q,r$th element given by
$E(m_{j_{k,q}}(W_i,\theta)m_{j_{k,r}}(W_i,\theta)|X_i=x_k)$
is invertible.
\end{assumption}

This assumption simply says that the binding components of $m(W_i,\theta)$
have a nonsingular conditional covariance matrix at the point where they
bind.  A sufficient condition for this is for the conditional covariance
matrix of $m(W_i,\theta)$ given $X_i$ to be nonsingular at these points.

I also make the following assumption on the function $S$, which translates
continuity of the distribution of $Z$ to continuity of the distribution of
$S(Z)$.

\begin{assumption}\label{abs_cont_S_assump}
For any Lebesgue measure zero set $A$, $S^{-1}(A)$ has Lebesgue measure
zero.
\end{assumption}

Under these conditions, the asymptotic distribution in Theorem
\ref{inf_dist_thm_multi} is continuous.  In addition to showing that the
rate derived in that theorem is the exact rate of convergence (since the
distribution is not a point mass at zero or some other value), this shows
that inference based on this asymptotic approximation will be
asymptotically exact.

\begin{theorem}\label{abs_cont_thm}
Under Assumptions \ref{smoothness_assump_multi}, \ref{bdd_y_assump_multi},
and \ref{inv_mat_assump}, the asymptotic distribution in Theorem
\ref{inf_dist_thm_multi} is continuous.  If Assumptions
\ref{S_assump} and \ref{abs_cont_S_assump} hold as well, the asymptotic
distribution in Corollary \ref{inf_dist_S_cor} is continuous.
\end{theorem}

Thus, an asymptotically exact test of $E(m(W_i,\theta)|X_i)\ge 0$ can be
obtained by comparing the quantiles of
$S(\inf_{s,t}E_nm(W_i,\theta)I(s<X_i<s+t))$ to the quantiles of any
consistent estimate of the distribution of $S(Z)$.  I propose two methods
for estimating this distribution.  The first is a generic subsampling
procedure.  The second method uses the fact that the distribution of $Z$
in Theorem \ref{inf_dist_thm_multi} depends on the data generating process
only through finite dimensional parameters to simulate an estimate of the
asymptotic distribution.

Subsampling is a
generic procedure for estimating the distribution of a statistic using
versions of the statistic formed with a smaller sample size
\citep{politis_subsampling_1999}.  Since many independent smaller samples
are available, these can be used to estimate the distribution of the
original statistic as long as the distribution of the scaled statistic is
stable as a function of the sample size.  To describe the subsampling
procedure, let $T_n(\theta)=\inf_{s,t}E_nm(W_i,\theta)I(s<X_i<s+t)$. For
any set of indices $\mathcal{S}\subseteq \{1,\ldots,n\}$, define
$T_{\mathcal{S}}(\theta)=\inf_{s,t}\frac{1}{|\mathcal{S}|}
\sum_{i\in \mathcal{S}}m(W_i,\theta)I(s<X_i<s+t)$.  The subsampling
estimate of $P(S(Z)\le t)$ is, for some subsample size $b$,
\begin{align*}
\frac{1}{{n\choose b}}\sum_{|\mathcal{S}|=b} I\left(b^{(d_X+2)/(d_X+4)}
S(T_{\mathcal{S}}(\theta))\le t\right).
\end{align*}
One can also estimate the null distribution using the centered subsampling estimate
\begin{align*}
\frac{1}{{n\choose b}}\sum_{|\mathcal{S}|=b} I\left(b^{(d_X+2)/(d_X+4)}
[S(T_{\mathcal{S}}(\theta))-S(T_n(\theta))]\le t\right).
\end{align*}
For some nominal level $\alpha$, let $\hat q_{b,1-\alpha}$ be the
$1-\alpha$ quantile of either of these subsampling distributions.  We reject
the null hypothesis that $\theta$ is in the identified set at level $\alpha$ if
$n^{(d_X+2)/(d_X+4)}S(T_n(\theta))> \hat q_{b,1-\alpha}$ and fail to reject
otherwise.  The following theorem states that this procedure is
asymptotically exact.  The result follows immediately from general results
for subsampling in \citet{politis_subsampling_1999}.

\begin{theorem}\label{subsamp_thm}
Under Assumptions \ref{smoothness_assump_multi}, \ref{bdd_y_assump_multi},
\ref{S_assump}, \ref{inv_mat_assump} and \ref{abs_cont_S_assump}, the
probability of rejecting using the subsampling procedure described
above with nominal level $\alpha$ converges to $\alpha$ as long as
$b\to \infty$ and $b/n\to 0$.
\end{theorem}

While subsampling is valid under general conditions, subsampling estimates
may be less precise than estimates based on %
knowledge of how the asymptotic distribution relates to the data
generating process.  One possibility is to note that the asymptotic
distribution in Theorem \ref{inf_dist_thm_multi} depends on the underlying
distribution only through the set $\mathcal{X}_0$ and, for points $x_k$ in
$\mathcal{X}_0$, the density $f_X(x_k)$, the conditional second moment
matrix $E(m_{J(k)}(W_i,\theta)m_{J(k)}(W_i,\theta)'|X=x_k)$, and the
second derivative matrix $V(x_k)$ of the conditional mean.  Thus, with
consistent estimates of these objects, we can estimate the distribution in
Theorem \ref{inf_dist_thm_multi} by replacing these objects with their
consistent estimates and simulating from the corresponding distribution.

In order to accommodate different methods of estimating $f_X(x_k)$,
$E(m_{J(k)}(W_i,\theta)m_{J(k)}(W_i,\theta)'|X=x_k)$, and $V(x_k)$, I
state the consistency of these estimators as a high level condition, and
show that the procedure works as long as these estimators are consistent.
Since these objects only appear as
$E(m_{J(k)}(W_i,\theta)m_{J(k)}(W_i,\theta)'|X=x_k) f_X(x_0)$ and
$f_X(x_k)V(x_k)$ in the asymptotic distribution, we actually only need
consistent estimates of these objects.

\begin{assumption}\label{avar_est_assump}
The estimates $\hat M_k(x_k)$, $\hat f_X(x_k)$, and $\hat V(x_k)$ satisfy
$\hat f_X(x_k)\hat V(x_k)\stackrel{p}{\to} f_X(x_k)V(x_k)$ and
$\hat M_k(x_k)\hat f_X(x_k)\stackrel{p}{\to}
E(m_{J(k)}(W_i,\theta)m_{J(k)}(W_i,\theta)'|X=x_k)f_X(x_k)$.
\end{assumption}

For $k$ from $1$ to $\ell$, let $\hat {\mathbb{G}}_{P,x_k}(s,t)$ and $\hat
g_{P,x_k}(s,t)$ be the random process and mean function defined in the same
way as $\mathbb{G}_{P,x_k}(s,t)$ and $g_{P,x_k}(s,t)$, but with the
estimated quantities replacing the true quantities.
We estimate the distribution of $Z$ defined to have $j$th element
\begin{align*}
Z_j=\min_{m \text{ s.t. } j\in J(k)} \inf_{(s,t)\in\mathbb{R}^{2d}}
      \mathbb{G}_{P,x_k,j}(s,t)+g_{P,x_k,j}(s,t)
\end{align*}
using the distribution of $\hat Z$ defined to have $j$th element
\begin{align*}
\hat Z_j=\min_{k \text{ s.t. } j\in J(k)} \inf_{\|(s,t)\|\le B_n}
      \hat{\mathbb{G}}_{P,x_k,j}(s,t)+\hat g_{P,x_k,j}(s,t)
\end{align*}
for some sequence $B_n$ going to infinity.  The convergence of the
distribution $\hat Z$ to the distribution of $Z$ is in the sense of
conditional weak convergence in probability often used in proofs of the
validity of the bootstrap
\citep[see, for example,][]{lehmann_testing_2005}.
From this, it follows that tests that replace the quantiles of $S(Z)$ with
the quantiles of $S(\hat Z)$ are asymptotically exact under the conditions
that guarantee the continuity of the limiting distribution.

\begin{theorem}\label{Z_hat_thm}
Under Assumption \ref{avar_est_assump}, %
$\rho(\hat Z,Z)\stackrel{p}{\to} 0$ where $\rho$ is any metric on
probability distributions that metrizes weak convergence.
\end{theorem}

\begin{corollary}\label{Z_hat_cor}
Let $\hat q_{1-\alpha}$ be the $1-\alpha$ quantile of $S(\hat Z)$.  Then,
under Assumptions \ref{smoothness_assump_multi}, \ref{bdd_y_assump_multi},
\ref{S_assump}, \ref{inv_mat_assump}, \ref{abs_cont_S_assump}, and
\ref{avar_est_assump}, the test that rejects when
$n^{(d_X+2)/(d_X+4)}S(T_n(\theta))> \hat q_{1-\alpha}$ and fails to reject
otherwise is an asymptotically exact level $\alpha$ test.
\end{corollary}

If the set $\mathcal{X}_0$ is known, the quantities needed to compute
$\hat Z$ can be estimated consistently using standard methods for
nonparametric estimation of densities, conditional moments, and their
derivatives.
However, typically $\mathcal{X}_0$ is not known,
and the researcher will not even want to impose that this set is finite.
In Section \ref{rate_test_sec},
I propose methods for testing Assumption
\ref{smoothness_assump_multi} and estimating the set $\mathcal{X}_0$ under
weaker conditions on the smoothness of the conditional mean.  These
conditions allow for both the $n^{(d_X+2)/(d_X+4)}$ asymptotics that arise
from Assumption \ref{smoothness_assump_multi} and the $\sqrt{n}$
asymptotics that arise from a positive probability contact set.

Before describing these results, I extend the results of Section
\ref{inf_dist_sec} to other shapes of the conditional mean.  These results
are needed for the tests in Section \ref{subsamp_rate_subsec}, which rely
on the rate of convergence being sufficiently well behaved if it is in a
certain range.

\section{Other Shapes of the Conditional Mean}\label{inf_dist_alpha_sec}

Assumption \ref{smoothness_assump_multi} states that the components of the
conditional mean $\bar m(\theta,x)$ are minimized on a finite set and
have strictly positive second derivative matrices at the minimum.  More
generally, if the conditional mean is less smooth, or does not take an
interior minimum, $\bar m(\theta,x)$ could be minimized on a finite set,
but behave differently near the minimum.  Another possibility is that the
minimizing set could have zero probability, while containing infinitely
many elements (for example, an infinite countable set, or a lower
dimensional set when $d_X>1$).

In this section, I derive the asymptotic distribution and rate of
convergence of KS statistics under a broader class of shapes of the
conditional mean $\bar m(\theta,x)$.  I replace part (ii) of Assumption
\ref{smoothness_assump_multi} with the following assumption.

\begin{assumption}\label{smoothness_assump_alpha}
For $j\in J(k)$, $\bar m_j(\theta,x)=E(m_j(W_i,\theta)|X=x)$ is continuous
on $B(x_k)$ and satisfies
\begin{align*}
\sup_{\|x-x_k\|\le \delta} \left\|
  \frac{\bar m_j(\theta,x)-\bar m_j(\theta,x_k)}{\|x-x_k\|^{\gamma(j,k)}}
  -\psi_{j,k}\left(\frac{x-x_k}{\|x-x_k\|}\right) \right\|
\stackrel{\delta \to 0}{\to} 0
\end{align*}
for some $\gamma(j,k)>0$ and some function
$\psi_{j,k}:\{t\in\mathbb{R}^{d_X}|\|t\|=1\}\to \mathbb{R}$ with
$\overline \psi \ge \psi_{j,k}(t)\ge \underline \psi$ for some
$\overline \psi <\infty$ and $\underline \psi>0$.
For future reference, define
$\gamma=\max_{j,k}\gamma(j,k)$ and $\tilde J(k)=\{j\in J(k)|
\gamma(j,k)=\gamma\}$.
\end{assumption}

When Assumption \ref{smoothness_assump_alpha} holds,
the rate of convergence will be determined by
$\gamma$, and the asymptotic distribution will depend on the local
behavior of the objective function for $j$ and $k$ with $j\in \tilde
J(k)$.

Under Assumption \ref{smoothness_assump_multi}, Assumption
\ref{smoothness_assump_alpha} will hold with $\gamma=2$ and
$\psi_{j,k}(t)=\frac{1}{2}t V_j(x_k)t$ (this holds by a second order
Taylor expansion, as described in the appendix).  For $\gamma=1$,
Assumption \ref{smoothness_assump_alpha} states that $\bar m_j(\theta,x)$
has a directional derivative for every direction, with the approximation
error going to zero uniformly in the direction of the derivative.  More
generally, Assumption \ref{smoothness_assump_alpha} states that
$\bar m_j(\theta,x)$ increases like $\|x-x_k\|^\gamma$ near elements $x_k$
in the minimizing set $\mathcal{X}_0$.  For $d_X=1$, this follows from
simple conditions on the higher derivatives of the conditional mean with
respect to $x$.  With enough derivatives, the first derivative that is
nonzero uniformly on the support of $X_i$ determines $\gamma$.  I state
this formally in the next theorem.
For higher dimensions, Assumption \ref{smoothness_assump_alpha} requires
additional conditions to rule out contact sets of dimension less than
$d_X$, but greater than $1$.

\begin{theorem}\label{deriv_alpha_thm}
Suppose $\bar m(\theta,x)$ has $p$ bounded derivatives, $d_X=1$ and
$\text{supp}(X_i)=[\underline x,\overline x]$.
Then, if $\min_j \inf_x \bar m_j(\theta,x)=0$,
either Assumption \ref{smoothness_assump_alpha} holds, with the contact set
$\mathcal{X}_0$ possibly containing the boundary points $\underline x$ and
$\overline x$,
for $\gamma=r$
for some integer $r<p$, or, for some $x_0$ on the support of $X_i$ and
some finite $B$,
$\bar m_j(\theta,x)\le B |x-x_0|^p$ for some $j$.
\end{theorem}

Theorem \ref{deriv_alpha_thm} states that, with $d_X=1$ and $p$ bounded
derivatives, either Assumption \ref{smoothness_assump_alpha} holds for
$\gamma$ some integer less than $p$, or, for some $j$, $\bar
m_j(\theta,x)$ is less than
or equal to the function $B|x-x_0|^p$, which would make Assumption
\ref{smoothness_assump_alpha} hold for $\gamma=p$.  In the latter case,
the rate of convergence for the KS statistic must be at least as slow as
the rate of convergence when Assumption \ref{smoothness_assump_multi}
holds with $\gamma=p$.  While an interior minimum with a strictly positive
second derivative or a minimum at $\underline x$ or $\overline x$ with a
nonzero first derivative seem most likely, Theorem \ref{deriv_alpha_thm}
shows that Assumption \ref{smoothness_assump_alpha} holds under broader
conditions on the smoothness of the conditional mean.
This, along with the rates of convergence in
Theorem \ref{inf_dist_thm_alpha} below, will be useful for the
methods described later in Section \ref{rate_test_sec} for testing between
rates of convergence.  With enough smoothness assumptions on the
conditional mean, the rate of convergence will either be $n^\beta$ for
$\beta$ in some known range, or strictly slower than $n^\beta$ for some
known $\beta$.  With this prior knowledge of the possible types of
asymptotic behavior of $T_n(\theta)$ in hand, one can use a modified
version of the estimators of the rate of convergence proposed by
\citet{politis_subsampling_1999} to estimate $\gamma$ in Assumption
\ref{smoothness_assump_alpha}, and to test whether this assumption holds.

Under Assumption \ref{smoothness_assump_multi} with part (ii) replaced by
Assumption \ref{smoothness_assump_alpha}, the following modified version
of Theorem \ref{inf_dist_thm_multi}, with a different rate of convergence
and limiting distribution, will hold.

\begin{theorem}\label{inf_dist_thm_alpha}
Under Assumption \ref{smoothness_assump_multi}, with part (ii) replaced by
Assumption \ref{smoothness_assump_alpha}, and
Assumption \ref{bdd_y_assump_multi},
\begin{align*}
n^{(d_X+\gamma)/(d_X+2\gamma)} \inf_{s,t} E_nm(W_i,\theta)I(s<X_i<s+t)
\stackrel{d}{\to} Z
\end{align*}
where $Z$ is the random vector on $\mathbb{R}^{d_Y}$ defined as in Theorem
\ref{inf_dist_thm_multi}, but with $J(k)$ replaced by $\tilde J(k)$ and
$g_{P,x_k,j}(s,t)$ defined as
\begin{align*}
g_{P,x_k,j}(s,t)=f_X(x_k)\int_{s_1}^{s_1+t_1}\cdots
\int_{s_{d_X}}^{s_{d_X}+t_{d_X}}
  \psi_{j,k}\left(\frac{x}{\|x\|}\right)
    \|x\|^\gamma \, dx_{d_X}\cdots dx_1
\end{align*}
for $j\in \tilde J(k)$.  If Assumption \ref{S_assump} holds as well, then
\begin{align*}
n^{(d_X+\gamma)/(d_X+2\gamma)} S(\inf_{s,t}E_nm(W_i,\theta)I(s<X_i<s+t))
\stackrel{d}{\to} S(Z).
\end{align*}

If Assumption \ref{inv_mat_assump} holds as well, $Z$ has a continuous
distribution.  If Assumptions \ref{S_assump}, \ref{inv_mat_assump} and
\ref{abs_cont_S_assump} hold, $S(Z)$ has a continuous distribution.
\end{theorem}

Theorem \ref{inf_dist_thm_alpha} can be used once Assumption
\ref{smoothness_assump_alpha} is known to hold for some $\gamma$, as long
as $\gamma$ can be estimated.  I treat this topic in the next section.
Theorem \ref{deriv_alpha_thm} gives primitive conditions for this to hold
for the case where $d_X=1$ that rely only on the smoothness of the
conditional mean.  The only additional condition needed to use this
theorem is to verify that the set $\mathcal{X}_0$ does not contain the
boundary points $\underline x$ and $\overline x$.  In fact, the
requirement in Theorems \ref{inf_dist_thm_multi} and
\ref{inf_dist_thm_alpha} that $\mathcal{X}_0$ not contain boundary points
could be relaxed, as long as the boundary is sufficiently smooth.  The
results will be similar as long as the density of $X_i$ is bounded away
from zero on its support, and cases where the density of $X_i$ converges
to zero smoothly near its support could be handled using a transormation
of the data \citep[see][for an example of this approach in a slightly
  different setting]{armstrong_weighted_2011}.
Alternatively, a pre-test can be done to see if the conditional mean is
bounded away from zero near the boundary of the support of $X_i$ so that
these results can be used as stated.

\section{Testing Rate of Convergence Conditions}\label{rate_test_sec}

The $n^{(d_X+2)/(d_X+4)}$ convergence derived in Section
\ref{inf_dist_sec} holds when the minimum of $\bar
m_j(\theta,x)=E(m_j(W_i,\theta)|X_i=x)$ is taken at a finite number of
points, each with a strictly positive definite second derivative matrix.
The results in Section \ref{inf_dist_alpha_sec} extend these results to
other shapes of the conditional mean near the contact set, which result in
different rates of convergence.
In contrast, if the minimum is taken on a positive probability set,
convergence will be at the slower $\sqrt{n}$ rate.  Under additional
conditions on the smoothness of $\bar m_j(\theta,x)$ as a function
of $x$, it is possible to test for the conditions that lead to the faster
convergence rates.  In this section, I describe two methods for testing
between these conditions.  In Section \ref{subsamp_rate_subsec}, I
describe tests that use a generic test for rates of convergence based on
subsampling proposed by \citet{politis_subsampling_1999}.
These tests are valid as long as the KS statistic converges to a
nondegenerate distribution at some polynomial rate, or converges more
slowly than some imposed rate, and the results in Section
\ref{inf_dist_alpha_sec} give primitive conditions for this.  In Section
\ref{deriv_rate_subsec}, I propose tests of Assumption
\ref{smoothness_assump_multi} based on estimating the second derivative
matrix of the conditional mean.

\subsection{Tests Based on Estimating the Rate of Converence Directly}
  \label{subsamp_rate_subsec}

The pre-tests proposed in this section mostly follow Chapter 8 of
\citet{politis_subsampling_1999}, using the results in Section
\ref{inf_dist_alpha_sec} to give primitive conditions under which the rate
of convergence will be well behaved so that these results can be applied,
with some modifications to accomodate the possibility that the statistic
may not converge at a polynomial rate if the rate is slow enough.
Following the notation of
\citet{politis_subsampling_1999}, define
\begin{align*}
L_{n,b}(x|\tau)\equiv \frac{1}{{n \choose b}}
  \sum_{|\mathcal{S}|=b} I(\tau_b
      [S(T_{\mathcal{S}}(\theta))-S(T_n(\theta))]\le x)
\end{align*}
for any sequence $\tau_n$, and define
\begin{align*}
L_{n,b}(x|1)\equiv \frac{1}{{n \choose b}}
  \sum_{|\mathcal{S}|=b}
    I(S(T_{\mathcal{S}}(\theta))-S(T_n(\theta))\le x).
\end{align*}
Let
\begin{align*}
L_{n,b}^{-1}(t|1)=\inf \{x|L_{n,b}(x|1)\ge t\}
\end{align*}
be the $t$th quantile of $L_{n,b}(x|1)$, and define $L_{n,b}^{-1}(t|\tau)$
similarly.  Note that $\tau_b L_{n,b}^{-1}(t|1)=L_{n,b}^{-1}(t|\tau)$.  If
$\tau_n$ is the true rate of convergence, $L_{n,b_1}^{-1}(t|\tau)$ and
$L_{n,b_2}^{-1}(t|\tau)$ both approximate the $t$th quantile of the
asymptotic distribution.  Thus, if $\tau_n=n^\beta$ for some $\beta$,
$b_1^\beta L_{n,b_1}^{-1}(t|1)$ and $b_1^\beta L_{n,b_1}^{-1}(t|1)$ should
be approximately equal, so that an estimator for $\beta$ can be formed by
choosing $\hat \beta$ to set these quantities equal.  Some calculation
gives
\begin{align}\label{rate_est_formula}
\hat \beta
  = (\log L_{n,b_2}^{-1}(t|1)) - \log L_{n,b_1}^{-1}(t|1))
    /(\log b_1-\log b_2).
\end{align}
This is a special case of the class of estimators described in
\citet{politis_subsampling_1999} which allow averaging of more than two
block sizes and more than one quantile (these estimators could be used
here as well).

Note that the estimate $L_{n,b}(x|\tau)$ centers the subsampling draws
around the KS statistic $S(T_n(\theta))$ rather than its limiting value,
$0$.  This is necessary for the rate of convergence estimate not to
diverge under fixed alternatives.  Once the rate of convergence is known
or estimated, either $L_{n,b}(x|\tau)$ or an uncentered version, defined
as
\begin{align*}
\tilde L_{n,b}(x|\tau)\equiv \frac{1}{{n \choose b}}
  \sum_{|\mathcal{S}|=b} I(\tau_b
      S(T_{\mathcal{S}}(\theta))\le x),
\end{align*}
can be used to estimate the null distribution of the scaled statistic.

The results in \citet{politis_subsampling_1999} show that subsampling with
the estimated rate of convergence $n^{\hat \beta}$ is valid as long as the
true rate of convergence is $n^\beta$ for some $\beta>0$.  However, this
will not always be the case for the estimators considered in this paper.
For example, under the conditions of Theorem \ref{deriv_alpha_thm}, the
rate of convergence will either be $n^{(1+\gamma)/(1+2\gamma)}$ for some
$\gamma<p$ (here, $d_X=1$), or the rate of convergence will be at least as
slow as $n^{(1+p)/(1+2p)}$, but may converge at a slower rate, or
oscillate between slower rates of convergence.  Even if Assumption
\ref{deriv_alpha_thm} holds for some $\gamma$ for $\theta$ on the boundary
of the identified set, the rate of convergence will be faster for $\theta$
on the interior of the identified set, where trying not to be conservative
typically has little payoff in terms of power against parameters outside
of the identified set.

To remedy these issues, I propose truncated
versions defined as follows.  For some
$1/2\le \underline \beta<\overline \beta<1$, let $\hat\beta$ be the
estimate given by (\ref{rate_est_formula}) for $b_1=n^{\chi_1}$ and
$b_2=n^{\chi_2}$ for some $1>\chi_1>\chi_2>0$, and let $\hat\beta_a$ be the
estimate given by (\ref{rate_est_formula}) for $b_2=n^{\chi_a}$ for some
$1>\chi_a>0$ and $b_1$ some fixed constant that does not change with the
sample size (if $L_{n,b_1}^{-1}(t|1))=0$, replace this with an arbitrary
positive constant in the formula for $\hat\beta_a$ so that $\hat\beta_a$
is well defined).
The test described in the theorem below uses
$\hat\beta_a$ to test whether the rate of convergence is slow enough that
the conservative rate $n^{1/2}$ should be used, and uses $\hat \beta$ to
estimate the rate of convergence otherwise, as long as it is not
implausibly large.  If the rate of convergence is estimated to be larger
than $\overline \beta$ (which, for large enough $\overline \beta$, will
typically only occur on the interior of the identified set), the estimate
is truncated to $\overline \beta$.
When the rate of convergence is only known to be either $n^\beta$ for some
$\beta\in [\underline\beta,\overline\beta]$, or either slower than
$n^{\underline\beta}$ or faster than $n^{\overline\beta}$,
this procedure
provides a conservative approach that is
still asymptotically exact when the exponent of the rate of convergence is
in $(\underline\beta,\overline\beta)$.

\begin{theorem}\label{rate_thm_alpha}
Suppose that Assumptions \ref{bdd_y_assump_multi},
\ref{S_assump} and \ref{abs_cont_S_assump} hold, and that $S$ is convex
and $E(m(W_i,\theta)m(W_i,\theta)'|X_i=x)$ is continuous and strictly
positive definite.
Suppose that, for some $\overline \gamma$, %
Assumptions
\ref{smoothness_assump_multi} and \ref{inv_mat_assump}
hold with part (ii) of Assumption \ref{smoothness_assump_multi} replaced
by Assumption \ref{smoothness_assump_alpha} for some
$\gamma\le\overline\gamma$, where the set $\mathcal{X}_0=\{x|\bar
m_j(\theta,x)=0 \text{ some $j$}\}$ may be empty, or, for some
$x_0\in\mathcal{X}_0$ such that $X_i$ has a continuous density in a
neighborhood of $x_0$ and $B<\infty$, $\bar m_j(\theta,x)\le
B\|x-x_0\|^{\gamma}$ for some $\gamma>\overline \gamma$ and some $j$.

Let $\overline\beta=(d_X+\underline\gamma)/(d_X+2\underline\gamma)$
for some $\underline\gamma<\overline\gamma$ and let
$\underline \beta=(d_X+\overline\gamma)/(d_X+2\overline\gamma)$.
Let $\hat\beta$, $\hat\beta_a$ and
be defined as above for
some $0<\chi_1<\chi_2<1$ and $0<\chi_a<1$.
Consider the following test.  If $\hat\beta_a\ge\underline\beta$,
reject if
$n^{\hat\beta \wedge \overline\beta}S(T_n(\theta))
  >L_{n,b}(1-\alpha|b^{\hat\beta \wedge \overline\beta})$
(or if $n^{\hat\beta \wedge \overline\beta}S(T_n(\theta))
  >\tilde L_{n,b}(1-\alpha|b^{\hat\beta \wedge \overline\beta})$)
where $b=n^{\chi_3}$ for some $0<\chi_3<1$.
If $\hat\beta_a<\underline\beta$, perform any (possibly conservative)
asymptotically level $\alpha$ test that compares
$n^{1/2}S(T_n(\theta))$ to a critical value that is bounded away from
zero.

Under these conditions, this test is asymptotically level $\alpha$.
If Assumption \ref{smoothness_assump_multi} holds with part (ii) of
Assumption \ref{smoothness_assump_multi} replaced by Assumption
\ref{smoothness_assump_alpha} for some
$\underline\gamma<\gamma<\overline\gamma$ and
$\mathcal{X}_0$ nonempty, this test will be asymptotically exact level
$\alpha$.

\end{theorem}

In the one dimensional case, the conditions of Theorem
\ref{rate_thm_alpha} follow immediately from smoothness assumptions on the
conditional mean by Theorem \ref{deriv_alpha_thm}.
As discussed above, the condition that the
minimum not be taken on the boundary of the support of $X_i$ could be
removed, or the result can be used as stated with a pre-test for this
condition.

\begin{theorem}
Suppose that $d_X=1$, Assumptions \ref{bdd_y_assump_multi},
\ref{S_assump} and \ref{abs_cont_S_assump} hold, and that $S$ is convex
and $E(m(W_i,\theta)m(W_i,\theta)'|X_i=x)$ is continuous and strictly
positive definite.  Suppose that
$\text{supp}(X_i)=[\underline x,\overline x]$ and that
$\bar m(\theta,x)$ is bounded away from zero near $\underline x$ and
$\overline x$ and has $p$ bounded derivatives.  Then the conditions of
Theorem \ref{rate_thm_alpha} hold for any $\overline \gamma<p$.
\end{theorem}

\subsection{Tests Based on Estimating the Second Derivative}
  \label{deriv_rate_subsec}

I make the following assumptions on the conditional mean and the
distribution of $X_i$.  These conditions are used to estimate the second
derivatives of $\bar m(\theta,x)= E(m_j(W_i,\theta)|X_i=x)$, and the
results are stated for local polynomial estimates.  The conditions and
results here are from \citet{ichimura_chapter_2007}.
Other nonparametric estimators of conditional means and their
derivatives and conditions for uniform convergence of such estimators
could be used instead.  The results in this section related to testing
Assumption \ref{smoothness_assump_multi} are stated for $m_j(W_i,\theta)$
for a fixed index $j$.  The consistency of a procedure that combines these
tests for each $j$ then follows from the consistency of the test for each
$j$.

\begin{assumption}\label{cond_mean_assump}
The third derivatives of
$\bar m_j(\theta,x)$
with respect to $x$ are Lipschitz continuous and uniformly bounded.
\end{assumption}

\begin{assumption}\label{x_assump}
$X_i$ has a uniformly continuous density $f_X$ such that, for some compact
set $D\in\mathbb{R}^d$, $\inf_{x\in D} f_X(x)>0$, and
$E(m_j(W_i,\theta)|X_i)$ is bounded away from zero outside of $D$.
\end{assumption}

\begin{assumption}\label{x_given_y_assump}
The conditional density of $X_i$ given $m_j(W_i,\theta)$ exists and is
uniformly bounded.
\end{assumption}

Note that Assumption \ref{x_given_y_assump} is on the density of $X_i$
given $m_j(W_i,\theta)$, and not the other way around, so that, for
example, count data for the dependent variable in an interval regression
is okay.

Let $\mathcal{X}_0^j$ be the set of minimizers of
$\bar m_j(\theta,x)$
if this function is less than or equal to $0$
for some $x$ and the empty set otherwise.  In order to test Assumption
\ref{smoothness_assump_multi}, I first note that, if the conditional mean
is smooth, the positive definiteness of the second derivative matrix on
the contact set will imply that the contact set is finite.  This reduces
the problem to determining whether the second derivative matrix is
positive definite on the set of minimizers of $\bar m_j(\theta,x)$, a
problem similar to testing local identification conditions in nonlinear
models \citep[see][]{wright_detecting_2003}.
I record this observation in the following lemma.

\begin{lemma}\label{min_set_lemma}
Under Assumptions \ref{cond_mean_assump} and \ref{x_assump}, if
the second derivative matrix of $E(m_j(W_i,\theta)|X_i=x)$ is strictly
positive definite on $\mathcal{X}_0^j$,
then $\mathcal{X}_0^j$ must be finite.
\end{lemma}

According to Lemma \ref{min_set_lemma}, once we know that the second
derivative matrix of $E(m_j(W_i,\theta)|X_i)$ is positive definite on the
set of minimizers $E(m_j(W_i,\theta)|X_i)$, the conditions of Theorem
\ref{inf_dist_thm_multi}
will hold.  This reduces the problem to testing the conditions of the
lemma. One simple way of doing this is to take a preliminary estimate of
$\mathcal{X}_0^j$ that contains this set with probability approaching one,
and then test whether the second derivative matrix of
$E(m_j(W_i,\theta)|X_i)$ is positive definite on this set.  In
what follows, I describe an approach based on local polynomial regression
estimates of the conditional mean and its second derivatives, but other
methods of estimating the conditional mean would work under appropriate
conditions.  The methods require knowledge of a set $D$ satisfying
Assumption \ref{x_assump}.
This set could be chosen with another preliminary test, an extension which
I do not pursue.

Under the conditions above, we can estimate
$\bar m_j(\theta,x)$
and its derivatives at a given point $x$ with a local second order
polynomial regression estimator defined as follows.  For a kernel function
$K$ and a bandwidth
parameter $h$, run a regression of $m_j(W_i,\theta)$ on a second order
polynomial of $X_i$, weighted by the distance of $X_i$ from $x$ by
$K((X-x)/h)$.  That is,
for each $j$ and any $x$,
define $\hat{\bar m}_j(\theta,x)$, $\hat\beta_j(x)$, and $\hat V_j(x)$ to
be the values of $m$, $\beta$, and $V$ that minimize
\begin{align*}
&E_n
\left\{\left[m_j(W_i,\theta)
-\left(m + (X_i-x)'\beta
+ \frac{1}{2} (X_i-x)'V(X_i-x)
\right)\right]^2
\times K((X_i-x)/h)\right\}.
\end{align*}
The pre-test uses $\hat{\bar m}_j(\theta,x)$ as an estimate of $\bar
m_j(\theta,x)$ and $\hat V_j(x)$ as an estimate of $V_j(x)$.

The following theorem, taken from
\citet[Theorem 4.1]{ichimura_chapter_2007},
gives rates of convergence for these estimates of the conditional mean and
its second derivatives that will be used to estimate $\mathcal{X}_0^j$ and
$V_j(x)$ as described above.  The theorem uses an additional assumption on
the kernel $K$.

\begin{assumption}\label{K_assump}
The kernel function $K$ is bounded, has compact support, and satisfies,
for some $C$ and for any $0\le j_1+\cdots +j_r\le 5$,
$|u_1^{j_1}\cdots u_r^{j_r}K(u)-v_1^{j_1}\cdots v_r^{j_r}K(v)|
\le C \|u-v\|$.
\end{assumption}

\begin{theorem}\label{local_poly_conv}
Under iid data and Assumptions \ref{bdd_y_assump_multi},
\ref{cond_mean_assump}, \ref{x_assump}, \ref{x_given_y_assump}, and
\ref{K_assump},
\begin{align*}
\sup_{x\in D} %
\left|\hat V_{j,rs}(x)
-V_{j,rs}(x)\right|
=\mathcal{O}_p((\log n/(nh^{d_X+4}))^{1/2})+\mathcal{O}_p(h)
\end{align*}
for all $r$ and $s$, where $V_{j,rs}$ is the $r,s$ element of $V_j$, and
\begin{align*}
\sup_{x\in D} %
\left|\hat{\bar m}_j(\theta,x)-\bar m_j(\theta,x)\right|
=\mathcal{O}_p((\log n/(nh^{d_X}))^{1/2})+\mathcal{O}_p(h^3).
\end{align*}
\end{theorem}

For both the conditional mean and the derivative, the first term in the
asymptotic order of convergence is the variance term and the second is the
bias term.  The optimal choice of $h$ sets both of these to be the same
order, and is $h_n=(\log n/n)^{1/(d_X+6)}$ in both cases.  This gives a
$(\log n/n)^{1/(d_X+6)}$ rate of convergence for the second derivative, and
a $(\log n/n)^{3/(d_X+6)}$ rate of convergence for the conditional mean.
However, any choice of $h$ such that both terms go to zero can be used.

In
order to test the conditions of Lemma \ref{min_set_lemma}, we can use the
following procedure.  For some sequence $a_n$ growing to infinity such
that $a_n [(\log n/(nh^{d_X}))^{1/2}\vee h^3]$ converges to zero, let
$\hat{\mathcal{X}}_0^j
=\{x\in D|\hat{\bar m}_j(\theta,x)-(\inf_{x'\in D} \hat{\bar
  m}_j(\theta,x')\wedge 0)|
\le [a_n (\log n/(nh^{d_X}))^{1/2}\vee h^3]\}$.
By Theorem \ref{local_poly_conv}, $\hat{\mathcal{X}}_0^j$ will contain
$\mathcal{X}_0^j$ with probability approaching one.
Thus, if we can determine that $V_j(x)$ is positive definite on
$\hat{\mathcal{X}}_0^j$, then, asymptotically, we will know that $V_j(x)$ is
positive definite on $\mathcal{X}_0^j$.
Note that $\hat{\mathcal{X}}_0^j$ is an estimate of the set of
minimizers of $\overline m_j(x,\theta)$ over $x$ if the moment
inequality binds or fails to hold, and is eventually equal to the
empty set if the moment inequality is slack.

Since the determinant is a differentiable map from $\mathbb{R}^{d_X^2}$ to
$\mathbb{R}$, the
$\mathcal{O}_p((\log n/(nh^{d_X+4}))^{1/2})+\mathcal{O}_p(h)$ rate of
uniform convergence for $\hat V_j(x)$ translates to the same (or faster)
rate of convergence for $\det \hat V_j(x)$.  If, for some $x_0\in
\mathcal{X}_0^j$, $V_j(x_0)$ is not positive definite, then $V_j(x_0)$
will be singular (the second derivative matrix at an interior minimum must
be positive semidefinite if the second derivatives are continuous in a
neighborhood of $x_0$), and $\det V_j(x_0)$ will be zero.  Thus,
$\inf_{x\in \hat{\mathcal{X}}_0^j} \det \hat V_j(x)\le \det \hat V_j(x_0)
=\mathcal{O}_p((\log n/(nh^{d_X+4}))^{1/2})+\mathcal{O}_p(h)$
where the inequality holds with probability approaching one.
Thus, letting $b_n$ be any sequence going to infinity such that $b_n
[(\log n/(nh^{d_X+4}))^{1/2}\vee h]$ converges to zero, if $V_j(x_0)$ is not
positive definite for some $x_0\in\mathcal{X}_0^j$, we will have
$\inf_{x\in \hat{\mathcal{X}}_0^j} \det \hat V_j(x)\le b_n
[(\log n/(nh^{d_X+4}))^{1/2}\vee h]$ with probability approaching one
(actually, since we are only dealing with the point $x_0$, we can use
results for pointwise convergence of the second derivative of the
conditional mean, so the $\log n$ term can be replaced by a constant,
but I use the uniform convergence results for simplicity).

Now, suppose $V_j(x)$ is positive definite for all $x\in\mathcal{X}_0^j$.  By
Lemma \ref{min_set_lemma}, we will have, for some $B>0$, $\det V_j(x)\ge B$
for all $x\in\mathcal{X}_0^j$.  By continuity of $V_j(x)$, we will
also have, for
some $\varepsilon>0$, $\det V_j(x)\ge B/2$ for all $x\in
{\mathcal{X}_0^j}^\varepsilon$ where
${\mathcal{X}_0^j}^\varepsilon=\{x|\inf_{x'\in\mathcal{X}_0^j} \|x-x'\|\le
\varepsilon\}$ is the $\varepsilon$-expansion of $\mathcal{X}_0^j$.
Since $\hat{\mathcal{X}}_0^j \subseteq {\mathcal{X}_0^j}^\varepsilon$ with
probability approaching one, we will also have
$\inf_{x\in \hat{\mathcal{X}}_0^j} \det V_j(x)\ge B/2$
with probability approaching one.  Since $\det \hat V_j(x)\to \det V_j(x)$
uniformly over $D$, we will then have
$\inf_{x\in \hat{\mathcal{X}}_0^j} \det \hat V_j(x)
\ge b_n [(\log n/(nh^{d_X+4}))^{1/2}\vee h]$
with probability approaching one.

This gives the following theorem.
\begin{theorem}\label{pre_test_thm}
Let $\hat V_j(x)$ and
$\hat{\bar m}_j(\theta,x)$
be the local second order polynomial
estimates defined with some kernel $K$ with $h$ such that the rate of
convergence terms in Theorem \ref{local_poly_conv} go to zero.  Let
$\hat{\mathcal{X}}_0^j$ be defined as above with $a_n [(\log
n/(nh^{d_X}))^{1/2}\vee  h^3]$ going to zero and $a_n$ going to infinity, and
let $b_n$ be any sequence going to infinity such that $b_n
[(\log n/(nh^{d_X+4}))^{1/2}\vee h]$ goes to zero.
Suppose that Assumptions
\ref{bdd_y_assump_multi}, \ref{cond_mean_assump}, \ref{x_assump},
\ref{x_given_y_assump}, and \ref{K_assump}, hold, and the null
hypothesis holds with $E(m(W_i,\theta)m(W_i,\theta)'|X_i=x)$
continuous and the data are iid.  Then, if Assumption
\ref{smoothness_assump_multi} holds, we will have
$\inf_{x\in\hat{\mathcal{X}}_0^j} \det \hat V_j(x) > b_n [(\log
  n/(nh^{d_X+4}))^{1/2}\vee h]$ for each $j$
with probability approaching one.  If Assumption
\ref{smoothness_assump_multi} does not hold, we will have
$\inf_{x\in\hat{\mathcal{X}}_0^j} \det \hat V_j(x) \le b_n [(\log
  n/(nh^{d_X+4}))^{1/2}\vee h]$ for some $j$ with probability
approaching one.
\end{theorem}

The purpose of this test of Assumption \ref{smoothness_assump_multi} is as
a preliminary consistent test in a procedure that uses the asymptotic
approximation in Theorem \ref{inf_dist_thm_multi} if the test finds
evidence in favor of Assumption \ref{smoothness_assump_multi}, and uses
the methods that are robust to
different types of contact sets, but possibly conservative, such as those
described in \citet{andrews_inference_2009}, otherwise.  It follows from
Theorem \ref{pre_test_thm} that such a procedure will have the correct size
asymptotically.  In the statement of the following theorem, it is
understood that Assumptions \ref{inv_mat_assump} and
\ref{avar_est_assump}, which refer to objects in Assumption
\ref{smoothness_assump_multi}, do not need to hold if the data generating
process is such that Assumption \ref{smoothness_assump_multi} does not
hold.

\begin{theorem}
Consider the following test.  For some $b_n\to \infty$ and $h\to 0$
satisfying the conditions of Theorem \ref{pre_test_thm}, perform a
pre-test that finds evidence in favor of Assumption
\ref{smoothness_assump_multi}
iff. $\inf_{x\in\hat{\mathcal{X}}_0} \det \hat V_j(x) \ge b_n
[(\log n/(nh^{d_X+4}))^{1/2}\vee h]$ for each $j$.  If
$\hat{\mathcal{X}}_0=\emptyset$, do not
reject the null hypothesis that $\theta\in\Theta_0$.
If $\inf_{x\in\hat{\mathcal{X}}_0} \det \hat
V_j(x) > b_n [(\log n/(nh^{d_X+4}))^{1/2}\vee h]$ for each $j$, reject
the null hypothesis that $\theta\in\Theta_0$ if
$n^{(d_X+2)/(d_X+4)}S(T_n(\theta))>\hat q_{1-\alpha}$ where $\hat
q_{1-\alpha}$ is
an estimate of the $1-\alpha$ quantile of the distribution of $S(Z)$
formed using one of the methods in Section \ref{inference_sec}.
If $\inf_{x\in\hat{\mathcal{X}}_0} \det \hat
V_j(x) \le b_n [(\log n/(nh^{d_X+4}))^{1/2}\vee h]$ for some $j$,
perform any (possibly conservative) asymptotically level $\alpha$ test.
Suppose that Assumptions
\ref{bdd_y_assump_multi}, \ref{S_assump}, \ref{inv_mat_assump},
\ref{abs_cont_S_assump}, \ref{cond_mean_assump}, \ref{x_assump},
\ref{x_given_y_assump}, and \ref{K_assump} hold,
$E(m(W_i,\theta)m(W_i,\theta)'|X_i=x)$ is continuous, and the data are
iid.
Then this provides an asymptotically level $\alpha$ test of
$\theta\in\Theta_0$ if the subsampling procedure is used or if Assumption
\ref{avar_est_assump} holds and the procedure based on estimating the
asymptotic distribution directly is used.  If Assumption
\ref{smoothness_assump_multi} holds, this test is asymptotically exact.
\end{theorem}

The estimates used for this pre-test can also be used to construct
estimates of the quantities in Assumption \ref{avar_est_assump} that
satisfy the consistency requirements of this assumption.  Suppose
that we have estimates $\hat M(x)$, $\hat f_X(x)$, and $\hat V(x)$ of
$E(m(W_i,\theta)m(W_i,\theta)'|X=x)$, $f_X(x)$, and $V(x)$ that are
consistent uniformly over $x$ in a neighborhood of $\mathcal{X}_0$.  Then,
if we have estimates of
$\mathcal{X}_0$ and $J(k)$, we can estimate the quantities in Assumption
\ref{avar_est_assump} using $\hat M_k(x_k)$, $\hat f_X(x_k)$, and $\hat
V(x_k)$ for each $x_k$ in the estimate of $\mathcal{X}_0$, where $\hat
M_k(x_k)$ is a sparse version of $\hat M(x_k)$ with elements with indices
not in the estimate of $J(k)$ set to zero.

The estimate $\hat{\mathcal{X}}_0$ contains infinitely many points, so it
will not work for this purpose.  Instead, define the estimate
$\tilde{\mathcal{X}}_0$ of $\mathcal{X}_0$ and the estimate $\hat J(k)$
of $J(k)$ as follows.  Let $a_n$ be as in Theorem \ref{pre_test_thm},
and let $\varepsilon_n^2\to 0$ more slowly than
$a_n [(\log n/(nh^{d_X}))^{1/2}\vee h^3]$.  Let $\hat \ell_j$ be the smallest
number such that
$\hat{\mathcal{X}}_0^j\subseteq \cup_{k=1}^{\hat
\ell_j}B_{\varepsilon_n}(\hat x_{j,k})$ for some $\hat x_{j,1},\ldots,\hat
x_{j,\hat \ell_j}$.  Define
an equivalence relation $\sim$ on the set $\{(j,k)|1\le j\le d_Y, 1\le
k\le \hat \ell_j\}$ by $(j,k)\sim (j',k')$ iff. there is a sequence
$(j,k)=(j_1,k_1),(j_2,k_2),\ldots,(j_r,k_r)=(j',k')$ such that
$B_{\varepsilon_n}(\hat x_{j_s,k_s})
\cap B_{\varepsilon_n}(\hat x_{j_{s+1},k_{s+1}})\ne \emptyset$ for $s$
from $1$ to $r-1$.  Let $\hat \ell$ be the number of equivalence classes,
and, for each equivalence class, pick exactly one $(j,k)$ in the
equivalence class and let $\tilde x_r=\hat x_{j,k}$ for some $r$ between $1$
and $\hat \ell$. Define the estimate of the set $\mathcal{X}_0$ to be
$\tilde{\mathcal{X}}_0\equiv \{\tilde x_1,\ldots,\tilde x_{\hat \ell}\}$, and
define the estimate $\hat J(r)$ for $r$ from $1$ to $\hat \ell$ to be the
set of indices $j$ for which some $(j,k)$ is in the same equivalence class
as $\tilde x_r$.

Although these estimates of $\mathcal{X}_0$, $\ell$, and
$J(1),\ldots,J(\ell)$ require some
cumbersome notation to define, the intuition behind them is simple.
Starting with the initial estimates $\hat{\mathcal{X}}_j$, turn these sets
into discrete sets of points by taking the centers of balls that contain
the sets $\hat{\mathcal{X}}_j$ and converge at a slower rate.  This gives
estimates of the points at which the conditional moment inequality indexed
by $j$ binds for each $j$, but to estimate the asymptotic distribution in
Theorem \ref{inf_dist_thm_multi}, we also need to determine which
components, if any, of $\bar m(\theta,x)$ bind at the same value of $x$.
The procedure described above does this by testing whether the balls used
to form the estimated contact points for each index of $\bar m(\theta,x)$
intersect across indices.

The following theorem shows that this is a consistent
estimate of the set $\mathcal{X}_0$ and the indices of the binding
moments.

\begin{theorem}\label{X0_est_thm}
Suppose that Assumptions
\ref{smoothness_assump_multi}, \ref{cond_mean_assump}, \ref{x_assump},
\ref{x_given_y_assump}, and \ref{K_assump} hold.
For the estimates $\tilde{\mathcal{X}}_0$, $\hat \ell$ and $\hat J(r)$,
$\hat \ell=\ell$ with probability approaching one and, for some labeling
of the indices of $\tilde x_1,\ldots,\tilde x_{\hat \ell}$ we have, for
$k$ from $1$ to $\ell$, $\tilde x_k\stackrel{p}{\to} x_k$ and, with
probability approaching one, $\hat J(k)=J(k)$.
\end{theorem}

An immediate consequence of this is that this estimate of $\mathcal{X}_0$
can be used in combination with consistent estimates of
$E(m(W_i,\theta)m(W_i,\theta)'|X=x)$, $f_X(x)$, and $V(x)$
to form estimates of these functions evaluated at points in
$\mathcal{X}_0$ that satisfy the assumptions needed for the procedure for
estimating the asymptotic distribution described in Section
\ref{inference_sec}.

\begin{corollary}
If the estimates $\hat M_k(x)$, $\hat f_X(x)$, and $\hat V(x)$ are
consistent uniformly over $x$ in a neighborhood of $\mathcal{X}_0$, then,
under Assumptions \ref{smoothness_assump_multi}, \ref{cond_mean_assump},
\ref{x_assump}, \ref{x_given_y_assump}, and \ref{K_assump},
the estimates $\hat M_k(\tilde x_k)$, $\hat f_X(\tilde x_k)$, and $\hat
V_j(\tilde x_k)$ satisfy Assumption \ref{avar_est_assump}.
\end{corollary}

\section{Local Alternatives}\label{local_alt_sec}

Consider local alternatives of the form $\theta_n=\theta_0+a_n$ for some
fixed $\theta_0$ such that $m(W_i,\theta_0)$ satisfies Assumption
\ref{smoothness_assump_multi} and $a_n\to 0$.  Here, I keep the data
generating process fixed and vary the parameter being tested.  Similar
ideas will apply when the parameter is fixed and the data generating
process is changed so that the parameter approaches the identified set.
Throughout this section, I restrict attention to the conditions in Section
\ref{inf_dist_sec}, which corresponds to the more general setup in Section
\ref{inf_dist_alpha_sec} with $\gamma=2$.
To translate the $a_n$ rate of convergence to $\theta_0$ to a rate of
convergence for the sequence of conditional means, I make the following
assumptions.  As before, define
$\bar m(\theta,x)=E(m(W_i,\theta)|X_i=x)$.

\begin{assumption}\label{diff_m_assump}
For each $x_k\in \mathcal{X}_0$,
$\bar m(\theta,x)$ has a derivative as a function of
$\theta$ in a neighborhood of $(\theta_0,x_k)$, denoted $\bar
  m_\theta(\theta,x)$, that is continuous as a
function of $(\theta,x)$ at $(\theta_0,x_k)$
and,
for any neighborhood of $x_k$, there is a neighborhood of $\theta_0$ such
that $\bar m_j(\theta,x)$ is bounded away from zero for $\theta$ in the
given neighborhood of $\theta_0$ and $x$ outside of the given neighborhood
of $x_k$ for $j\in J(k)$ and for all $x$ for $j\notin J(k)$.
\end{assumption}

\begin{assumption}\label{m2_assump}
For each $x_k\in\mathcal{X}_0$ and $j\in J(k)$,
$E\{[m_j(W_i,\theta)-m_j(W_i,\theta_0)]^2|X_i=x\}$ converges to zero
uniformly in $x$ in some neighborhood of $x_k$ as $\theta\to \theta_0$.
\end{assumption}

I also make the following assumption, which extends Assumption
\ref{bdd_y_assump_multi} to a neighborhood of $\theta_0$.

\begin{assumption}\label{bdd_y_assump_local}
For some fixed $\overline Y<\infty$ and $\theta$ in a some neighborhood of
$\theta_0$, $|m(W_i,\theta)|\le \overline Y$ with probability one.
\end{assumption}

In the interval regression example, these conditions are satisfied as long
as Assumption \ref{smoothness_assump_multi} holds at $\theta_0$ and the
data have finite support.  These conditions are also likely to hold in a
variety of models once Assumption \ref{smoothness_assump_multi} holds at
$\theta_0$.  Note that smoothness conditions are in terms of the
conditional mean $\bar m(\theta,x)$, rather than $m(W_i,\theta)$, so that
the conditions can still hold when the sample moments are nonsmooth
functions of $\theta$.

Set $a_n=b_na$ for some sequence of scalars $b_n\to 0$ and a constant
vector $a$.
Going through the argument for Theorem \ref{inf_dist_thm_multi},
the variance term in the local process is now
\begin{align*}
&\frac{\sqrt{n}}{\sqrt{h_n^{d_X}}}(E_n-E)m(W_i,\theta_0+b_na)I(h_ns<X-x_k<h_n(s+t))
\\
&=\frac{\sqrt{n}}{\sqrt{h_n^{d_X}}}(E_n-E)m(W_i,\theta_0)I(h_ns<X-x_k<h_n(s+t))
\\
&+\frac{\sqrt{n}}{\sqrt{h_n^{d_X}}}(E_n-E)
[m(W_i,\theta_0+b_na)-m(W_i,\theta_0)]I(h_ns<X-x_k<h_n(s+t)).
\end{align*}
The first term is the variance term under the null, and the second term
should be small under Assumption \ref{m2_assump}.

As for the drift term,
\begin{align*}
&\frac{1}{h_n^{d_X+2}}Em(W_i,\theta+b_na)I(h_ns<X_i-x_k<h_n(s+t))  \\
&=\frac{1}{h_n^{d_X+2}}Em(W_i,\theta)I(h_ns<X_i-x_k<h_n(s+t))  \\
&+\frac{1}{h_n^{d_X+2}}E[m(W_i,\theta+b_na)-m(W_i,\theta)]
I(h_ns<X_i-x_k<h_n(s+t)).
\end{align*}
The first term is the drift term under the null.  The second term is
\begin{align*}
&\frac{1}{h_n^{d_X+2}}E[\bar m(\theta+b_na,X_i)-\bar m(\theta,X_i)]
I(h_ns<X_i-x_k<h_n(s+t))  \\
&\approx \frac{1}{h_n^{d_X+2}}Eb_n\bar m_\theta(\theta,X_i)a
I(h_ns<X_i-x_k<h_n(s+t))  \\
&\approx \frac{b_n}{h_n^{d_X+2}}f_X(x_k)\bar m_\theta(\theta,x_k)a
\int_{h_ns<x-x_k<h_n(s+t)} \, dx
= \frac{b_n}{h_n^2}f_X(x_k)\bar m_\theta(\theta,x_k)a
\prod_i t_i.
\end{align*}

Setting $b_n=h_n^2=n^{-2/(d_X+4)}$ gives a constant that does not change
with $n$, so we should expect to have power against $n^{-2/(d_X+4)}$
alternatives.  The following theorem formalizes these ideas.

\begin{theorem}\label{local_alt_exact_thm}
Let $\theta_0$ be such that $E(m(W_i,\theta_0)|X_i)\ge 0$ almost surely and
Assumptions \ref{smoothness_assump_multi}, \ref{diff_m_assump},
\ref{m2_assump}, and \ref{bdd_y_assump_local}
are satisfied for $\theta_0$.  Let
$a\in \mathbb{R}^{d_\theta}$ and let $a_n=an^{-2/(d_X+4)}$.  Let $Z(a)$ be a
random variable defined the same way as $Z$ in Theorem
\ref{inf_dist_thm_multi}, but with the functions $g_{P,x_k,j}(s,t)$
replaced by the functions
\begin{align*}
g_{P,x_k,j,a}(s,t)
=\frac{1}{2}f_X(x_k)\int_{s<x<s+t}x'V_j(x_k)x\, dx
+\bar m_{\theta,j}(\theta_0,x_k) a f_X(x_k)\prod_i t_i
\end{align*}
for $j\in J(k)$ for each $k$ where $\bar m_{\theta,j}$ is the $j$th row of
the derivative matrix $\bar m_{\theta}$.  Then
\begin{align*}
n^{(d_X+2)/(d_X+4)}\inf_{s,t}E_nm(W_i,\theta+a_n)I(s<X_i<s+t)
\stackrel{d}{\to} Z(a).
\end{align*}
\end{theorem}

Thus, an exact test gives power against $n^{-2/(d_X+4)}$ alternatives (as
long as $\bar m_{\theta,j}(\theta_0,x_k)a$ is negative for each $j$ or
negative enough for at least one $j$),
but not against alternatives that converge strictly faster.  The
dependence on the dimension of $X_i$ is a result of the curse of
dimensionality.  With a fixed amount of ``smoothness,'' the speed at which
local alternatives can converge to the null space and still be detected is
decreasing in the dimension of $X_i$.

Now consider power against local alternatives of this form, with a
possibly different sequence $a_n$, using the conservative estimate that
$\sqrt{n}\inf_{s,t}E_nm(W_i,\theta)I(s<X_i<s+t)\stackrel{p}{\to} 0$ for
$\theta\in\Theta_0$.  That is, we
fix some $\eta>0$ and reject if
$\sqrt{n}S(\inf_{s,t}E_nm(W_i,\theta_0+a_n)I(s<X_i<s+t))>\eta$.
For the drift term $Em_j(W_i,\theta_0+a_n)I(s<X_i-x_k<s+t)$ of the
local alternative, we have, for $t$ near zero and $s$ near any
$x_k\in\mathcal{X}_0$,
\begin{align*}
&\sqrt{n}Em_j(W_i,\theta_0+a_n)I(s<X_i-x_k<s+t)  \\
&\approx \sqrt{n}
E[\overline m_j(\theta_0,X_i)+\overline m_{\theta,j}(\theta_0,X_i)a_n]
I(s<X_i-x_k<s+t)  \\
&\approx \sqrt{n}f_X(x_k)\int_{s<x<s+t}
\left(\frac{1}{2}x'V_j(x_k)x
+\overline m_{\theta,j}(\theta_0,x_k)a_n\right) \, dx.
\end{align*}
For any $a$ and $b$,
\begin{align*}
&f_X(x_k)\int_{s<x<s+t} \left(\frac{1}{2}x'Vx+a\right) \, dx
=(a/b)f_X(x_k)\int_{s<x<s+t}
\left\{\frac{1}{2}[(b/a)^{1/2}x]'V[(b/a)^{1/2}x]+b\right\} \, dx  \\
&= (a/b)f_X(x_k)\int_{(b/a)^{1/2}s<x<(b/a)^{1/2}(s+t)}
\left(\frac{1}{2}u'Vu+b\right) \, (b/a)^{-d_X/2}du.
\end{align*}
For any $(s,t)$, the last line in the display is equal to
$(a/b)^{(d_X+2)/2}$ times the first expression in the display evaluated
at a different value of $(s,t)$ with $a$ replaced with $b$.
It follows that the minimized expression for $b$ is $(a/b)^{(d_X+2)/2}$
times the minimized expression for $a$.  Thus, if $a_n=a b_n$, the
drift term is of order $\sqrt{n}b_n^{(d_X+2)/2}$, so we should expect to
have power against local alternatives with
$\sqrt{n}b_n^{(d_X+2)/2}=\mathcal{O}(1)$ or $b_n=n^{-1/(d_X+2)}$ (note
that setting $n^{(d_X+2)/(d_X+4)}b_n^{(d_X+2)/2}=\mathcal{O}(1)$ so that
the drift term is of the same order of magnitude as the exact rate of
convergence gives the $n^{-2/(d_X+4)}$ rate derived in the previous
theorem for the exact test).  Since
the infimum of the drift term is taken at a point where $t$ is small,
we should expect the mean zero term to converge at a faster than
$\sqrt{n}$ rate, so that the limiting distribution will be degenerate.
This is formalized in the following theorem.

\begin{theorem}\label{local_alt_degen_thm}
Let $\theta_0$ be such that $E(m(W_i,\theta_0)|X_i)\ge 0$ almost surely and
Assumptions \ref{smoothness_assump_multi}, \ref{diff_m_assump},
\ref{m2_assump}, and \ref{bdd_y_assump_local}
are satisfied for $\theta_0$.  Let
$a\in \mathbb{R}^{d_\theta}$ and let $a_n=an^{-1/(d_X+2)}$.
Then, for each $j$,
\begin{align*}
&\sqrt{n} \inf_{s,t} E_nm_j(W_i,\theta_0+a_n)I(s<X<s+t)  \\
&\stackrel{p}{\to} \min_{k \text{ s.t. } j\in J(k)}\inf_{s,t}
f_X(x_k)\int_{s<x<s+t}
\left(\frac{1}{2}x'Vx
+\overline m_{\theta,j}(\theta_0,x_k)a\right) \, dx.
\end{align*}
\end{theorem}

The $n^{-1/(d_X+2)}$ rate is slower than the $n^{-2/(d_X+4)}$ rate for
detecting local alternatives with the asymptotically exact test.
As with the asymptotically exact tests, the conservative tests do
worse against this form of local alternative as
the dimension of the conditioning variable $X_i$ increases.

\section{Monte Carlo}\label{mc_sec}

I perform a monte carlo study to examine the finite sample behavior of the
tests I propose, and to see how well the asymptotic results in this paper
describe the finite sample behavior of KS statistics.  First, I simulate
the distribution of KS statistics for various sample sizes under parameter
values and data generating processes that satisfy Assumption
\ref{smoothness_assump_multi},
and for data generating processes that lead to a $\sqrt{n}$ rate of
convergence.  As predicted by Theorem \ref{inf_dist_thm_multi}, for the
data generating process that satisfies Assumption
\ref{smoothness_assump_multi},
the distribution of the KS statistic is roughly stable across sample sizes
when scaled up by $n^{(d_X+2)/(d_X+4)}$.  For the data generating process
that leads to $\sqrt{n}$ convergence, scaling by $\sqrt{n}$ gives a
distribution that is stable across sample sizes.  Next,
I examine the size and power of KS statistic based tests using the
asymptotic distributions derived in this paper.  I include procedures that
test between the conditions leading to $\sqrt{n}$ convergence and the
faster rates derived in this paper using the subsampling estimates of the
rate of convergence %
described in Section \ref{subsamp_rate_subsec}, as well
as infeasible procedures that use prior knowledge of the correct rate of
convergence to estimate the asymptotic distribution.

\subsection{Monte Carlo Designs}

Throughout this section, I consider two monte carlo designs for a mean
regression model with missing data.  In this model, the latent variable
$W_i^*$ satisfies $E(W_i^*|X_i)=\theta_1+\theta_2 X_i$, but $W_i^*$ is
unobserved, and can only be bounded by the observed variables
$W_i^H=\overline w I(\text{$W_i^*$ missing})
  +W_i^*I(\text{$W_i^*$ observed})$
and
$W_i^L=\underline w I(\text{$W_i^*$ missing})
  +W_i^*I(\text{$W_i^*$ observed})$ are observed, where
$[\underline w,\overline w]$ is an interval known to contain $W_i^*$.  The
identified set $\Theta_0$ is the set of values of $(\theta_1,\theta_2)$
such that the moment inequalities
$E(W_i^H-\theta_1-\theta_2 X_i|X_i)\ge 0$ and
$E(\theta_1+\theta_2 X_i-W_i^L|X_i)\ge 0$ hold with probability one.
For both designs, I draw $X_i$ from a uniform distribution on $(-1,1)$
(here, $d_X=1$).  Conditional on $X_i$, I draw $U_i$ from an independent
uniform $(-1,1)$ distribution, and set
$W_i^*=\theta_{1,*}+\theta_{2,*} X_i+U_i$, where $\theta_{1,*}=0$ and
$\theta_{2,*}=.1$.  I then set $W_i^*$ to be missing with probability
$p^*(X_i)$ for some function $p^*$ that differs across designs.  I set
$[\underline w,\overline w]=[-.1-1,.1+1]=[-1.1,1.1]$, the unconditional
support of $W_i^*$.  Note that, while the data are generated using a
particular value of $\theta$ in the identified set and a censoring process
that satisfies the missing at random assumption (that the probability of
data missing conditional on $(X_i,W_i^*)$ does not depend on $W_i^*$), the
data generating process is consistent with forms of endogenous censoring
that do not satisfy this assumption.  The identified set contains all
values of $\theta$ for which the data generating process is consistent
with the latent variable model for $\theta$ and some, possibly endogenous,
censoring mechanism.

The shape of the conditional moment inequalities as a function of $X_i$
depends on $p^*$.  For Design 1, I set
$p^*(x)=(0.9481 x^4 + 1.0667 x^3 -0.6222 x^2 -0.6519 x + 0.3889)\wedge 1$.
The coefficients of this quartic polynomial were chosen to make $p^*(x)$
smooth, but somewhat wiggly, so that the quadratic approximation to the
resulting conditional moments used in Theorem \ref{inf_dist_thm_multi}
will not be good over the entire support of $X_i$.  The resulting
conditional means of the bounds on $W_i^*$ are
$E(W_i^L|X_i=x)=(1-p^*(x))(\theta_{1,*}+\theta_{2,*}x)+p^*(x)\underline w$
and
$E(W_i^H|X_i=x)=(1-p^*(x))(\theta_{1,*}+\theta_{2,*}x)+p^*(x)\overline w$.
In the monte carlo study, I examine the distribution of the KS statistic
for the upper inequality at
$(\theta_{1,D1},\theta_{2,D1})\equiv(1.05,.1)$, a parameter
value on the boundary of the identified set for which Assumption
\ref{smoothness_assump_multi} holds, along with confidence intervals
for the intercept parameter $\theta_1$ with the slope parameter $\theta_2$
fixed at $.1$.  For the confidence regions, I also restrict attention to
the moment inequality corresponding to $W_i^H$, so that the confidence
regions are for the one sided model with only this conditional moment
inequality.  Figure \ref{cond_means_d1_fig} plots the conditional means of
$W_i^H$ and $W_i^L$, along with the regression line corresponding to
$\theta=(1.05,.1)$.  The confidence intervals for the slope parameter
invert a family of tests corresponding to values of $\theta$ that move
this regression line vertically.

For Design 2, I set $p^*(x)=[(|x-.5|\vee .25)-.15]\wedge .7)$.  Figure
\ref{cond_means_d3_fig} plots the resulting conditional means.  For this
design, I examine the distribution of the KS statistic for the upper
inequality at $(\theta_{1,D2},\theta_{2,D2})=(1.1,.9)$, which leads to a
positive probability
contact set for the upper moment inequality and a $n^{1/2}$ rate of
convergence to a nondegenerate distribution.  The regression line
corresponding to this parameter is plotted in Figure
\ref{cond_means_d3_fig} as well.  For this design, I form confidence
intervals for the slope parameter $\theta_1$ with $\theta_2$ fixed at
$.9$, using the KS statistic for the moment inequality for $W_i^H$.

The confidence intervals reported in this section are computed by
inverting KS statistic based tests on a grid of parameter values.
I use a grid with meshwidth $.01$ that covers the area of the parameter
space with distance to the boundary of the identified set no more than
$1$.
In practice, monotonicity of the KS statistic in certain parameters (in
this case, the KS statistic for each moment inequality is monotonic in the
intercept parameter) can often be used to get a rough estimate of the
boundary of the identified set before mapping out the confidence region
exactly.  In this
case, a rough estimate of the boundary of the identified set for the
intercept parameter could be formed by finding the point where the KS
statistic for the moment inequality for $W_i^H$ crosses a fixed critical
value before performing the test with critical values estimated for each
value of $\theta$. All of the results in this section use 1000 monte carlo
draws for each sample size and monte carlo design.

\subsection{Distribution of the KS Statistic}

To examine how well Theorem \ref{inf_dist_thm_multi}
describes the finite sample distribution of KS statistics under Assumption
\ref{smoothness_assump_multi},
I simulate from Design 1 for a range of sample sizes and form
the KS statistic for testing $(\theta_{1,D1},\theta_{2,D1})$.  Since
Assumption \ref{smoothness_assump_multi} holds for testing this value of
$\theta$ under this data generating process, Theorem
\ref{inf_dist_thm_multi} predicts that the distribution of the KS
statistic scaled up by $n^{(d_X+2)/(d_X+4)}=n^{3/5}$ should be similar
across the sample sizes.  The performance of this asymptotic prediction in
finite samples is examined in Figure
\ref{ks_stat_hist_d1_fig}, which plots histograms of the scaled KS
statistic $n^{3/5}S(T_n(\theta))$ for the sample sizes
$n\in\{100,500,1000,2000,5000\}$.  The scaled distributions appear roughly
stable across sample sizes, as predicted.

In contrast, under Design 2, the KS statistic
for testing $(\theta_{1,D2},\theta_{2,D2})$
will converge at a $n^{1/2}$ rate to a nondegenerate distribution.  Thus,
asymptotic approximation suggests that, in this case, scaling by $n^{1/2}$
will give a distribution that is roughly stable across sample sizes.
Figure \ref{ks_stat_hist_d3_fig} plots histograms of the scaled statistic
$n^{1/2}S(T_n(\theta))$ for this case.  The scaling suggested by
asymptotic approximations appears to give a distribution that is stable
across sample sizes here as well.

\subsection{Finite Sample Performance of the Tests}

I now turn to the finite sample performance of confidence regions for the
identified set based on critical values formed using the asymptotic
approximations derived in this paper, along with possibly conservative
confidence regions that use the $n^{1/2}$ approximation.  The critical
values use subsampling with different assumed rates of convergence.  I
report results for the tests based on subsampling estimates of the rate of
convergence described in Section \ref{subsamp_rate_subsec}, tests that use
the conservative rate $n^{1/2}$, and infeasible tests that use a $n^{3/5}$
rate under Design 1, and a $n^{1/2}$ rate under Design 2.
The implementation details are as follows.
For the critical values using the conservative
rate of convergence, I estimate the $.9$ and $.95$ quantiles
of the distribution of the KS statistic at each value of $\theta$ using
subsampling, and add the correction factor $.001$ to prevent the critical
value from going to zero.  The critical values using estimated rates of
convergence are computed as described in Section \ref{subsamp_rate_subsec}.
I use the subsample sizes $b_1=\lceil n^{1/2} \rceil$ and $b_2=\lceil
n^{1/3} \rceil$ to estimate the rate of convergence $\hat \beta$ for
subsampling, and $b_2=5$ for the rate estimate $\hat
\beta_a$ that is used to test whether the conservative rate should be
used.
For both rate estimates, I average the estimates computed using the
quantiles $.5$, $.9$, and $.95$.
For the upper and lower truncation points for the rate of convergence, I
use $\underline\beta=.55$ and $\overline\beta=2/3$.  These truncation
points allow for exact inference for values of $\theta$ such that
Assumption \ref{smoothness_assump_alpha} holds with $\gamma=2$ (twice
differentiable conditional mean) or $\gamma=1$ (directional derivatives
from both sides).  The upper truncation point $\overline \beta$
corresponds to $\gamma=1$, and the lower truncation point $\underline
\beta$ is halfway between the rate of convergence exponent $3/5$ for
$\gamma=2$, and the conservative rate exponent $1/2$.
In addition, I truncate $\hat \beta$ from below at $1/2$ in cases where
$\hat\beta<1/2$.  For both the
conservative and estimated rates of convergence, I use the uncentered
subsampling estimate with subsample size $\lceil n^{1/2}\rceil$.  All
subsampling estimates use 1000 subsample draws.  For
values of $\theta$ such that the pre-test finds that the conservative
approximation should be used ($\hat\beta_a<\underline \beta$), I use the
same method of estimating the critical values as in the tests that always
use the conservative rate of convergence.

Table \ref{cov_prob_table_d1} reports the coverage probabilities for
$(\theta_{1,D1},\theta_{2,D1})$ under Design 1.  As discussed above, under
Design 1, $(\theta_{1,D1},\theta_{2,D1})$ is on the boundary of the
identified set and satisfies Assumption \ref{smoothness_assump_multi}.
As predicted, the tests that subsample with the $n^{1/2}$ rate are
conservative.  The nominal 95\% confidence regions that use the $n^{1/2}$
rate cover $(\theta_{1,D1},\theta_{2,D1})$ with probability at least $.99$
for all of the sample sizes.  Subsampling with the exact $n^{3/5}$ rate of
convergence, an infeasible procedure that uses prior knowledge that
Assumption \ref{smoothness_assump_multi} holds under
$(\theta_{1,D1},\theta_{2,D1})$ for this data generating process, gives
confidence regions that cover $(\theta_{1,D1},\theta_{2,D1})$ with
probability much closer to the nominal coverage.  The subsampling tests
with the estimated rate of convergence also perform well, attaining close
to the nominal coverage.

Table \ref{cov_prob_table_d3} reports coverage probabilities for testing
$(\theta_{1,D2},\theta_{2,D2})$ under Design 2.  In this case, subsampling
with a $n^{1/2}$ rate gives an asymptotically exact test of
$(\theta_{1,D2},\theta_{2,D2})$, so we should expect the coverage
probabilities for the tests that use the $n^{1/2}$ rate of convergence to
be close to the nominal coverage probabilities, rather than being
conservative.  The coverage probabilities for the $n^{1/2}$ rate are
generally less conservative here than for Design 1, as the asymptotic
approximations predict, although the coverage is considerably greater than
the nominal coverage, even with $5000$ observations.  In this case, the
infeasible procedure is identical to the conservative test, since the
exact rate of convergence is $n^{1/2}$.  The confidence regions that use
subsampling with the estimated rate contain
$(\theta_{1,D2},\theta_{2,D2})$ with probability close to the nominal
coverage, but are generally more liberal than their nominal level.

Given that
subsampling with the estimated rate increases type I error by having
coverage probability close to the nominal coverage probability rather than
being conservative, we should expect a decrease in type II error.  The
results in Section \ref{local_alt_sec} show that critical values based on
the exact $n^{3/5}$ rate of convergence lead to tests that detect local
alternatives that approach the identified set at a $n^{2/(d_X+4)}=n^{2/5}$
rate, while the conservative tests detect local alternatives that approach
the identified set at a slower $n^{1/(d_X+2)}=n^{1/3}$ rate.  For
confidence regions that invert these tests, this is reflected in the
portion of the parameter space the confidence region covers outside of the
true identified set.

Tables \ref{ci_length_table_d1} and
\ref{ci_length_table_d3} summarize the portion of the parameter space
outside of the identified set covered by
confidence intervals for the intercept
parameter $\theta_1$ with $\theta_2$ fixed at $\theta_{2,D1}$ for Design 1
and $\theta_{2,D2}$ for Design 2.  The entries in each table report
the upper endpoint of one of the confidence regions minus
the upper endpoint of the identified set for the slope parameter, averaged
over the monte
carlo draws.  As discussed above, the true upper endpoint of the
identified set for $\theta_1$ under Design 1 with $\theta_2$ fixed at
$\theta_{2,D1}$ is $\theta_{1,D1}$, and the true upper endpoint of the
identified set for $\theta_1$ under Design 2 with $\theta_2$ fixed at
$\theta_{2,D2}$ is $\theta_{1,D2}$, so, letting $\hat u_{1-\alpha}$ be
the greatest value of $\theta_1$ such that $(\theta_1,\theta_{2,D1})$ is
not rejected, Table \ref{ci_length_table_d1} reports averages of
$\hat u_{1-\alpha}-\theta_{2,D1}$, and similarly for Table
\ref{ci_length_table_d3} and Design 2.

The results of Section \ref{local_alt_sec} suggest
that, for the results for Design 1 reported in Table
\ref{ci_length_table_d1}, the difference between the upper endpoint of the
confidence region and the upper endpoint of the identified set should
decrease at a $n^{2/5}$ rate for the critical values that use or estimate
the exact rate of convergence (the first and third rows), and a
$n^{1/3}$ rate for subsampling with the conservative rate and adding
$.001$ to the critical value (the second row).  This appears roughly
consistent with the values reported in these tables.  The conservative
confidence regions start out slightly larger, and then converge more
slowly.
For Design 2, the KS statistic converges at a $n^{1/2}$ rate on the
boundary of the identified set for $\theta_1$ for $\theta_2$ fixed at
$\theta_{2,D2}$, and arguments in \citet{andrews_inference_2009} show that
$n^{1/2}$ approximation to the KS statistic give power against sequences
of alternatives that approach the identified set at a $n^{1/2}$ rate.  The
confidence regions do appear to shrink to the identified set at
approximately this rate over most sample sizes, although the decrease in
the width of the confidence region is larger than predicted for smaller
sample sizes, perhaps reflecting time taken by the subsampling procedures
to find the binding moments.

\section{Illustrative Empirical Application}\label{application_sec}

As an illustrative empirical application, I apply the methods in this
paper to regressions of out of pocket prescription drug spending on income
using data from the Health and Retirement Study (HRS).  In this survey,
respondents who did not report point values for these and other variables
were asked whether the variables were within a series of brackets, giving
point values for some observations and intervals of different sizes for
others.  The income variable used here is taken from the RAND contribution
to the HRS,
which adds up reported income from different sources elicited in the
original survey.
For illustrative purposes, I focus on the subset of respondents
who report point values for income, so that only prescription drug
spending, the dependent variable, is interval valued.  The resulting
confidence regions are valid under any potentially endogenous process
governing the size of the reported interval for prescription expenditures,
but require that income be missing or interval reported at random.
Methods similar to those proposed in this paper could also be used along
with the results of \citet{manski_inference_2002} for interval reported
covariates to use these additional observations to potentially gain
identifying power (but still using an assumption of exogenous interval
reporting for income).
I use the 1996 wave of the survey and restrict attention to women with no
more than \$15,000 of yearly income who report using prescription
medications.  This results in a data set with 636 observations.  Of these
observations, 54 have prescription expenditures reported as an interval
of nonzero width with finite endpoints, and an additional 7 have no
information on prescription expenditures.

To describe the setup formally, let $X_i$ and $W_i^*$ be income and
prescription drug expenditures for the $i$th observation.  We observe
$(X_i,W_i^L,W_i^H)$, where $[W_i^L,W_i^H]$ is an interval that contains
$W_i^*$.  For observations where no interval is reported for prescription
drug spending, I set $W_i^L=0$ and $W_i^H=\infty$.  I estimate an interval
median regression model where the median $q_{1/2}(W_i^*|X_i)$ of $W_i^*$
given $X_i$ is assumed to follow a linear regression model
$q_{1/2}(W_i^*|X_i)=\theta_1+\theta_2 X_i$.  This leads to the conditional
moment inequalities $E(m(W_i,\theta)|X_i)\ge 0$ almost surely, where
$m(W_i,\theta)=(I(\theta_1+\theta_2 X_i\le W_i^H)-1/2,
  1/2-I(\theta_1+\theta_2 X_i\le W_i^L))$
and $W_i=(X_i,W_i^L,W_i^H)$.

Figure \ref{data_plot_fig} shows the data graphically.  The horizontal
axis measures income, while the vertical axis measures out of pocket
prescription drug expenditures.
Observations for which prescription expenditures are
reported as a point value are plotted as points.  For observations
where a nontrivial interval is reported, a plus symbol marks the upper
endpoint, and an x marks the lower endpoint.  For
observations where no information on prescription expenditures is obtained
in the survey, a circle is placed on the $x$ axis at the value of income
reported for that observation.
In order to show in detail the ranges of
spending that contain most of the observations, the vertical axis is
truncated at \$15,000, leading to $5$ observations not being shown
(although these observations are used in forming the confidence regions
reported below).

I form 95\% confidence intervals by inverting level $.05$ tests using the
KS statistics described in this paper with critical values calculated
using the conservative rate of convergence $n^{1/2}$, and rates of
convergence estimated using the methods described in Section
\ref{subsamp_rate_subsec}.  For the function $S$, I set $S(t)=\max_k
|t_k\wedge 0|$.  The rest of the implementation details are the same as
for the monte carlos in Section \ref{mc_sec}.

For comparison, I also compute point estimates and confidence regions
using the least absolute deviations (LAD) estimator
\citep{koenker_regression_1978} for the median regression model with only
the observations for which a point value for spending was reported.
These are valid under the additional assumption that the decision to
report an interval or missing value is independent of spending conditional
on income.  The confidence regions use Wald tests based on the asymptotic
variance estimates computed by Stata.  These asymptotic variance estimates
are based on formulas in \citet{koenker_robust_1982} and require
additional assumptions on the data generating process, but I use these
rather than more robust standard errors in order to provide a comparison
to an alternative procedure using default options in a standard
statistical package.

Figure \ref{cr95_est_fig} plots the outline of the 95\% confidence region
for $\theta$ using the pre-tests and rate of convergence estimates
described above, while Figure \ref{cr95_cons_fig} plots
the outline of the 95\% confidence region using the conservative
approximation.  Figure \ref{cr95_wald_fig} plots the outline of the 95\%
confidence region from estimating a median regression model on the subset
of the data with point values reported for spending.
Table \ref{ci_table} reports the
corresponding confidence intervals for the components of $\theta$.  For
the confidence regions based on KS tests, I use
the projections of the confidence region for $\theta$ onto each
component.  For the confidence regions based on median regression with
point observations, the 95\% confidence regions use the limiting normal
approximation for each component of $\theta$ separately.

The results show a sizeable increase in statistical power from
using the estimated rates of convergence.  With the conservative tests,
the 95\% confidence region estimates that a \$1,000 increase in income is
associated with at least a \$3 increase in out of pocket prescription
spending at the median.  With the tests that use the estimated rates of
convergence, the 95\% confidence region bounds the increase in out of
pocket prescription spending associated with a \$1,000 increase in income
from below by \$11.30.

The 95\% confidence region based on median regression using observations
reported as points overlaps with both moment inequality based confidence
regions, but gives a different picture of which parameter values can be
ruled out by the data.  The upper bound for the increase in spending
associated with a \$1,000 increase in income is \$24.40 using LAD,
compared to \$37.20 and \$34.70 using KS statistics with all observations
and the conservative and estimated rates respectively.  The corresponding
lower bound is \$10 using LAD with point observations, substantially
larger than the lower bound of \$3 using the conservative procedure, but
actually smaller than the \$11.30 lower bound under the estimated rate.
Thus, while the interval reporting at random assumption for the dependent
variable allows one to tighten the upper bound for the slope parameter, a
lower bound close to the lower bound of the LAD confidence interval can be
obtained using the new asymptotic approximations developed in this paper.

Note also that these tests could, but do not, provide evidence against the
assumptions required for LAD on the point reported values.  If the
LAD 95\% confidence region did not overlap with one of the moment
inequality 95\% confidence regions, there would be no parameter value
consistent with this assumption at the $.1$ level (for any parameter
value, we can reject the joint null of both models holding using
Bonferroni's inequality and the results of the $.05$ level tests).  This
type of test will not necessarily have power if the interval reporting at
random assumption for the dependent variable does not hold, so it should
not be taken as evidence that the more robust interval regression
assumptions can be replaced with LAD methods.

\section{Discussion}\label{discussion_sec}

Under some smoothness conditions, the asymptotic approximations
derived in Sections \ref{inf_dist_sec} and \ref{inf_dist_alpha_sec} can be
combined with the methods
in Sections \ref{inference_sec} and \ref{rate_test_sec} to form tests
that are asymptotically exact on portions of the boundary of the
identified set where the $\sqrt{n}$ approximation only allows for
conservative inference.  Since these methods require assumptions on
the conditional mean that are not needed for conservative inference
using the $\sqrt{n}$ approximation, the decision of which method to
use involves a tradeoff between power and robustness.  The results in
Section \ref{local_alt_sec} quantify these tradeoffs.  While
approximations to the distribution of a KS statistic based on the
asymptotic distribution in Section \ref{inf_dist_sec} and the tests in
Sections \ref{inference_sec} and \ref{rate_test_sec} may not be robust to
certain types of
nonsmooth conditional means, when they are valid, they can detect
parameters in a $n^{-2/(d_X+4)}$ region of the identified set,
while the $\sqrt{n}$ approximation can only detect parameters
in a $n^{-1/(d_X+2)}$ region of the identified set.  It should be noted
that, even if the pre-tests in Section \ref{rate_test_sec} find a rate of
convergence that is too fast, Lemma \ref{large_s_lemma2_multi} in the
Appendix shows that the rate of convergence will typically be within $\log
n$ of $1/n$ for testing $\theta$ on the interior of the identified set, so
the resulting confidence region, while failing to contain values of
$\theta$ near the boundary of the identified set with high probability,
will not be too much smaller than the true identified set.

The results in this paper also shed light on the tradeoff between the KS
statistics based on integrating conditional moments to get unconditional
moments considered in this paper and other methods for inference with
conditional moment inequalities, such as those based on kernel or series
estimation
\citep{chernozhukov_intersection_2009,ponomareva_inference_2010} or
increasing numbers of
unconditional moments \citep{menzel_estimation_2008}.  With the bandwidth
chosen to decrease at the correct rate, kernel methods based on a supremum
statistic will give
close to the same $n^{-2/(d_X+4)}$ rate (up to a power of $\log n$) for
detecting the local alternatives considered in this paper.  With enough
derivatives imposed on the conditional mean, higher order kernels or
series methods could be used to get even more power.
However,
kernel based methods will perform worse with suboptimal bandwidth choices,
or against
local alternatives in which the conditional moment inequality fails to
hold on a larger set.  The $n^{-2/(d_X+4)}$ rate for detecting local
alternatives can also be achieved within a $\log n$ term using the
increasing truncation point variance weighting proposed in
\citet{armstrong_weighted_2011}.  Unlike the tests proposed in this paper,
those methods are robust to nonsmooth conditional means.  These tests also
have the advantage of adapting to different shapes of the conditional mean
without estimating the optimal bandwidth, as would be necessary with
kernel estimates, or estimating the rate of convergence of a test
statistic, as required by the tests in this paper.
However, they
have less power by a $\log n$ term when applied to this setting, and
require choosing a conservative critical value, which decreases the power
further (but not the rate at which local alternatives can converge to the
identified set and still be detected).

While the results in this paper and \citet{armstrong_weighted_2011}
characterize how moment selection and weighting functions affect relative
efficiency in this setting, the choice of test statistic (supremum norm,
as considered here, or $L_p$ norm, as with Cramer-von Mises statistics)
and instrument functions are also of interest.  While the results in this
paper and in \citet{armstrong_weighted_2011} give some insight into these
problems (for example, it is clear from the arguments in these papers that
Cramer-von Mises style statistics will have less power in this setting
unless new asymptotic distribution results or moment selection procedures
are used)
more complete answers to these questions are topics of ongoing research.

It is also interesting to compare the nonsimilarity problem with the
statistics in this paper to nonsimilarity problems encountered with kernel
based methods.  The rate of convergence of supremum statistics based on
kernel estimates of the conditional moments also depends on the contact
set, but to a lesser extent.  \citet{ponomareva_inference_2010} shows that
the rate of convergence of these statistics differs by a factor of
$\log n$ depending on the contact set.  Arguing as in Section 6 of
\citet{armstrong_weighted_2011}, this would lead to an increase in the
rate at which local alternatives can approach the identified set and still
be detected by a factor of $\log n$.  In contrast, the polynomial
difference in the rates of convergence of the KS statistics based on
integrated moments considered in the present paper leads to increases in
local power by factors of $n$ rather than $\log n$.  Thus, the gains in
terms of power from using exact approximations are much larger in this
context.

In addition to these immediate practical applications, the
results in this paper are also of independent interest in their
relation to broader questions in the literatures on moment
inequalities and nonparametric estimation.  In testing multiple moment
inequalities, the asymptotic distribution of test statistics typically
only depends on inequalities that bind as equalities.  Since the
non binding moments do affect the finite sample distribution of the
test statistic, this means that asymptotic distributions may provide
poor approximations to finite sample distributions.  The existing
literature on moment inequalities has taken several approaches to this
issue.  One is to use conservative approximations using ``least
favorable'' asymptotic distributions in which all moment inequalities
bind.
Another approach is to design tests that are robust to
sequences where the data generating process or test statistic changes
as the number of observations increases.
\citet{menzel_estimation_2008} considers asymptotic approximations in
which the number of moment inequalities used for a test statistic
increases with the number of observations.
\citet{andrews_inference_2009} show that the tests they consider
using test statistics similar to the ones in this paper, but using a
(possibly degenerate) $\sqrt{n}$ asymptotic distribution, have the
correct size asymptotically
when data generating processes change with the sample size within
certain classes of data generating processes.  Since these classes of
data generating processes include sequences where some moment
inequalities are slack, but close to binding, this suggests that the
methods they propose will not suffer from problems with non binding
inequalities affecting the finite sample distribution.

In contrast,
the asymptotic distributions presented in Sections \ref{inf_dist_sec} and
\ref{inf_dist_alpha_sec} of
the present paper are, to my knowledge, the first known case of the
asymptotic distribution of a test statistic that uses a fixed (although,
in this case, infinite) set of moment inequalities depending on moment
inequalities that do not bind.  These results show that, under the
conditions in this paper, the ``moment selection'' problem takes the
form of a balancing of expected value and variance of moments that
are close to binding.  This leads to ideas typically associated with
kernel smoothing
and nonstandard M-estimation %
applying to test
statistics for moment inequalities.
As with the objective functions for nonstandard M-estimation considered by
\citet{kim_cube_1990}, the asymptotic distribution of the KS statistic is
the limit of local processes under a scaling that balances a drift term
and a variance term.  This balancing of drift and variance terms mirrors
the equating of bias and variance terms in choosing the optimal bandwidth
for nonparametric kernel estimation
\citep[see, for example,][]{pagan_nonparametric_1999}.
This is especially interesting since one of the appealing features of KS
style statistics in this setting is that they get rid of the need for
bandwidth parameters.  In the settings I consider, the choice of
``bandwidth'' is made automatically by the balancing of the drift and
variance terms, which determines the scale of the moments that matter
asymptotically.  However, this shows up in the rate of convergence, so
that tests to determine which ``bandwidth'' was chosen are still needed
for exact inference.  Thus, in a sense, the bandwidth selection problem
shows up in the moment selection problem through the rate of convergence.

In another paper \citep{armstrong_weighted_2011}, I show that KS
statistics similar to the ones in the present paper can be made to choose
the moments that correspond to the optimal bandwidth by using a variance
weighting with an increasing sequence of truncation points.  This helps
alleviate the problem with different rates of convergence of the KS
statistic along the boundary of the identified set, but loses a $\log n$
term relative to the tests based on unweighted KS statistics (or KS
statistics with bounded weights) and asymptotic approximations based on
the exact rate of convergence.  Thus, moment selection (in the form of
testing for rates of convergence) and variance weighting play similar
roles in this framework.  Even without the variance weighting of
\citet{armstrong_weighted_2011}, the statistics in this paper find the
moments that lead to the most local power.  Estimating the rate of
convergence of the test statistic is only needed to find the order of
magnitude (under the null) of the moments that were found.

The results in this paper are pointwise in the underlying distribution
$P$.  Since the procedures proposed in this paper involve pre tests, it is
natural to ask for which classes of underlying distributions these tests
are uniformly valid.  Since uniformity in the underlying distribution is
implicit in the bounds used in many of the arguments used to derive these
asymptotic distributions, it seems likely that these tests could be shown
to enjoy
uniformity in classes of distributions with uniform bounds on the
constants governing the smoothness conditions needed for the pointwise
results.
While this would be an interesting extension of the results in the paper,
uniformity in the underlying distribution is perhaps less interesting than
in other settings because many of the tradeoffs between the approach in
the present paper and more conservative approaches are already clear from
the pointwise results.  Smoothness conditions not needed for the
conservative approach to control the size uniformly in the underlying
distribution are needed even for the pointwise results derived here. Thus,
it is clear from the pointwise results that the power improvement achieved
by the tests in this paper comes at a cost of robustness to smoothness
conditions.

Many of the results in this paper assume that the conditional mean $\bar
m(\theta,x)$ is minimized only on a finite set.  For the case where
$d_X=1$, this is implied by smoothness conditions on the conditional mean
(or, when it does not hold, the results in this paper bound the rate of
convergence so that the tests based on estimated rates are still valid).
In higher dimensions, the case where the contact set has infinitely many
points but is of a dimension less than $d_X$ is likely to be more
difficult, but similar ideas will apply.  The results in this paper could
also be extended to the case where the $\bar m(\theta,x)$ only approaches
$0$ near the (possibly infinite) boundary of the support of $x$.  These
cases are often relevant in %
performing inference on bounds on treatment effects such as those
considered by \citet{manski_nonparametric_1990}.
In the one dimensional case, $X_i$ can be
transformed into a uniform random variable so that the conditions on the
density of the conditioning variable used in this paper will apply (once
the density is positive and well behaved on its support, the assumption
that the contact point is on the interior of the support is easy to
relax).  If the density and conditional mean approach zero at polynomial
rates, the transformed model will fit into a slight extension of Theorem
\ref{inf_dist_thm_alpha} for some $\gamma$ that depends on these rates.
These transformations are used in a slightly different setting in
\citet{armstrong_weighted_2011}.

\section{Conclusion}\label{conclusion_sec}

This paper derives the asymptotic distribution of a class of
Kolmogorov-Smirnov style test statistics for conditional moment inequality
models under a general set of conditions.  I show how to use these results
to form valid tests that are more powerful than existing approaches based
on this statistic.  Local power results for the new tests and existing
tests are derived, which quantify this power improvement.  While the
increase in power comes
at a cost of robustness to smoothness conditions, a complementary paper
\citep{armstrong_weighted_2011} proposes methods for inference that
achieve almost the same power improvement while still being robust to
failure of smoothness conditions.

In addition to their immediate practical application to asymptotically
exact inference, the results in this paper add to our understanding of
how familiar issues in the literatures on moment inequalities and
nonparametric estimation, such as moment selection and the curse of
dimensionality, manifest themselves in the use of one sided KS
statistics for conditional moment inequalities.  Under the conditions
in this paper, the asymptotic distribution of the KS statistic depends
on nonbinding moments, which are determined through a balancing of a
bias term and a variance term in a way that is similar to the objective
functions for the point
estimators considered by \citet{kim_cube_1990}.  The dimension of the
conditioning variables and the smoothness of the conditional mean
determine which moments matter asymptotically and which types of local
alternatives the KS statistic can detect.

\section*{Appendix}

This appendix contains proofs of the theorems in this paper.  The proofs
are organized into subsections according to the section containing the
theorem in the body of the paper.  In cases where a result follows
immediately from other theorems or arguments in the body of the paper, I
omit a separate proof.  Statements involving convergence in distribution
in which random elements in the converging sequence are not measurable
with respect to the relevant Borel sigma algebra are in the sense of outer
weak convergence \citep[see][]{van_der_vaart_weak_1996}.
For notational convenience, I use $d=d_X$
throughout this appendix.

\subsection*{Asymptotic Distribution of the KS Statistic}

In this subsection of the appendix, I prove Theorem
\ref{inf_dist_thm_multi}.  For notational convenience, let
$Y_i=m(W_i,\theta)$ and $Y_{i,J(m)}=m_{J(m)}(W_i,\theta)$ and let $d=d_X$
and $k=d_Y$ throughout this subsection.

The asymptotic distribution comes from the behavior of the objective
function $E_nY_{i,j}I(s<X_i<s+t)$ for $(s,t)$ near $x_m$ such that
$j\in J(m)$.  The bulk of the proof involves showing that the objective
function doesn't matter for $(s,t)$ outside of neighborhoods of $x_m$ with
$j\in J(m)$ where these neighborhoods shrink at a fast enough rate.
First, I derive the limiting distribution over such shrinking
neighborhoods and the rate at which they shrink.

\begin{theorem}\label{local_process_thm_multi}
Let $h_n=n^{-\alpha}$ for some $0<\alpha<1/d$.  Let
\begin{align*}
\mathbb{G}_{n,x_m}(s,t)
=\frac{\sqrt{n}}{h_n^{d/2}}(E_n-E)Y_{i,J(m)}I(h_ns<X_i-x_m<h_n(s+t))
\end{align*}
and let $g_{n,x_m}(s,t)$ have $j$th element
\begin{align*}
g_{n,x_m,j}(s,t)=\frac{1}{h_n^{d+2}}EY_{i,j}I(h_ns<X_i-x_m<h_n(s+t))
\end{align*}
if $j\in J(m)$ and zero otherwise.
Then, for any finite $M$,
$(\mathbb{G}_{n,x_1}(s,t),\ldots,\mathbb{G}_{n,x_\ell}(s,t))
\stackrel{d}{\to}
(\mathbb{G}_{P,x_1}(s,t),\ldots,\mathbb{G}_{P,x_\ell}(s,t))$
taken as random
processes on $\|(s,t)\|\le M$ with the supremum norm
and $g_{n,x_m}(s,t)\to g_{P,x_m}(s,t)$ uniformly in $\|(s,t)\|\le M$
where $\mathbb{G}_{P,x_m}(s,t)$ and $g_{P,x_m}(s,t)$ are defined as in
Theorem \ref{inf_dist_thm_multi} for $m$ from $1$ to $\ell$.
\end{theorem}
\begin{proof}
The convergence in distribution in the first statement follows from
verifying the conditions of Theorem 2.11.22 in \citet{van_der_vaart_weak_1996}.
To derive the covariance kernel, note that
\begin{align*}
&cov(\mathbb{G}_{n,x_m}(s,t),\mathbb{G}_{n,x_m}(s',t'))  \\
&=h_n^{-d}EY_{i,J(m)}Y_{i,J(m)}'I\left\{h_n(s\vee s')<X-x_m
<h_n\left[(s+t)\wedge (s'+t')\right]\right\}  \\
&-h_n^{-d}\left\{EY_{i,J(m)}I\left[h_ns<X-x_m<h_n(s+t)\right]\right\}
\left\{EY_{i,J(m)}'I\left[h_ns'<X-x_m<h_n(s'+t')\right]\right\}.
\end{align*}
The second term goes to zero as $n\to\infty$.  The first is equal to the
claimed covariance kernel plus the error term
\begin{align*}
h_n^{-d}\int_{h_n(s\vee s')<x-x_m<h_n\left[(s+t)\wedge (s'+t')\right]}
\left[E(Y_{i,J(m)}Y_{i,J(m)}'|X=x)f_X(x)
-E(Y_{i,J(m)}Y_{i,J(m)}'|X=x_m)f_X(x_m)\right]\, dx,
\end{align*}
which is bounded by
\begin{align*}
&\left\{\max_{\|x-x_m\|\le 2h_nM}\left[E(Y_{i,J(m)}Y_{i,J(m)}'|X=x)f_X(x)
-E(Y_{i,J(m)}Y_{i,J(m)}'|X=x_m)f_X(x_m)\right]\right\}  \\
&\times h_n^{-d}\int_{h_n(s\vee s')<x-x_m<h_n\left[(s+t)\wedge (s'+t')\right]}
\, dx  \\
&= \left\{\max_{\|x-x_m\|\le 2h_nM}\left[E(Y_{i,J(m)}Y_{i,J(m)}'|X=x)f_X(x)
-E(Y_{i,J(m)}Y_{i,J(m)}'|X=x_m)f_X(x_m)\right]\right\}  \\
&\times \int_{(s\vee s')<x-x_m<(s+t)\wedge (s'+t')}\, dx.
\end{align*}
This goes to zero as $n\to\infty$ by continuity of
$E(Y_{i,J(m)}Y_{i,J(m)}'|X=x)$ and $f_X(x)$.  For $m\ne r$ and
$\|(s,t)\|\le M$, $\|(s',t')\|\le M$,
$cov(\mathbb{G}_{n,x_m}(s,t),\mathbb{G}_{n,x_r}(s',t'))$ is eventually
equal to
\begin{align*}
-h_n^{-d}\left\{EY_{i,J(m)}I\left[h_ns<X-x_m<h_n(s+t)\right]\right\}
\left\{EY_{i,J(r)}'I\left[h_ns'<X-x_r<h_n(s'+t')\right]\right\},
\end{align*}
which goes to zero, so the processes for different elements of
$\mathcal{X}_0$ are independent as claimed.

For the claim regarding $g_{n,x_m}(s,t)$, first note that the assumptions
imply
that, for $j\in J(m)$, the first derivative of $x\mapsto E(Y_{i,j}|X=x)$
at $x=x_m$ is $0$, and that this function has a second order Taylor
expansion:
\begin{align*}
E(Y_{i,j}|X=x)=\frac{1}{2}(x-x_m)'V_j(x_m)(x-x_m)+R_n(x)
\end{align*}
where
\begin{align*}
R_n(x)=\frac{1}{2}(x-x_m)'V_j(x^*(x))(x-x_m)
-\frac{1}{2}(x-x_m)'V_j(x_m)(x-x_m)
\end{align*}
and $V_j(x^*)$ is the second derivative matrix evaluated at some $x^*(x)$
between $x_m$ and $x$.

We have
\begin{align*}
g_{n,x_m,j}(s,t)&=\frac{1}{2h_n^{d+2}}
\int_{h_ns<x-x_m<h_n(s+t)} (x-x_m)'V_j(x_m)(x-x_m) f_X(x_m)\, dx  \\
&+\frac{1}{2h_n^{d+2}}\int_{h_ns<x-x_m<h_n(s+t)} (x-x_m)'V_j(x_m)(x-x_m)
    [f_X(x)-f_X(x_m)]\, dx  \\
&+\frac{1}{h_n^{d+2}}\int_{h_ns<x-x_m<h_n(s+t)} R_n(x)f_X(x)\, dx.
\end{align*}
The first term is equal to $g_{P,x_m,j}(s,t)$ by a change of variable $x$
to $h_nx+x_m$ in the integral.  The second term is bounded by
$g_{P,x_m,j}(s,t)\sup_{\|x-x_m\|\le 2h_n M} [f_X(x)-f_X(x_m)]/f_X(x_m)$,
which goes to zero uniformly in $\|(s,t)\|\le M$ by continuity of $f_X$.
The third term is equal to (using the same change of variables)
\begin{align*}
\frac{1}{2} \int_{s<x<s+t}
[x'V_j(x^*(h_nx+x_m))x-x'V_j(x_m)x]f_X(h_nx+x_m)\, dx.
\end{align*}
This is bounded by a constant times
$\sup_{\|x\|\le M} |x'V_j(x^*(h_nx+x_m))x-x'V_j(x_m)x|$,
which goes to zero as $n\to\infty$ by continuity of the second
derivatives.
\end{proof}

Thus, if we let $h_n$ be such that
$\sqrt{n}/h_n^{d/2}=1/h_n^{d+2}
\Longleftrightarrow h_n=n^{-1/(d+4)}$ and scale up by
$\sqrt{n}/h_n^{d/2}=1/h_n^{d+2}=n^{(d+2)/(d+4)}$, we will have
\begin{align*}
&n^{(d+2)/(d+4)}(E_nY_{i,J(1)}I(h_ns<X-x_1<h_n(s+t)),
\ldots,E_nY_{i,J(\ell)}I(h_ns<X-x_\ell<h_n(s+t))  \\
&=(\mathbb{G}_{n,x_1}(s,t)+g_{n,x_1}(s,t),
\ldots,\mathbb{G}_{n,x_\ell}(s,t)+g_{n,x_\ell}(s,t))  \\
&\stackrel{d}{\to}
(\mathbb{G}_{P,x_1}(s,t)+g_{P,x_1}(s,t),\ldots,
\mathbb{G}_{P,x_m}(s,t)+g_{P,x_m}(s,t))
\end{align*}
taken as stochastic processes in $\{\|(s,t)\|\le M\}$ with the supremum
norm.  From now on, let $h_n=n^{-1/(d+4)}$ so that this will hold.

We would like to show that the infimum of these stochastic
processes over all of $\mathbb{R}^{2d}$ converges to the infimum of the
limiting process over all of $\mathbb{R}^{2d}$, but this does not follow
immediately since we only have uniform convergence on compact sets.
Another way of thinking about this problem is that convergence in
distribution in $\{\|(s,t)\|\le M\}$ with the supremum norm for any $M$
implies convergence in distribution in $\mathbb{R}^{2d}$ with the topology
of uniform convergence on compact sets \citep[see][]{kim_cube_1990},
but the infimum over all of $\mathbb{R}^{2d}$ is not a continuous mapping
on this space since uniform convergence on all compact sets does not imply
convergence of the infimum over all of $\mathbb{R}^{2d}$.  To get the
desired result, the following lemma will be useful.  The idea is to show
that values of $(s,t)$ far away from zero won't matter for the limiting
distribution, and then use convergence for fixed compact sets.

\begin{lemma}\label{inf_dist_lemma_multi}
Let $\mathbb{H}_n$ and $\mathbb{H}_P$ be random functions from
$\mathbb{R}^{k_1}$ to $\mathbb{R}^{k_2}$
such that, (i) for all $M$, $\mathbb{H}_n\stackrel{d}{\to} \mathbb{H}_P$
when $\mathbb{H}_n$ and $\mathbb{H}_P$ are taken as random processes on
$\{t\in \mathbb{R}^{k_1}| \|t\|\le M\}$ with the supremum norm, (ii) for all
$r<0$, $\varepsilon>0$, there exists an $M$ such that
$P\left(\inf_{\|t\|>M} \mathbb{H}_{P,j}(t) \le r
  \text{ some $j$}\right)<\varepsilon$
and an $N$ such that
$P\left(\inf_{\|t\|>M} \mathbb{H}_{n,j}(t) \le r \text{ some
  $j$}\right)<\varepsilon$
for all $n\ge N$ and (iii) $\inf_t \mathbb{H}_{n,j}(t)\le 0$ and $\inf_t
\mathbb{H}_{P,j}(t)\le 0$ with probability one.  Then $\inf_{t\in\mathbb{R}^{k_1}}
\mathbb{H}_n(t)\stackrel{d}{\to} \inf_{t\in\mathbb{R}^{k_1}} \mathbb{H}_P(t)$.
\end{lemma}
\begin{proof}
First, by the Cramer-Wold device, it suffices to show that, for all $w\in
\mathbb{R}^{k_2}$,
$w'\inf_{t\in\mathbb{R}^{k_1}} \mathbb{H}_n(t)\stackrel{d}{\to}
w'\inf_{t\in\mathbb{R}^{k_1}} \mathbb{H}_P(t)$.  For this,
it suffices to show that for all
$r\in\mathbb{R}$,
$\liminf_n P\left(w'\inf_{t\in\mathbb{R}^{k_1}} \mathbb{H}_n(t)<r\right)
\ge P\left(w'\inf_{t\in\mathbb{R}^{k_1}} \mathbb{H}_P(t)<r\right)$
and
$\limsup_n P\left(w'\inf_{t\in\mathbb{R}^{k_1}} \mathbb{H}_n(t)\le r\right)
\le P\left(w'\inf_{t\in\mathbb{R}^{k_1}} \mathbb{H}_P(t)\le r\right)$
since, arguing along the lines of the Portmanteau Lemma,
when $r$ is a continuity point
of the limiting distribution, we will have
\begin{align*}
&P\left(w'\inf_{t\in\mathbb{R}^{k_1}} \mathbb{H}_P(t)\le r\right)
=P\left(w'\inf_{t\in\mathbb{R}^{k_1}} \mathbb{H}_P(t)<r\right)
\le \liminf_n P\left(w'\inf_{t\in\mathbb{R}^{k_1}} \mathbb{H}_n(t)<r\right)  \\
&\le \liminf_n P\left(w'\inf_{t\in\mathbb{R}^{k_1}} \mathbb{H}_n(t)\le r\right)
\le \limsup_n P\left(w'\inf_{t\in\mathbb{R}^{k_1}} \mathbb{H}_n(t)\le r\right)
\le P\left(w'\inf_{t\in\mathbb{R}^{k_1}} \mathbb{H}_P(t)\le r\right).
\end{align*}

Given $\varepsilon>0$, let $M$ and $N$ be as in the assumptions of the
lemma, but with $r$ replaced by $r/(k_2 \max_i |w_i|)$.  Then
\begin{align*}
P\left(w'\inf_{\|t\|\ge M} \mathbb{H}_P(t)<r\right)
\le P\left((k_2 \max_i |w_i|)\inf_{\|t\|\ge M} \mathbb{H}_{P,j}(t)<r
\text{ some $j$}\right)
\le \varepsilon
\end{align*}
so that
$P\left(w'\inf_{\|t\|\le M} \mathbb{H}_P(t)<r\right)
+\varepsilon
\ge P\left(w'\inf_{t\in\mathbb{R}^{k_1}} \mathbb{H}_P(t)<r\right)$
and, for $n\ge N$,
\begin{align*}
P\left(w'\inf_{\|t\|\ge M} \mathbb{H}_n(t)\le r\right)
\le P\left((k_2 \max_i |w_i|)\inf_{\|t\|\ge M} \mathbb{H}_{n,j}(t)\le r
\text{ some $j$}\right)
\le \varepsilon
\end{align*}
so that
$P\left(w'\inf_{\|t\|\le M} \mathbb{H}_n(t)\le r\right)
+\varepsilon
\ge P\left(w'\inf_{t\in\mathbb{R}} \mathbb{H}_n(t)\le r\right)$.
Thus, by convergence in distribution of the infima over $\|t\|\le M$,
\begin{align*}
&\liminf_n P\left(w'\inf_{t\in\mathbb{R}^{k_1}} \mathbb{H}_n(t)<r\right)
\ge \liminf_n P\left(w'\inf_{\|t\|\le M} \mathbb{H}_n(t)<r\right)
\ge P\left(w'\inf_{\|t\|\le M} \mathbb{H}_P(t)<r\right)  \\
&\ge P\left(w'\inf_{t\in\mathbb{R}^{k_1}}
\mathbb{H}_P(t)<r\right)-\varepsilon
\end{align*}
and
\begin{align*}
&\limsup_n P\left(w'\inf_{t\in\mathbb{R}^{k_1}} \mathbb{H}_n(t)\le r\right)
\le \limsup_n P\left(w'\inf_{\|t\|\le M} \mathbb{H}_n(t)\le r\right)
+ \varepsilon  \\
&\le P\left(w'\inf_{\|t\|\le M} \mathbb{H}_P(t)\le r\right)
+ \varepsilon
\le P\left(w'\inf_{t\in\mathbb{R}^{k_1}} \mathbb{H}_P(t)\le r\right)
+ \varepsilon.
\end{align*}
Since $\varepsilon$ was arbitrary, this gives the desired result.

\end{proof}

Technically, this lemma does not apply to
\begin{align*}
(\mathbb{G}_{n,x_1}(s,t)+g_{n,x_1}(s,t),
\ldots,\mathbb{G}_{n,x_\ell}(s,t)+g_{n,x_\ell}(s,t))
\end{align*}
since, for $m\ne r$, $\mathbb{G}_{n,x_m}(s,t)+g_{n,x_m}(s,t)$ evaluated at
some increasing values of $(s,t)$ may actually be equal to
$\mathbb{G}_{n,x_r}(s',t')+g_{n,x_r}(s',t')$ for some small values of
$(s',t')$, since, once the local indices are large enough, the original
indices overlap.  Instead, noting that, for any $\eta>0$,
\begin{align*}
&n^{(d+2)/(d+4)} \inf_{s,t} E_nY_iI(s<X_i<s+t)  \\
&=\left(\min_{m \text{ s.t. } 1\in J(m)} \inf_{\|(s,t)\|\le \eta/h_n}
\mathbb{G}_{n,x_m,1}(s,t)+g_{n,x_m,1}(s,t)),\ldots,\right.  \\
&\left.\min_{m \text{ s.t. } k\in J(m)} \inf_{\|(s,t)\|\le \eta/h_n}
\mathbb{G}_{n,x_m,k}(s,t)+g_{n,x_m,k}(s,t)\right)  \\
&\wedge \left(n^{(d+2)/(d+4)}
\inf_{\|(s-x_m,t)\|>\eta \text{ all $m$ s.t. $1\in J(m)$}}
E_nY_{i,1}I(s<X_i<s+t),\ldots,\right.  \\
&\left.n^{(d+2)/(d+4)}
\inf_{\|(s-x_m,t)\|>\eta \text{ all $m$ s.t. $k\in J(m)$}}
E_nY_{i,k}I(s<X_i<s+t)\right)  \\
&\equiv Z_{n,1}\wedge Z_{n,2},
\end{align*}
I show that, for some $\eta>0$, $Z_{n,2}\stackrel{p}{\to} 0$ using a
separate argument, and use Lemma \ref{inf_dist_lemma_multi} to show that,
for the same $\eta$,
\begin{align*}
&(\inf_{s,t}[\mathbb{G}_{n,x_1}(s,t)+g_{n,x_1}(s,t)]I(\|(s,t)\|\le \eta/h_n),
\ldots,\inf_{s,t}[\mathbb{G}_{n,x_\ell}(s,t)+g_{n,x_\ell}(s,t)]
I(\|(s,t)\|\le \eta/h_n))  \\
&\stackrel{d}{\to}
(\inf_{s,t}\mathbb{G}_{P,x_1}(s,t)+g_{P,x_1}(s,t),
\ldots,\inf_{s,t}\mathbb{G}_{P,x_\ell}(s,t)+g_{P,x_\ell}(s,t)),
\end{align*}
from which it follows that $Z_{n,1}\stackrel{d}{\to} Z$ for $Z$ defined as
in Theorem \ref{inf_dist_thm_multi} by the continuous mapping theorem.

Part (i) of Lemma \ref{inf_dist_lemma_multi} follows from Theorem
\ref{local_process_thm_multi} (the $I(\|(s,t)\|\le \eta/h_n)$ term does
not change this, since it is equal to one for $\|(s,t)\|\le M$
eventually).  Part (iii) follows since the processes involved are equal to
zero when $t=0$.  To verify part (ii), first note that it suffices to
verify part (ii) of the lemma for
$\mathbb{G}_{n,x_m,j}(s,t)+g_{n,x_m,j}(s,t)$ and
$\mathbb{G}_{P,x_m,j}(s,t)+g_{P,x_m,j}(s,t)$ for each $m$ and $j$
individually.  Part (ii) of the lemma holds trivially for $m$ and $j$ such
that $j\notin J(m)$, so we need to verify this part of the lemma for $m$
and $j$ such that $j\in J(m)$.

The next two lemmas provide bounds that will be used
to verify condition (ii) of Lemma \ref{inf_dist_lemma_multi} for
$\mathbb{G}_{n,x_m,j}(s,t)+g_{n,x_m,j}(s,t)$ and
$\mathbb{G}_{P,x_m,j}(s,t)+g_{P,x_m,j}(s,t)$ for $m$ and $j$ with $j\in
J(m)$.
To do this, the bounds in the lemmas are applied to sets of
$(s,t)$ with $\|(s,t)\|$ increasing.  The idea is similar to the
``peeling'' argument of, for example, \citet{kim_cube_1990},
but different arguments are required to deal with values of $(s,t)$ for
which, even though $\|s\|$ is large, $\prod_i t_i$ is small so that the
objective function on average uses only a few observations, which may
happen to be negative.  To get bounds on the suprema of the limiting and
finite sample processes where $t$ may be small relative to $s$, the
next two lemmas bound the supremum by a maximum over $s$ in a finite grid
of suprema over $t$ with $s$ fixed, and then use exponential bounds on
suprema of the processes with fixed $s$.

\begin{lemma}\label{tail_bnd_G_multi}
Fix $m$ and $j$ with $j\in J(m)$.
For some $C>0$ that depends only on $d$, $f_X(x_m)$ and
$E(Y_{i,j}^2|X=x_m)$, we have,
for any $B\ge 1$, $\varepsilon>0$, $w>0$,
\begin{align*}
P\left(
\sup_{\|(s,t)\|\le B, \prod_i t_i\le \varepsilon}
|\mathbb{G}_{P,x_m,j}(s,t)|\ge w \right)
\le 2\left\{3B[B^d/(\varepsilon\wedge 1)]+2\right\}^{2d}
\exp\left(-C\frac{w^2}{\varepsilon}\right)
\end{align*}
for $\frac{w^2}{\varepsilon}$ greater than some constant that depends only
on $d$, $f_X(x_m)$ and $E(Y_{i,j}^2|X=x_m)$.
\end{lemma}
\begin{proof}
Let $\mathbb{G}(s,t)=\mathbb{G}_{P,x_m,j}(s,t)$.  We have, for any
$s_0\le s\le s+t\le s_0+t_0$,
\begin{align*}
&\mathbb{G}(s,t)
=\mathbb{G}(s_0,t+s-s_0)  \\
&+\sum_{1\le j \le d} (-1)^j \sum_{1\le i_1<i_2<\ldots<i_j\le d}  \\
&\mathbb{G}(s_0,(t_1+s_1-s_{0,1},\ldots,
t_{i_1-1}+s_{i_1-1}-s_{0,i_1-1},s_{i_1}-s_{0,i_1},t_{i_1+1}+s_{i_1+1}-s_{0,i_1+1},  \\
& \ldots,
t_{i_j-1}+s_{i_j-1}-s_{0,i_j-1},s_{i_j}-s_{0,i_j},t_{i_j+1}+s_{i_j+1}-s_{0,i_j+1},\ldots,
t_d+s_d-s_{0,d})).
\end{align*}
Thus, since there are $2^d$ terms in the above display, each with absolute
value bounded by $\sup_{t\le t_0} |\mathbb{G}(s_0,t)|$,
\begin{align*}
\sup_{s_0\le s\le s+t\le s_0+t_0}
|\mathbb{G}(s,t)|
\le 2^d \sup_{t\le t_0} |\mathbb{G}(s_0,t)|
\stackrel{d}{=}
2^d \sup_{t\le t_0} |\mathbb{G}(0,t)|.
\end{align*}

Let $A$ be a grid of meshwidth $(\varepsilon\wedge 1)/B^d$ covering
$[-B,2B]^d$.  For any $(s,t)$ with $\|(s,t)\|\le B$ and $\prod_i t_i\le
\varepsilon$,
there are $s_0$ and $t_0$ with $s_0,s_0+t_0\in A$ such that
$s_0\le s\le s+t\le s_0+t_0$, and
$\prod_i t_{0,i}\le \prod_i (t_i+(\varepsilon\wedge 1)/B^d)
=\sum_{j=0}^d [(\varepsilon\wedge 1)/B^d]^j
\sum_{I\in\{1,\ldots,d\}, |I|=d-j}\prod_{i\in I} t_i
\le \prod_i t_i
+\sum_{j=1}^d [(\varepsilon\wedge 1)/B^d]^j
{d\choose d-j} B^{d-j}
\le \varepsilon
+\varepsilon \sum_{j=1}^d B^{-dj}
{d\choose d-j} B^{d-j}
\le 2^d \varepsilon$.  For this $s_0,t_0$, we
will then have, by the above display,
$|\mathbb{G}(s,t)|\le 2^d \sup_{t\le t_0} |\mathbb{G}(s_0,t)|$.

This gives
\begin{align*}
\sup_{\|(s,t)\|\le B, \prod_i t_i\le \varepsilon} |\mathbb{G}(s,t)|
\le 2^d \max_{s_0,s_0+t_0\in A, \prod_i t_{0,i}\le 2^d\varepsilon}
\sup_{t\le t_0} |\mathbb{G}(s_0,t)|,
\end{align*}
so that
\begin{align*}
&P\left(\sup_{\|(s,t)\|\le B, \prod_i t_i\le \varepsilon} |\mathbb{G}(s,t)|
\ge w\right)
\le |A|^2\max_{s_0,s_0+t_0\in A, \prod_i t_{0,i}\le 2^d\varepsilon}
P\left(2^d \sup_{t\le t_0} |\mathbb{G}(s_0,t)|\ge w\right)  \\
&= |A|^2\max_{s_0,s_0+t_0\in A, \prod_i t_{0,i}\le 2^d\varepsilon}
P\left(2^d \sup_{t\le 1} \left(\prod_i t_{0,i}\right)^{1/2}
|\mathbb{G}(0,t)| \ge w\right)  \\
&\le |A|^2%
P\left(\sup_{t\le 1}
|\mathbb{G}(0,t)| \ge \frac{w}{2^d 2^{d/2}\varepsilon^{1/2}}\right).
\end{align*}
The result then follows using the fact that
$|A|\le \left\{3B[B^d/(\varepsilon\wedge 1)]+2\right\}^{d}$
and using Theorem 2.1 (p.43) in \citet{adler_introduction_1990}
to bound the probability in the last line of the display
(the theorem in \citet{adler_introduction_1990} shows that the probability
in the above display is bounded by
$2\exp(-K_1w^2/\varepsilon+K_2w/\varepsilon^{1/2}+K_3)$
for some constants $K_1$, $K_2$, and
$K_3$ with $K_1>0$ that depend only on $d$, $f_X(x_m)$ and
$E(Y_{i,j}^2|X=x_m)$ and this expression is less than
$2\exp(-(K_1/2)w^2/\varepsilon)$ for $w^2/\varepsilon$ greater than some
constant that depends only on $K_1$, $K_2$, and $K_3$).

\end{proof}

\begin{lemma}\label{tail_bound_Gn}
Fix $m$ and $j$ with $j\in J(m)$.
For some $C>0$ that depends only on the distribution of $(X,Y)$ and some
$\eta>0$, we have,
for any $1\le B\le h_n^{-1}\eta$, $w>0$ and
$\varepsilon \ge n^{-4/(d+4)}(1+\log n)^2$,
\begin{align*}
P\left(
\sup_{\|(s,t)\|\le B, \prod_i t_i\le \varepsilon}
|\mathbb{G}_{n,x_m,j}(s,t)|\ge w \right)
\le 2\left\{3B[B^d/(\varepsilon\wedge 1)]+2\right\}^{2d}
\exp\left(-C\frac{w}{\varepsilon^{1/2}}\right).
\end{align*}
\end{lemma}
\begin{proof}
Let $\mathbb{G}_{n}(s,t)=\mathbb{G}_{n,x_m,j}(s,t)$.
By the same argument as in the previous lemma with $\mathbb{G}$ replaced
by $\mathbb{G}_n$, we have
\begin{align*}
\sup_{s_0\le s\le s+t\le s_0+t_0} |\mathbb{G}_n(s,t)|
\le 2^d \sup_{t\le t_0} |\mathbb{G}_n(s_0,t)|.
\end{align*}
As in the previous lemma, let $A$ be a grid of meshwidth
$(\varepsilon\wedge 1)/B^d$ covering $[-B,2B]^d$.  Arguing as in the
previous lemma, we have, for any $(s,t)$ with $\|(s,t)\|\le B$ and
$\prod_i t_i\le \varepsilon$, there exists some $s_0,t_0$ with
$s_0,s_0+t_0\in A$ such that $\Pi_it_{0,i}\le 2^d\varepsilon$ and
$|\mathbb{G}_n(s,t)|\le 2^d \sup_{t\le t_0} |\mathbb{G}_n(s_0,t)|$.  Thus,
\begin{align*}
&\sup_{\|(s,t)\|\le B, \prod_i t_i\le \varepsilon} |\mathbb{G}_n(s,t)|
\le 2^d \max_{s_0,s_0+t_0\in A, \prod_i t_{0,i}\le 2^d\varepsilon}
\sup_{t\le t_0} |\mathbb{G}_n(s_0,t)|  \\
&= 2^d \max_{s_0,s_0+t_0\in A, \prod_i t_{0,i}\le 2^d\varepsilon}
\sup_{t\le t_0} \frac{\sqrt{n}}{h_n^{d/2}}
|(E_n-E)Y_{i,j}I(h_ns_0\le X_i-x_m\le h_n(s_0+t))|.
\end{align*}
This gives
\begin{align*}
&P\left(\sup_{\|(s,t)\|\le B, \prod_i t_i\le 2^d\varepsilon}
|\mathbb{G}_n(s,t)|\ge w\right)  \\
&\le |A|^2 \max_{s_0,s_0+t_0\in A, \prod_i t_{0,i}\le 2^d\varepsilon}
P\left(2^d\sup_{t\le t_0} \frac{\sqrt{n}}{h_n^{d/2}}
|(E_n-E)Y_{i,j}I(h_ns_0\le X_i-x_m\le h_n(s_0+t))| \ge w\right).
\end{align*}
We have, for some universal constant $K$ and all $n$ with $\varepsilon
\ge n^{-4/(d+4)}(1+\log n)^2$, letting
$\mathcal{F}_n=\{(x,y)\mapsto y_jI(h_ns_0\le x-x_m\le h_n(s_0+t))
|t\le t_0\}$ and defining $\|\cdot\|_{P,\psi_1}$ to be the Orlicz norm
defined on p.90 of \citet{van_der_vaart_weak_1996} for
$\psi_1(x)=\exp(x)-1$,
\begin{align*}
&\|2^d\sup_{f\in\mathcal{F}_n} |\sqrt{n}(E_n-E)f(X_i,Y_i)|\|_{P,\psi_1}  \\
&\le K\left[
E \sup_{f\in\mathcal{F}_n} |\sqrt{n}(E_n-E)f(X_i,Y_i)|
+ n^{-1/2}(1+\log n)
\||Y_{i,j}|I(h_ns_0\le X_i-x_m\le h_n(s_0+t_0))\|_{P,\psi_1}\right]  \\
&\le K\left[
J(1,\mathcal{F}_n,L^2) \left\{E
[|Y_{i,j}|I(h_ns_0<X_i-x_m<h_n(s_0+t_0))]^2\right\}^{1/2}
+ n^{-1/2}(1+\log n)
\|Y\|_{P,\psi_1}\right]  \\
&\le K\left[
J(1,\mathcal{F}_n,L^2) \overline f^{1/2} \overline Y
h_n^{d/2} 2^{d/2}\varepsilon^{1/2}
+ n^{-1/2}(1+\log n)
\|Y_{i,j}\|_{P,\psi_1}\right]  \\
&\le K\left[
J(1,\mathcal{F}_n,L^2) \overline f^{1/2} \overline Y 2^{d/2}
+ \|Y_{i,j}\|_{P,\psi_1}\right]h_n^{d/2}\varepsilon^{1/2}.
\end{align*}
The first inequality follows by Theorem 2.14.5 in
\citet{van_der_vaart_weak_1996}.
The second uses Theorem 2.14.1 in \citet{van_der_vaart_weak_1996}.
The fourth inequality uses the fact that
$h_n^{d/2}\varepsilon^{1/2}= n^{-d/[2(d+4)]}\varepsilon^{1/2}
\ge n^{-1/2}(1+\log n)$ once
$\varepsilon^{1/2}\ge n^{-1/2+d/[2(d+4)]}(1+\log n)
=n^{-2/(d+4)}(1+\log n)$.
Since each
$\mathcal{F}_n$ is contained in the larger class $\mathcal{F}$ defined
in the same way but replacing $s_0$ with $s$, and allowing $(s,t)$ to vary
over all of $\mathbb{R}^{2d}$, we can replace $\mathcal{F}_n$ by
$\mathcal{F}$ on the last line of this display.  Since
$J(1,\mathcal{F},L^2)$ and $\|Y_{i,j}\|_{\psi_1}$ are finite
($\mathcal{F}$ is a VC class and $Y_{i,j}$ is bounded), the bound is equal
to $C^{-1} \varepsilon^{1/2} h_n^{d/2}$ for a constant $C$ that depends only
on the distribution of $(X_i,Y_i)$.

This bound along with Lemma 8.1 in \citet{kosorok_introduction_2008}
implies
\begin{align*}
&P\left(2^d\sup_{t\le t_0} \frac{\sqrt{n}}{h_n^{d/2}}
|(E_n-E)Y_{i,j}I(h_ns_0\le X_i-x_m\le h_n(s_0+t))| \ge w\right)  \\
&=P\left(2^d\sup_{f\in\mathcal{F}_n} |\sqrt{n}(E_n-E)f(X_i,Y_i)|
\ge w h_n^{d/2}\right)  \\
&\le 2\exp\left(
-\frac{w h_n^{d/2}}
{\|2^d\sup_{f\in\mathcal{F}_n}|\sqrt{n}(E_n-E)f(X_i,Y_i)|\|_{P,\psi_1}}
\right)  \\
&\le 2\exp\left(-\frac{w h_n^{d/2}}
  {C^{-1} h_n^{d/2}\varepsilon^{1/2}}\right)
= 2\exp\left(-Cw/\varepsilon^{1/2}\right).
\end{align*}
The result follows using this and the fact that
$|A|\le \left\{3B[B^d/(\varepsilon\wedge 1)]+2\right\}^{d}$.
\end{proof}

The following theorem verifies the part of condition (ii) of Lemma
\ref{inf_dist_lemma_multi} concerning the limiting process
$\mathbb{G}_{P,x_m,j}(s,t)+g_{P,x_m,j}(s,t)$.

\begin{theorem}\label{inf_bound_G}
Fix $m$ and $j$ with $j\in J(m)$.
For any $r<0$, $\varepsilon>0$ there exists an $M$ such that
\begin{align*}
P\left(\inf_{\|(s,t)\|>M} \mathbb{G}_{P,x_m,j}(s,t)+g_{P,x_m,j}(s,t)
\le r\right)<\varepsilon.
\end{align*}
\end{theorem}
\begin{proof}
Let $\mathbb{G}(s,t)=\mathbb{G}_{P,x_m,j}(s,t)$ and
$g(s,t)=g_{P,x_m,j}(s,t)$.
Let $S_k=\{k\le \|(s,t)\|\le k+1\}$ and let $S_k^L=S_k\cap \{\prod_i
t_i\le (k+1)^{-\delta}\}$ for some fixed $\delta$.  By
Lemma \ref{tail_bnd_G_multi},
\begin{align*}
&P\left(\inf_{S_k^L} \mathbb{G}(s,t)+g(s,t)\le r\right)
\le P\left(\sup_{S_k^L} |\mathbb{G}(s,t)|\ge |r|\right)  \\
&\le 2\left\{3(k+1)[(k+1)^d/k^{-\delta}]+2\right\}^{2d}
\exp\left(-C r^2 (k+1)^{\delta}\right)
\end{align*}
for $k$ large enough where $C$ depends only on $d$.  This bound is
summable over $k$.

For any $\alpha$ and $\beta$ with $\alpha<\beta$, let
$S_k^{\alpha,\beta}=S_k\cap \{(k+1)^\alpha<\prod_i t_i\le (k+1)^\beta\}$.  We
have, for some $C_1>0$ that depends only on $d$ and
$V_j(x_m)$, $g(s,t)\ge C_1\|(s,t)\|^2\prod_i t_i$.
(To see this, note that $g(s,t)$ is greater than or equal to a constant
times $\int_{s_1}^{s_1+t_1}\cdots\int_{s_d}^{s_d+t_d} \|x\|^2dx_d\cdots
dx_1=\left(\Pi_{i=1}^d t_i\right)\sum_{i=1}^d (s_i^2+t_i^2/3+s_it_i)$,
and the sum can be bounded below by a constant times $\|(s,t)\|^2$ by
minimizing over $s_i$ for fixed $t_i$ using calculus.  The claimed
expression for the integral follows from evaluating the inner integral to
get an expression involving the integral for $d-1$, and then using
induction.)
Using this and Lemma \ref{tail_bnd_G_multi},
\begin{align*}
&P\left(\inf_{S_k^{\alpha,\beta}} \mathbb{G}(s,t)+g(s,t)\le r\right)
\le P\left(\sup_{S_k^{\alpha,\beta}} |\mathbb{G}(s,t)|\ge C_1
k^{2+\alpha}\right)  \\
&\le 2\left\{3(k+1)[(k+1)^d/((k+1)^\beta \wedge 1)]+2\right\}^{2d}
\exp\left(-C C_1^2\frac{k^{4+2\alpha}}{(k+1)^\beta}\right).
\end{align*}
This is summable over $k$ if $4+2\alpha-\beta>0$.

Now, note that, since $\prod_i t_i\le (k+1)^d$ on $S_k$, we have, for any
$-\delta<\alpha_1<\alpha_2<\ldots<\alpha_{\ell-1}<\alpha_\ell=d$,
$S_k=S_k^L\cup S_k^{-\delta,\alpha_1}\cup
S_k^{\alpha_1,\alpha_2}\cup\ldots \cup
S_k^{\alpha_{\ell-1},\alpha_\ell}$.  If we choose $\delta<3/2$ and
$\alpha_i=i$ for $i\in\{1,\ldots,d\}$, the arguments above will show that
the probability of the infimum being less than or equal to $r$ over
$S_k^L$, $S_k^{-\delta,\alpha_1}$ and each $S_k^{\alpha_i,\alpha_{i+1}}$
is summable over $k$, so that $P\left(\inf_{S_k} \mathbb{G}(s,t)+g(s,t)\le
r\right)$ is summable over $k$, so setting $M$ so that the tail of this
sum past $M$ is less than $\varepsilon$ gives the desired result.
\end{proof}

The following theorem verifies condition (ii) of Lemma
\ref{inf_dist_lemma_multi} for the sequence of finite sample processes
$\mathbb{G}_{n,x_m,j}(s,t)+g_{n,x_m,j}(s,t)$ with $\eta/h_n\ge\|(s,t)\|$.
As explained above, the case where $\eta/h_n\le\|(s,t)\|$ is handled by a
separate argument.
\begin{theorem}\label{inf_bound_Gn}
Fix $m$ and $j$ with $j\in J(m)$.
There exists an $\eta>0$ such that for any $r<0$, $\varepsilon>0$, there
exists an $M$ and $N$ such that, for
all $n\ge N$,
\begin{align*}
P\left(\inf_{M<\|(s,t)\|\le \eta/h_n}
\mathbb{G}_{n,x_m,j}(s,t)+g_{n,x_m,j}(s,t)
\le r\right)<\varepsilon.
\end{align*}
\end{theorem}
\begin{proof}
Let $\mathbb{G}_n(s,t)=\mathbb{G}_{n,x_m,j}(s,t)$ and
$g_n(s,t)=g_{n,x_m,j}(s,t)$.
Let $\eta$ be small enough that the assumptions hold for $\|x-x_m\|\le
2\eta$ and that, for some constant $C_2$, $E(Y_{i,j}|X_i=x)\ge C_2 \|x-x_m\|^2$
for $\|x-x_m\|\le 2\eta$.  This implies that, for $\|(s,t)\|\le
h_n^{-1}\eta$,
\begin{align*}
&g_n(s,t)
\ge \frac{C_2}{h_n^{d+2}}
\int_{h_ns<x-x_m<h_n(s+t)}
\|x-x_m\|^2 f_X(x)
\, dx  \\
&\ge \frac{C_2\underline f}{h_n^{d+2}}
\int_{h_ns<x-x_m<h_n(s+t)}
\|x-x_m\|^2 \, dx
= C_2\underline f
\int_{s<x<s+t}
\|x\|^2
\, dx_d\cdots dx_1
\ge C_3 \|(s,t)\|^2 \prod_i t_i
\end{align*}
where $C_3$ is a constant that depends only on $\underline f$ and $d$
and the last inequality follows from bounding the integral as explained in
the proof of the previous theorem.

As in the proof of the previous theorem, let $S_k=\{k\le \|(s,t)\|\le
k+1\}$ and let $S_k^L=S_k\cap \{\prod_i t_i\le (k+1)^{-\delta}\}$ for some
fixed $\delta$.  We have, using Lemma \ref{tail_bound_Gn},
\begin{align*}
&P\left(\inf_{S_k^L} \mathbb{G}_n(s,t)+g_n(s,t)\le r\right)
\le P\left(\sup_{S_k^L} |\mathbb{G}_n(s,t)|\ge |r|\right)  \\
&\le 2\left\{3(k+1)[(k+1)^d/k^{-\delta}]+2\right\}^{2d}
\exp\left(-C\frac{|r|}{(k+1)^{-\delta/2}}\right)
\end{align*}
for %
$(k+1)^{-\delta}\ge n^{-4/(d+4)}(1+\log n)^2
\Longleftrightarrow k+1\le n^{4/[\delta(d+4)]}(1+\log n)^{-2/\delta}$
so, if $\delta<4$, this will
hold eventually for all $(k+1)\le h_n^{-1}\eta$ (once
$h_n^{-1}\eta\le n^{4/[\delta(d+4)]}(1+\log n)^{-2/\delta}
\Longleftrightarrow \eta\le n^{(4/\delta-1)/(d+4)}(1+\log n)^{-2/\delta}$).
The bound is summable over $k$ for any $\delta>0$.

Again following the proof of the previous theorem, for $\alpha<\beta$,
define $S_k^{\alpha,\beta}=S_k\cap \{(k+1)^\alpha < \prod_i t_i\le
(k+1)^\beta\}$.  We have, again using Lemma \ref{tail_bound_Gn},
\begin{align*}
&P\left(\inf_{S_k^{\alpha,\beta}} \mathbb{G}_n(s,t)+g_n(s,t)\le r\right)
\le P\left(\sup_{S_k^{\alpha,\beta}} |\mathbb{G}_n(s,t)|
\ge C_3 k^{2+\alpha}\right)  \\
&\le 2\left\{3(k+1)[(k+1)^d/(k^{\alpha}\wedge 1)]+2\right\}^{2d}
\exp\left(-C\frac{C_3 k^{2+\alpha}}{(k+1)^{\beta/2}}\right)
\end{align*}
for $(k+1)^\beta\ge n^{-4/(d+4)}$ (which will hold once the same
inequality holds for $\delta$ for $-\delta<\beta$) and $k+1\le
h_n^{-1}\eta$.
The bound is summable over $k$ for any $\alpha,\beta$ with
$4+2\alpha-\beta>0$.

Thus, noting as in the previous theorem that, for any
$-\delta<\alpha_1<\alpha_2<\ldots<\alpha_{\ell-1}<\alpha_\ell=d$,
$S_k=S_k^L\cup S_k^{-\delta,\alpha_1}\cup
S_k^{\alpha_1,\alpha_2}\cup\ldots \cup
S_k^{\alpha_{\ell-1},\alpha_\ell}$,
if we choose $\delta<3/2$ and $\alpha_i=i$ for $i\in\{1,\ldots,d\}$
the probability of the infimum being less than or equal to $r$ over the sets
indexed by $k$ for any $k\le h_n^{-1}\eta$ is bounded uniformly in $n$ by
a sequence that is summable over $k$ (once
$\eta\le n^{(4/\delta-1)/(d+4)}(1+\log n)^{-2/\delta}$).  Thus, if we
choose $M$ such that the tail of this sum past $M$ is less than
$\varepsilon$ and let $N$ be large enough so that
$\eta\le N^{(4/\delta-1)/(d+4)}(1+\log N)^{-2/\delta}$, we will have the
desired result.

\end{proof}

To complete the proof of Theorem \ref{inf_dist_thm_multi}, we need to show
that
\begin{align*}
&Z_{n,2}
\equiv \left(n^{(d+2)/(d+4)}
\inf_{\|(s-x_m,t)\|>\eta \text{ all $m$ s.t. $1\in J(m)$}}
E_nY_{i,1}I(s<X_i<s+t),\ldots,\right.  \\
&\left.n^{(d+2)/(d+4)}
\inf_{\|(s-x_m,t)\|>\eta \text{ all $m$ s.t. $k\in J(m)$}}
E_nY_{i,k}I(s<X_i<s+t)\right)
\stackrel{p}{\to} 0.
\end{align*}
This follows from the next two lemmas.

\begin{lemma}\label{large_s_lemma1_multi}
Under Assumptions \ref{smoothness_assump_multi} and
\ref{bdd_y_assump_multi}, 
for any $\eta>0$, there exists some
$\underline B>0$ such that
$EY_{i,j}I(s<X_i<s+t)\ge \underline B P(s<X_i<s+t)$ for all $(s,t)$ with
$\|(s-x_m,t)\|>\eta$ for all $m$ with $j\in J(m)$.

\end{lemma}
\begin{proof}
Given $\eta>0$, we can make $\eta$ smaller without weakening the result,
so let $\eta$ be small enough that $\|x_m-x_r\|_\infty>2\eta$ for all $m\ne
r$ with $j\in J(m)\cap J(r)$ and
$f_X$ satisfies $0<\underline f\le f_X(x)\le \overline f<\infty$ for some
$\overline f$ and $\underline f$ on
$\{x|\|x-x_m\|_\infty\le \eta\}$.  If
$\|(s-x_m,t)\|>\eta$, then $\|(s-x_m,s+t-x_m)\|_\infty>\eta/(4d)$, so
it suffices to show that $EY_{i,j}I(s<X_i<s+t)\ge \underline B
P(s<X_i<s+t)$ for all $(s,t)$ with $\|(s-x_m,s+t-x_m)\|_\infty>\eta/(4d)$.
Let $\underline\mu>0$ be such that $E(Y_{i,j}|X_i=x)>\underline\mu$ when
$\|x-x_m\|_\infty\ge \eta/(8d)$ for $m$ with $j\in J(m)$.  For notational
convenience, let $\delta=\eta/(4d)$.

For $m$ with $j\in J(m)$, let
$B(x_m,\delta)=\{x|\|x-x_m\|_\infty\le\delta\}$ and
$B(x_m,\delta/2)=\{x|\|x-x_m\|_\infty\le\delta/2\}$.  First, I show
that, for any $(s,t)$ with $\|(s-x_m,s+t-x_m)\|_\infty\ge \delta$,
$P(\{s<X_i<s+t\}\cap B(x_m,\delta)\backslash B(x_m,\delta/2))
\ge (1/3)(\underline f/\overline f)P(\{s<X_i<s+t\}\cap B(x_m,\delta/2))$.
Intuitively, this holds because, taking any box with a corner outside of
$B(x_m,\delta)$, this box has to intersect with a substantial proportion of
$B(x_m,\delta)\backslash B(x_m,\delta/2)$ in order to intersect with
$B(x_m,\delta/2)$.

Formally, we have $\{s<x<s+t\}\cap B(x_m,\delta)
=\{s\vee (x_m-\delta)<x<(s+t)\wedge (x_m+\delta)\}$, so that, letting
$\lambda$ be the Lebesgue measure on $\mathbb{R}^d$,
$\lambda(\{s<x<s+t\}\cap B(x_m,\delta))
=\prod_i [(s_i+t_i)\wedge (x_{m,i}+\delta)-s_i\vee (x_{m,i}-\delta)]$.
Similarly, $\lambda(\{s<x<s+t\}\cap B(x_m,\delta/2))
=\prod_i [(s_i+t_i)\wedge (x_{m,i}+\delta/2)-s_i\vee (x_{m,i}-\delta/2)]$.
For all $i$,
$[(s_i+t_i)\wedge (x_{m,i}+\delta/2)-s_i\vee (x_{m,i}-\delta/2)]
\le [(s_i+t_i)\wedge (x_{m,i}+\delta)-s_i\vee (x_{m,i}-\delta)]$.  For
some $r$, we must have $s_r\le x_{m,r}-\delta$ or $s_r+t_r\ge x_{m,r}+\delta$.
For this $r$, we will have
$[(s_r+t_r)\wedge (x_{m,r}+\delta/2)-s_r\vee (x_{m,r}-\delta/2)]
\le 2[(s_r+t_r)\wedge (x_{m,r}+\delta)-s_r\vee (x_{m,r}-\delta)]/3$.
Thus,
$\lambda(\{s<x<s+t\}\cap B(x_m,\delta/2))
\le 2\lambda(\{s<x<s+t\}\cap B(x_m,\delta))/3$.  It then follows that
$\lambda(\{s<x<s+t\}\cap B(x_m,\delta)\backslash B(x_m,\delta/2))
\ge (1/3)\lambda(\{s<x<s+t\}\cap B(x_m,\delta))$, so that
$P(\{s<x<s+t\}\cap B(x_m,\delta)\backslash B(x_m,\delta/2))
\ge (1/3)(\underline f/\overline f)P(\{s<x<s+t\}\cap B(x_m,\delta))$.

Now, we use the fact that $E(Y_{i,j}|X_i)$ is bounded away from zero
outside of $B(x_m,\delta/2)$, and that the proportion of $\{s<x<s+t\}$
that intersects with $B(x_m,\delta/2)$ can't be too large.  We have, for
any $(s,t)$ with $\|(s-x_m,s+t-x_m)\|_\infty\ge \delta$,
\begin{align*}
&EY_{i,j}I(s<X_i<s+t)
\ge \underline \mu P(\{s<X_i<s+t\}\backslash [\cup_mB(x_m,\delta/2)])  \\
&=  %
\underline \mu P(\{s<X_i<s+t\}\backslash [\cup_m B(x_m,\delta)])
+ \underline \mu \sum_m
P(\{s<X_i<s+t\}\cap B(x_m,\delta)\backslash B(x_m,\delta/2))  \\
&\ge \underline \mu P(\{s<X_i<s+t\}\backslash [\cup_m B(x_m,\delta)])
+ \underline \mu \sum_m
(1/3)(\underline f/\overline f)P(\{s<X_i<s+t\}\cap B(x_m,\delta))  \\
&\ge \underline \mu (1/3)(\underline f/\overline f) P(s<X_i<s+t)
\end{align*}
where the unions are taken over $m$ such that $j\in J(m)$.  The equality
in the second line follows because the sets $B(x_m,\delta)$ are disjoint.

\end{proof}

\begin{lemma}\label{large_s_lemma2_multi}
Let $S$ be any set in $\mathbb{R}^{2d}$ such that, for some $\underline
\mu>0$ and all $(s,t)\in S$,%
$EY_{i,j}I(s<X_i<s+t)\ge \underline \mu P(s<X_i<s+t)$.
Then, under Assumption \ref{bdd_y_assump_multi}, for any sequence
$a_n\to\infty$ and $r<0$,
\begin{align*}
\inf_{(s,t)\in S} \frac{n}{a_n\log n} E_nY_{i,j}I(s<X_i<s+t)> r
\end{align*}
with probability approaching 1.
\end{lemma}
\begin{proof}
For $(s,t)\in S$,
\begin{align*}
&\frac{n}{a_n\log n} E_nY_{i,j}I(s<X_i<s+t)\le r  \\
&\Longrightarrow \frac{n}{a_n\log n} (E_n-E)Y_{i,j}I(s<X_i<s+t)
\le r - \frac{n}{a_n\log n}EY_{i,j}I(s<X_i<s+t)  \\
&\le r - \frac{n}{a_n\log n}\underline \mu P(s<X_i<s+t)
\le -\left\{|r|\vee
  \left[\frac{n}{a_n\log n}\underline \mu P(s<X_i<s+t)\right]\right\}
  \\
&\Longrightarrow 
\left[\frac{\frac{a_n\log n}{n}}
{\frac{a_n\log n}{n}\vee P(s<X_i<s+t)}\right]^{1/2}
|(E_n-E)Y_{i,j}I(s<X_i<s+t)|  \\
&\ge \left[\frac{\frac{a_n\log n}{n}}
{\frac{a_n\log n}{n}\vee P(s<X_i<s+t)}\right]^{1/2}
\left\{\left[\frac{a_n\log n}{n}|r|\right]
\vee \left[\underline \mu P(s<X_i<s+t)\right]\right\}.
\end{align*}
If $\frac{a_n\log n}{n}\ge P(s<X_i<s+t)$, then the last line
is greater than or equal to $\frac{a_n\log n}{n} |r|$.  If
$\frac{a_n\log n}{n}\le P(s<X_i<s+t)$, the last line is
greater than or equal to
$\left[\frac{\frac{a_n\log n}{n}}{P(s<X_i<s+t)}\right]^{1/2}
\underline \mu P(s<X_i<s+t)
=\left(\frac{a_n\log n}{n}\right)^{1/2}
\underline \mu \sqrt{P(s<X_i<s+t)}
\ge \underline \mu \frac{a_n\log n}{n}$.  Thus,
\begin{align*}
&P\left(\inf_{(s,t)\in S}
\frac{n}{a_n\log n} E_nY_{i,j}I(s<X_i<s+t)\le r\right)  \\
&\le P\left(\sup_{(s,t)\in S}
\left[\frac{\frac{a_n\log n}{n}}
{\frac{a_n\log n}{n}\vee P(s<X_i<s+t)}\right]^{1/2}
|(E_n-E)Y_{i,j}I(s<X_i<s+t)|
\ge (|r|\wedge \underline \mu) \frac{a_n\log n}{n}
\right).
\end{align*}
This converges to zero by Theorem 37 in \citet{pollard_convergence_1984}
with, in the notation of that theorem, $\mathcal{F}_n$ the class of
functions of the form
\begin{align*}
\left[\frac{\frac{a_n\log n}{n}}
{\overline Y^2\frac{a_n\log n}{n}\vee P(s<X_i<s+t)}\right]^{1/2}
Y_{i,j}I(s<X_i<s+t)
\end{align*}
with $(s,t)\in S$,
$\delta_n=\left(\frac{n}{a_n\log n}\right)^{1/2}$ and
$\alpha_n=1$.  To verify the conditions of the lemma, the covering number
bound holds because each $\mathcal{F}_n$ is contained in the larger class
$\mathcal{F}$ of functions of the form $wY_{i,j}I(s<X_i<s+t)$ where
$(s,t)$ ranges over $S$ and $w$ ranges over $\mathbb{R}$, and this larger
class is a VC subgraph class.  The supremum bound on functions in
$\mathcal{F}_n$ holds by Assumption \ref{bdd_y_assump_multi}.  To verify
the bound on the $L^2$ norm of functions in $\mathcal{F}_n$, note that
\begin{align*}
&E\left\{\left[\frac{\frac{a_n\log n}{n}}
{\overline Y^2\frac{a_n\log n}{n}\vee P(s<X_i<s+t)}\right]^{1/2}
Y_{i,j}I(s<X_i<s+t)\right\}^2  \\
&\le \frac{\frac{a_n\log n}{n}}
{\frac{a_n\log n}{n}\vee P(s<X_i<s+t)}
P(s<X_i<s+T)
\le \frac{a_n\log n}{n} = \delta_n^2
\end{align*}
since $ab/(a\vee b)\le a$ for any $a,b>0$.

\end{proof}

By Lemma \ref{large_s_lemma1_multi}, $\{\|(s-x_m,t)\|>\eta \text{ all
$m$ s.t. $j\in J(m)$}\}$ satisfies the conditions of Lemma
\ref{large_s_lemma2_multi}, so $E_nY_{i,j}I(s<X_i<s+t)$ converges to zero
at a $n/(a_n\log n)$ rate for any $a_n\to \infty$, which can be made
faster than the $n^{(d+2)/(d+4)}$ rate needed to show that
$Z_{n,2}\stackrel{p}{\to} 0$.  This completes the proof of Theorem
\ref{inf_dist_thm_multi}.

\subsection*{Inference}

I use the following lemma in the proof of Theorem \ref{abs_cont_thm}

\begin{lemma}\label{zero_var_atom_lemma_multi}
Let $\mathbb{H}$ be a Gaussian random process with sample paths that are
almost surely in the set $C(\mathbb{T},\mathbb{R}^k)$ of continuous
functions with respect to %
some semimetric
on the index set $\mathbb{T}$ with a countable dense subset
$\mathbb{T}_0$.
Then, for any set $A\in\mathbb{R}^k$ with Lebesgue measure zero,
$P(\inf_{t\in\mathbb{T}} \mathbb{H}(t)\in A)
\le P(\inf_{t\in\mathbb{T}, \det var(\mathbb{H}(t))<\varepsilon}
\mathbb{H}(t)\in A \text{ for all $\varepsilon>0$})$.
\end{lemma}
\begin{proof}
First, note that, if the infimum over $\mathbb{T}$
is in $A$, then, since $\{t\in\mathbb{T}|\det var(\mathbb{H}(t))\ge
\varepsilon\}$ and $\{t\in\mathbb{T}|\det var(\mathbb{H}(t))<\varepsilon\}$
partition $T$, the infimum over one of these sets must be in
$A$.  By Proposition 3.2 in \citet{pitt_local_1979},
the infimum of $\mathbb{H}(t)$ over the former set has a distribution that
is continuous with respect to the Lebesgue measure, so the probability of
the infimum of $\mathbb{H}(t)$ over this set being in $A$ is zero.
Thus, $P\left(\inf_{t\in\mathbb{T}} \mathbb{H}(t) \in A\right)\le
P\left(\inf_{t\in\mathbb{T}, \det var(\mathbb{H}(t))<\varepsilon}
\mathbb{H}(t)\in A\right)$.  Taking $\varepsilon$ to zero along a
countable sequence gives the result.

\end{proof}

\begin{proof}[Proof of Theorem \ref{abs_cont_thm}]
For $m$ from $1$ to $\ell$, let $\{j_{m,1},\ldots,j_{m,|J(m)|}\}=J(m)$.
Then, letting
\begin{align*}
\tilde Z
\equiv &(\inf_{s,t}\mathbb{G}_{P,x_1,j_{1,1}}(s,t)+g_{P,x_1,j_{1,1}}(s,t),\ldots,
\inf_{s,t}\mathbb{G}_{P,x_1,j_{1,|J(1)|}}(s,t)+g_{P,x_1,j_{1,|J(1)|}}(s,t),\ldots,  \\
&\inf_{s,t}\mathbb{G}_{P,x_\ell,j_{\ell,1}}(s,t)+g_{P,x_\ell,j_{\ell,1}}(s,t),\ldots,
\inf_{s,t}\mathbb{G}_{P,x_\ell,j_{\ell,|J(\ell)|}}(s,t)
+g_{P,x_\ell,j_{\ell,|J(\ell)|}}(s,t)),
\end{align*}
each element of $Z$ is the minimum of the elements of some subvector of
$\tilde Z$, where the subvectors corresponding to different elements of
$Z$ do not overlap. Thus, it suffices to show that $\tilde Z$ has an
absolutely continuous distribution.
For this, it suffices to show that, for each $m$,
\begin{align*}
(\inf_{s,t}\mathbb{G}_{P,x_m,j_{m,1}}(s,t)+g_{P,x_m,j_{m,1}}(s,t),\ldots,
\inf_{s,t}\mathbb{G}_{P,x_m,j_{m,|J(m)|}}(s,t)+g_{P,x_m,j_{m,|J(m)|}}(s,t))
\end{align*}
has an absolutely continuous distribution, since these are independent
across $m$.

To this end, fix $m$ and let $\mathbb{H}(s,t)$ be the random process with
sample paths
in $C(\mathbb{R}^{2d},\mathbb{R}^{|J(m)|})$ defined by
\begin{align*}
\mathbb{H}(s,t)
=(\mathbb{G}_{P,x_m,j_{m,1}}(s,t)+g_{P,x_m,j_{m,1}}(s,t),\ldots,
\mathbb{G}_{P,x_m,j_{m,|J(m)|}}(s,t)+g_{P,x_m,j_{m,|J(m)|}}(s,t)).
\end{align*}
By Assumption \ref{inv_mat_assump}
$var(\mathbb{H}(s,t))=M \prod_i{t_i}$ for some positive definite matrix
$M$, so that $\det var(\mathbb{H}(s,t))
=(\det M) \left(\prod_i{t_i}\right)^{|J(m)|}$.  Thus,
$\inf_{(s,t)\in\mathbb{R}^{2d}, \det var(\mathbb{H}(s,t))<\varepsilon}
\mathbb{H}(s,t)\in A \text{ for all $\varepsilon>0$}$ iff.
$\inf_{(s,t)\in\mathbb{R}^{2d}, \prod_i t_i<\varepsilon}
\mathbb{H}(s,t)\in A \text{ for all $\varepsilon>0$}$ so, by Lemma
\ref{zero_var_atom_lemma_multi},
$P(\inf_{(s,t)\in\mathbb{R}^{2d}} \mathbb{H}(s,t)\in A)
\le P(\inf_{(s,t)\in\mathbb{R}^{2d}, \prod_i t_i<\varepsilon}
\mathbb{H}(s,t)\in A \text{ for all $\varepsilon>0$})$.  For each $j$,
$\prod_it_i$ is equal to $var(\mathbb{H}_j(s,t))=\rho_j(0,(s,t))$ times some
constant, where $\rho_j$ is the covariance semimetric for component $j$
given by
$\rho_j((s,t),(s',t'))=var(\mathbb{H}_j(s,t)-\mathbb{H}_j(s',t'))$.  Thus,
there exists a constant $C$ such that $\prod_i t_i\le \varepsilon$
implies $\rho_j(0,(s,t))< C\varepsilon$ for all $j$, so that
$P(\inf_{(s,t)\in\mathbb{R}^{2d}} \mathbb{H}(s,t)\in A)
\le P(\inf_{(s,t)\in\mathbb{R}^{2d}, \rho_j(0,(s,t))< C\varepsilon \text{ all $j$}}
\mathbb{H}(s,t)\in A \text{ for all $\varepsilon>0$})$.

Since the sample paths of $\mathbb{H}$ are almost surely continuous with
respect to the semimetric $\max_j \rho_j((s,t),(s',t'))$
on the set $\|(s,t)\|\le M$ for any finite $M$,
$\inf_{\|(s,t)\|\le M, \rho_j(0,(s,t))< C\varepsilon \text{ all $j$}}
\mathbb{H}(s,t)\in A \text{ for all $\varepsilon>0$}$
implies that $\mathbb{H}(0)=0$ is a limit point of $A$ on this probability
one set.  Thus, for any set $A$ that does not have zero as a limit point,
$P(\inf_{\|(s,t)\|\le M} \mathbb{H}(s,t)\in A)=0$ for any finite $M$.
Applying this to $A\backslash B_\eta(0)$ where $B_\eta(0)$ is the
$\eta$-ball around $0$ in $\mathbb{R}^{|J(m)|}$, we have
\begin{align*}
&P\left(\inf_{(s,t)\in\mathbb{R}^{2d}} \mathbb{H}(s,t)\in A\right)
= P\left(\inf_{(s,t)\in\mathbb{R}^{2d}}
\mathbb{H}(s,t)\in A\cap B_\eta(0)\right)
+P\left(\inf_{(s,t)\in\mathbb{R}^{2d}}
\mathbb{H}(s,t)\in A\backslash B_\eta(0)\right)  \\
&\le P\left(\inf_{(s,t)\in\mathbb{R}^{2d}}
\mathbb{H}(s,t)\in A\cap B_\eta(0)\right)
+P\left(\inf_{\|(s,t)\|\le M}
\mathbb{H}(s,t)\in A\backslash B_\eta(0)\right)  \\
&+P\left(\inf_{\|(s,t)\|> M}
\mathbb{H}(s,t)\in A\backslash B_\eta(0)\right)  \\
&= P\left(\inf_{(s,t)\in\mathbb{R}^{2d}}
\mathbb{H}(s,t)\in A\cap B_\eta(0)\right)
+P\left(\inf_{\|(s,t)\|> M}
\mathbb{H}(s,t)\in A\backslash B_\eta(0)\right).
\end{align*}
Noting that $P\left(\inf_{\|(s,t)\|> M}
\mathbb{H}(s,t)\in A\backslash B_\eta(0)\right)$ can be made arbitrarily
small by making $M$ large, this shows that
$P\left(\inf_{(s,t)\in\mathbb{R}^{2d}} \mathbb{H}(s,t)\in A\right)
  =P\left(\inf_{(s,t)\in\mathbb{R}^{2d}}
     \mathbb{H}(s,t)\in A\cap B_\eta(0)\right)$
Taking $\eta$ to zero along a countable sequence, this shows that
$P\left(\inf_{(s,t)\in\mathbb{R}^{2d}} \mathbb{H}(s,t)\in A\right)
\le P\left(\inf_{(s,t)\in\mathbb{R}^{2d}}
\mathbb{H}(s,t)\in A\cap \{0\}\right)$ so that $\inf_{(s,t)\in\mathbb{R}^{2d}}
\mathbb{H}(s,t)$ has an absolutely continuous distribution with a possible
atom at zero.

To show that there can be no atom at zero, we argue as follows.
Fix $j\in J(m)$.  The component of $\mathbb{H}$ corresponding to this $j$
is $\mathbb{G}_{P,x_m,j}(s,t)+g_{P,x_m,j}(s,t)$.
For some constant $K$, for any $k\ge
0$, letting $s_{i,k}=(i/k,0,\ldots,0)$ and $t_k=(1/k,1,\ldots,1)$, we
will have $g_{P,x_m,j}(s_{i,k},t_k)\le K/k$ for $i\le k$, so that
\begin{align*}
&P\left(\inf_{(s,t)\in\mathbb{R}^{2d}}
  \mathbb{G}_{P,x_m,j}(s,t)+g_{P,x_m,j}(s,t)=0\right)
=P\left(\inf_{(s,t)\in\mathbb{R}^{2d}}
\mathbb{G}_{P,x_m,j}(s,t)+g_{P,x_m,j}(s,t)\ge 0\right)  \\
&\le P\left(\mathbb{G}_{P,x_m,j}(s_{i,k},t_k)+g_{P,x_m,j}(s_{i,k},t_k)
\ge 0 \text{ all $i\in\{0,\ldots,k\}$}\right)  \\
&\le P\left(\mathbb{G}_{P,x_m,j}(s_{i,k},t_k)+K/k\ge 0 \text{ all
  $i\in\{0,\ldots,k\}$}\right)  \\
&=P\left(\sqrt{k}\mathbb{G}_{P,x_m,j}(s_{i,k},t_k)+K/\sqrt{k}\ge 0
  \text{ all $i\in\{0,\ldots,k\}$}\right)  \\
&=P\left(\mathbb{G}_{P,x_m,j}(s_{i,1},t_1)+K/\sqrt{k}\ge 0
  \text{ all $i\in\{0,\ldots,k\}$}\right).
\end{align*}
The final line is the probability of $k+1$ iid normal random variables
each being greater than or equal to $-K/\sqrt{k}$, which can be made
arbitrarily small by making $k$ large.
\end{proof}

\begin{proof}[proof of Theorem \ref{subsamp_thm}]
This follows immediately from the continuity of the asymptotic
distribution \citep[see][]{politis_subsampling_1999}.
\end{proof}

\begin{proof}[proof of Theorem \ref{Z_hat_thm}]
It suffices to show that, for every subsequence, there exists a further
subsequence along which the distribution of $\hat Z$ converges weakly to
the distribution of $Z$.  Given a subsequence, let the further subsequence
be such that the convergence in probability in
Assumption \ref{avar_est_assump} is with probability one.

For any fixed $B>0$, the processes
\begin{align*}
\left[\hat{\mathbb{G}}_{P,x_k}(s,t)+\hat g_{P,x_k}(s,t)\right]
I(\|(s,t)\|\le B_n)
\end{align*}
are, along this subsequence, Gaussian processes with mean functions and
covariance kernels converging with probability one to those of the
distribution being estimated uniformly in $\|(s,t)\|\le B$.  Thus, with
probability one, the distributions of these processes converge weakly to
the distribution of the process being estimated along this subsequence
taken as random processes on $\|(s,t)\|\le B$.  Thus, to get the weak
convergence of the elementwise infimum, we just need to verify part (ii)
of Lemma \ref{inf_dist_lemma_multi}.  To this end, note that, along the
further subsequence, the infimum of
\begin{align*}
\left[\hat{\mathbb{G}}_{P,x_k,j}(s,t)+\hat g_{P,x_k,j}(s,t)\right]
I(\|(s,t)\|\le B_n)
\end{align*}
is eventually bounded from below (in the stochastic dominance sense) by
the infimum of a process defined the same way as
\begin{align*}
\mathbb{G}_{P,x_k,j}(s,t)+g_{P,x_k,j}(s,t),
\end{align*}
but with $E(m_{J(k)}(W_i,\theta)m_{J(k)}(W_i,\theta)'|X=x_k)$
replaced by $2E(m_{J(k)}(W_i,\theta)m_{J(k)}(W_i,\theta)'|X=x_k)$, and
$V(x_k)$ replaced by $V(x_k)/2$.
Once $n$ is large enough that this holds along this further subsequence,
part (ii) of  Lemma \ref{inf_dist_lemma_multi} will hold by Lemma
\ref{inf_bound_G} applied to this process.
\end{proof}

\begin{proof}[proof of Corollary \ref{Z_hat_cor}]
By Theorem \ref{Z_hat_thm}, the distribution of $S(\hat Z)$ converges
weakly conditionally in probability to the distribution of $S(Z)$,
and by Theorem \ref{inf_dist_thm_multi},
$n^{(d_X+2)/(d_X+4)}S(T_n(\theta))\stackrel{d}{\to} S(Z)$.  $S(Z)$ has a
continuous distribution by Theorem \ref{abs_cont_thm}, so the result
follows by standard arguments.
\end{proof}

\subsection*{Other Shapes of the Conditional Mean}

This section contains the proofs of the results in Section
\ref{inf_dist_alpha_sec}, which extend the results of Section
\ref{inf_dist_sec} to other shapes of the conditional mean.  First, I
show how Assumption \ref{smoothness_assump_multi} implies Assumption
\ref{smoothness_assump_alpha} with $\gamma=2$.
Next, I prove Theorem \ref{deriv_alpha_thm}, which gives an
interpretation of Assumption \ref{inf_dist_thm_alpha} in terms of
conditions on the number of bounded derivatives in the one dimensional
case.  Finally, I prove Theorem
\ref{inf_dist_thm_alpha}, which derives the asymptotic distribution of the
KS statistic under these assumptions.  The proof is mostly the same as the
proof of Theorem \ref{inf_dist_thm_multi}, and I present only the parts of
the proof that differ, referring to the proof of Theorem
\ref{inf_dist_thm_multi} for the parts that do not need to be changed.

To see that, under part (ii) from Assumption
\ref{smoothness_assump_multi}, Assumption \ref{smoothness_assump_alpha}
will hold with $\gamma=2$, note that, by a second order Taylor expansion,
for some $x^*(x)$ between $x$ and $x_k$,
\begin{align*}
\frac{\bar m_j(\theta,x)-\bar m_j(\theta,x_k)}{\|x-x_k\|^2}
=\frac{(x-x_k)V_j(x^*(x))(x-x_k)}{2\|x-x_k\|^2}
=\frac{1}{2}\frac{x-x_k}{\|x-x_k\|}V_j(x^*(x))\frac{x-x_k}{\|x-x_k\|}.
\end{align*}
Thus, letting
$\psi_{j,k}(t)=\frac{1}{2}t V_j(x_k)t$
we have
\begin{align*}
& \sup_{\|x-x_k\|\le \delta} \left\|
  \frac{\bar m_j(\theta,x)-\bar m_j(\theta,x_k)}{\|x-x_k\|^2}
  -\psi_{j,k}\left(\frac{x-x_k}{\|x-x_k\|}\right) \right\| \\
&=\sup_{\|x-x_k\|\le \delta} \left\|
  \frac{1}{2}\frac{x-x_k}{\|x-x_k\|}V_j(x^*(x))\frac{x-x_k}{\|x-x_k\|}
  -\frac{1}{2}\frac{x-x_k}{\|x-x_k\|}V_j(x_k)\frac{x-x_k}{\|x-x_k\|}
    \right\|.
\end{align*}
This goes to zero as $\delta\to 0$ by the continuity of the second
derivative matrix.

The proof of Theorem \ref{deriv_alpha_thm} below shows that, in the one
dimensional case, Assumption \ref{smoothness_assump_multi} follows more
generally from conditions on higher order derivatives.

\begin{proof}[proof of Theorem \ref{deriv_alpha_thm}]
It suffices to consider the case where $d_Y=1$.
First, suppose that $\mathcal{X}_0$ has infinitely many elements.  Let
$\{x_k\}_{k=1}^\infty$ be a nonrepeating sequence of elements in
$\mathcal{X}_0$.  Since $\mathcal{X}_0$ is compact, this sequence must
have a subsequence that converges to some $\tilde x\in\mathcal{X}_0$.  If
$\bar m(\theta,x)$ had a nonzero $r$th derivative at $\tilde x$ for some
$r<p$, then, by Lemma \ref{min_p_deriv_lemma} below, $\bar m(\theta,x)$
would be strictly greater than $\bar m(\theta,\tilde x)$ for $x$ in some
neighborhood of $\tilde x$, a contradiction.  Thus, a $p$th order taylor
expansion gives, using the notation $D_r(x)=\delta^r/\delta x^r \bar
m(\theta,x)$ for $r\le p$,
$\bar m(\theta,x)-\bar m(\theta,\tilde x)= D_p(x^*(x))(x-\tilde x)^p/p!
  \le \bar D |x-\tilde x|^p/p!$
where $\bar D$ is a bound on the $p$th derivative and $x^*(x)$ is some
value between $x$ and $\tilde x$.

If $\mathcal{X}_0$ has finitely many elements, then, for each $x_0\in
\mathcal{X}_0$, a $p$th order Taylor expansion gives
$\bar m(\theta,x)-\bar m(\theta,x_0)
  =D_1(x_0)(x-x_0)+\frac{1}{2}D_2(x_0)(x-x_0)^2
    +\cdots+\frac{1}{p!}D_p(x^*(x))(x-x_0)^p$.
If, for some $r<p$, $D_r(x_0)\ne 0$ and $D_{r'}(x_0)=0$ for $r'<r$, then
Assumption \ref{smoothness_assump_alpha} will hold at $x_0$ with
$\gamma=r$.
If not, we will have $\bar m(\theta,x)-\bar m(\theta,x_0)\le \bar D
|x-x_0|^p/p!$ for all $x$.
\end{proof}

\begin{lemma}\label{min_p_deriv_lemma}
Suppose that
$g:[\underline x,\overline x]\subseteq \mathbb{R}\to \mathbb{R}$ is
minimized at some $x_0$.  If the least nonzero derivative of $g$ is
continuous at $x_0$, then, for some $\varepsilon>0$, $g(x)>g(x_0)$ for
$|x-x_0|\le \varepsilon$, $x\ne x_0$.
\end{lemma}
\begin{proof}
Let $p$ be the least integer such that the $p$th derivative
$g^{(p)}(x_0)$ is nonzero.  By a $p$th
order Taylor expansion, $g(x)-g(x_0)=g^{(p)}(x^*(x))(x-x_0)^p$ for some
$x^*(x)$ between $x$ and $x_0$.  By continuity of $g^{(p)}(x)$,
$|g^{(p)}(x^*(x))-g^{(p)}(x_0)|>|g^{(p)}(x_0)|/2$ for $x$ close enough to
$x_0$, so that $g(x)-g(x_0)=g^{(p)}(x^*(x))(x-x_0)^p\ge
|g^{(p)}(x_0)|/2|x-x_0|^p>0$ (the $p$th derivative must have the same sign
as $x-x_0$ if $p$ is odd in order for $g$ to be minimized at $x_0$).
\end{proof}

I now prove Theorem \ref{inf_dist_thm_alpha}.  I prove the theorem under
the assumption that $\gamma(j,k)=\gamma$ for all $(j,k)$ with
$j\in J(k)$.  The general case follows from applying the argument to
neighborhoods of each $x_k$, and getting faster rates of convergence for
$(j,k)$ such that $\gamma(j,k)<\gamma$.
The proof is the same as the proof of Theorem \ref{inf_dist_thm_multi}
with the following modifications.
First, Theorem \ref{local_process_thm_multi} must be modified to the
following theorem, with the new definition of $g_{P,x_k,j}(s,t)$.

\begin{theorem}\label{local_process_thm_alpha}
Let $h_n=n^{-\beta}$ for some $0<\beta<1/d_X$.  Let
\begin{align*}
\mathbb{G}_{n,x_m}(s,t)
=\frac{\sqrt{n}}{h_n^{d/2}}(E_n-E)Y_{i,J(m)}I(h_ns<X_i-x_m<h_n(s+t))
\end{align*}
and let $g_{n,x_m}(s,t)$ have $j$th element
\begin{align*}
g_{n,x_m,j}(s,t)=\frac{1}{h_n^{d_X+\gamma}}EY_{i,j}I(h_ns<X_i-x_m<h_n(s+t))
\end{align*}
if $j\in J(m)$ and zero otherwise.
Then, for any finite $M$,
$(\mathbb{G}_{n,x_1}(s,t),\ldots,\mathbb{G}_{n,x_\ell}(s,t))
\stackrel{d}{\to}
(\mathbb{G}_{P,x_1}(s,t),\ldots,\mathbb{G}_{P,x_\ell}(s,t))$
taken as random
processes on $\|(s,t)\|\le M$ with the supremum norm
and $g_{n,x_m}(s,t)\to g_{P,x_m}(s,t)$ uniformly in $\|(s,t)\|\le M$
where $\mathbb{G}_{P,x_m}(s,t)$ and $g_{P,x_m}(s,t)$ are defined as in
Theorem \ref{inf_dist_thm_multi} for $m$ from $1$ to $\ell$.
\end{theorem}
\begin{proof}
The proof of the first display is the same.  For the proof of the claim
regarding $g_{n,x_m}(s,t)$, we have
\begin{align*}
g_{n,x_m,j}(s,t)&=\frac{1}{h_n^{d_X+\gamma}}
\int_{h_ns<x-x_m<h_n(s+t)} \psi_{j,k}\left(\frac{x-x_m}{\|x-x_m\|}\right)
  \|x-x_m\|^\gamma f_X(x_m)\, dx  \\
&+\frac{1}{h_n^{d_X+\gamma}}\int_{h_ns<x-x_m<h_n(s+t)}
  \psi_{j,k}\left(\frac{x-x_m}{\|x-x_m\|}\right)\|x-x_m\|^\gamma
    [f_X(x)-f_X(x_m)]\, dx  \\
&+\frac{1}{h_n^{d_X+\gamma}}\int_{h_ns<x-x_m<h_n(s+t)}
  \left[\bar m_j(\theta,x)-\bar m_j(\theta,x_m)  \right.  \\
& \left.-\psi_{j,k}\left(\frac{x-x_m}{\|x-x_m\|}\right)\|x-x_m\|^\gamma\right]
   f_X(x)\, dx.
\end{align*}
The first term is equal to $g_{P,x_m,j}(s,t)$ by a change of variable $x$
to $h_nx+x_m$ in the integral.  The second term is bounded by
$g_{P,x_m,j}(s,t)\sup_{\|x-x_m\|\le 2h_n M} [f_X(x)-f_X(x_m)]/f_X(x_m)$,
which goes to zero uniformly in $\|(s,t)\|\le M$ by continuity of $f_X$.
The third term is equal to (using the same change of variables)
\begin{align*}
&\int_{s<x<s+t}
  \left[\frac{\bar m_j(\theta,h_nx+x_m)-\bar m_j(\theta,x_m)}{h_n^\gamma}
  -\psi_{j,k}\left(\frac{x}{\|x\|}\right)\|x\|^\gamma\right]
   f_X(x)\, dx  \\
&=\int_{s<x<s+t}
  \|x\|^\gamma
  \left[\frac{\bar m_j(\theta,h_nx+x_m)-\bar m_j(\theta,x_m)}{\|h_nx\|^\gamma}
  -\psi_{j,k}\left(\frac{x}{\|x\|}\right)\right]
   f_X(x)\, dx.
\end{align*}
For $\|(s,t)\|\le M$, this is bounded by a constant times
\begin{align*}
\sup_{\|x\|\le 2M} 
\left\|\frac{\bar m_j(\theta,h_nx+x_m)-\bar m_j(\theta,x_m)}{\|h_nx\|^\gamma}
  -\psi_{j,k}\left(\frac{x}{\|x\|}\right)\right\|,
\end{align*}
which goes to zero as $n\to\infty$ by Assumption
\ref{smoothness_assump_alpha}.
\end{proof}

The drift term and the mean zero term will be of the same order of
magnitude if
$\sqrt{n}/h_n^{d_X/2}=1/h_n^{d_X+\gamma}
\Leftrightarrow h_n=n^{-1/(d_X+2\gamma)}$, so that
\begin{align*}
&n^{(d_X+\gamma)/(d+2\gamma)}(E_nY_{i,J(1)}I(h_ns<X-x_1<h_n(s+t)),
\ldots,E_nY_{i,J(\ell)}I(h_ns<X-x_\ell<h_n(s+t))  \\
&=(\mathbb{G}_{n,x_1}(s,t)+g_{n,x_1}(s,t),
\ldots,\mathbb{G}_{n,x_\ell}(s,t)+g_{n,x_\ell}(s,t))  \\
&\stackrel{d}{\to}
(\mathbb{G}_{P,x_1}(s,t)+g_{P,x_1}(s,t),\ldots,
\mathbb{G}_{P,x_m}(s,t)+g_{P,x_m}(s,t))
\end{align*}
taken as stochastic processes in $\{\|(s,t)\|\le M\}$ with the supremum
norm.  From now on, let $h_n=n^{-1/(d+2\gamma)}$ so that this will hold.

Lemmas \ref{tail_bnd_G_multi} and \ref{tail_bound_Gn} hold as stated,
except for the condition in Lemma \ref{tail_bound_Gn} that
$\varepsilon\ge n^{-4/(d+4)}(1+\log n)^2$
must be replaced by
$\varepsilon\ge n^{2\gamma/(d+2\gamma)}(1+\log n)^2$
so that $h_n^{d/2}2^{d/2}\varepsilon^{1/2}\ge n^{-1/2}(1+\log n)$,
which implies the fourth inequality in the last display in the proof of
this lemma, holds for the sequence $h_n$ in the general case.

The next part of the proof that needs to be modified is the proofs of
Theorems \ref{inf_bound_G} and \ref{inf_bound_Gn}.  For this, note that,
for some constants $C_1$ and $\eta>0$
\begin{align}\label{g_ineq}
g_{P,x_m,j}(s,t)\ge C_1 \|(s,t)\|^\gamma \prod_i t_i
\end{align}
and, for %
$\|(s,t)\|\le \eta/h_n$,
\begin{align}\label{gn_ineq}
g_{n,x_m,j}(s,t)\ge C_1 \|(s,t)\|^\gamma \prod_i t_i
\end{align}
for all $m$ and $j$.  To see this, note that
\begin{align*}
&g_{n,x_m,j}(s,t)=E\frac{1}{h_n^{d_X+\gamma}}EY_{i,j}I(h_ns<X_i-x_m<h_n(s+t)) \\
&=\frac{1}{h_n^{d_X+\gamma}}
  \int_{h_ns<x-x_m<h_n(s+t)} \bar m(\theta,x) f_X(x) \, dx
=\int_{s<x<s+t} \frac{\bar m(\theta,h_nx+x_m)}{\|h_nx\|^\gamma}\|x\|^\gamma
  f_X(h_nx+x_m) \, dx
\end{align*}
where the last equality follows from the change of variables $x$ to
$h_nx+x_m$.  For small enough $\eta$, this is greater than or equal to
$\frac{1}{2}\int_{s<x<s+t} \underline \psi \|x\|^\gamma f_X(x_m) \, dx$
for $\|(s,t)\|\le \eta/h_n$ by Assumption \ref{smoothness_assump_alpha}
and the continuity of $f_X$.  By definition, $g_{P,x_m,j}(s,t)$ is also
greater than or equal to a constant times
$\int_{s<x<s+t} \|x\|^\gamma \, dx$.  To see that this is greater than or
equal to a constant times $\|(s,t)\|^\gamma\prod_i t_i$, note that the
Euclidean norm is equivalent to the norm
$(s,t)\mapsto \max_i \max\{|s_i|,|s_i+t_i|\}$ and let $i^*$ be
an index such that $|s_{i^*}|=\max_i \max\{|s_i|,|s_i+t_i|\}$ or
$|s_{i^*}+t_{i^*}|=\max_i \max\{|s_i|,|s_i+t_i|\}$.  In the former case,
we will have $\|x\|\ge |s_{i^*}|/2$ for $x$ on the set
$\{s_{i^*}\le x_{i^*}\le s_{i^*}+|s_{i^*}|/2\}\cap \{s<x<s+t\}$, which has
Lebesgue measure
$\left(\prod_{i\ne i^*}t_i\right)\cdot |s_{i^*}|/2
  \ge \left(\prod_{i\ne i^*}t_i\right)\cdot t_{i^*}/4$,
so that
$\int_{s<x<s+t} \|x\|^\gamma \, dx
  \ge \left(\max_i \max\{|s_i|,|s_i+t_i|\}/2\right)^\gamma\prod_i t_i/4$,
and a symmetric argument holds in the latter case.

With these inequalities in hand, the modified proofs of Theorems
\ref{inf_bound_G} and \ref{inf_bound_Gn} are as follows.

\begin{proof}[proof of Theorem \ref{inf_bound_G} for general case]
Let $\mathbb{G}(s,t)=\mathbb{G}_{P,x_m,j}(s,t)$ and
$g(s,t)=g_{P,x_m,j}(s,t)$.
Let $S_k=\{k\le \|(s,t)\|\le k+1\}$ and let $S_k^L=S_k\cap \{\prod_i
t_i\le (k+1)^{-\delta}\}$ for some fixed $\delta$.  By
Lemma \ref{tail_bnd_G_multi},
\begin{align*}
&P\left(\inf_{S_k^L} \mathbb{G}(s,t)+g(s,t)\le r\right)
\le P\left(\sup_{S_k^L} |\mathbb{G}(s,t)|\ge |r|\right)  \\
&\le \left\{3(k+1)[(k+1)^d/k^{-\delta}]+2\right\}^{2d}
\exp\left(-C r^2 (k+1)^{\delta}\right)
\end{align*}
for $k$ large enough where $C$ depends only on $d$.  Thus, the infimum
over each $S_k^L$ is summable over $k$.

For any ${\underline \beta}$ and ${\overline \beta}$ with
${\underline \beta}<{\overline \beta}$, let
$S_k^{{\underline \beta},{\overline \beta}}
=S_k\cap \{(k+1)^{\underline \beta}<\prod_i t_i\le (k+1)^{\overline \beta}\}$.
Using Lemma \ref{tail_bnd_G_multi} and (\ref{g_ineq}),
\begin{align*}
&P\left(\inf_{S_k^{{\underline \beta},{\overline \beta}}} \mathbb{G}(s,t)+g(s,t)\le r\right)
\le P\left(\sup_{S_k^{{\underline \beta},{\overline \beta}}}
  |\mathbb{G}(s,t)| \ge C_1 k^{\gamma+{\underline \beta}}\right)  \\
&\le \left\{3(k+1)[(k+1)^d/((k+1)^{\overline \beta} \wedge 1)]+2\right\}^{2d}
\exp\left(-C C_1^2\frac{k^{2\gamma+2{\underline \beta}}}
  {(k+1)^{\overline \beta}}\right).
\end{align*}
This is summable over $k$ if
$2\gamma+2{\underline \beta}-{\overline \beta}>0$.

Now, note that, since $\prod_i t_i\le (k+1)^d$ on $S_k$, we have, for any
$-\delta<\beta_1<\beta_2<\ldots<\beta_{\ell-1}<\beta_\ell=d$,
$S_k=S_k^L\cup S_k^{-\delta,\beta_1}\cup
S_k^{\beta_1,\beta_2}\cup\ldots \cup
S_k^{\beta_{\ell-1},\beta_\ell}$.
If we choose $0<\delta<\gamma$, $\beta_1=0$, $\beta_2=\gamma$, and
$\beta_{i+1}=(2\beta_i)\wedge d$ for $i\ge 2$,
the arguments above will show that
the probability of the infimum being less than or equal to $r$ over
$S_k^L$, $S_k^{-\delta,\beta_1}$ and each $S_k^{\beta_i,\beta_{i+1}}$
is summable over $k$, so that $P\left(\inf_{S_k} \mathbb{G}(s,t)+g(s,t)\le
r\right)$ is summable over $k$, so setting $M$ be such that the tail of
this sum past $M$ is less than $\varepsilon$ gives the desired result.
\end{proof}

\begin{proof}[proof of Theorem \ref{inf_bound_Gn} for the general case]
Let $\mathbb{G}_n(s,t)=\mathbb{G}_{n,x_m,j}(s,t)$ and
$g_n(s,t)=g_{n,x_m,j}(s,t)$.
Let $\eta$ be small enough that (\ref{gn_ineq}) holds.

As in the proof of the previous theorem, let $S_k=\{k\le \|(s,t)\|\le
k+1\}$ and let $S_k^L=S_k\cap \{\prod_i t_i\le (k+1)^{-\delta}\}$ for some
fixed $\delta$.  We have, using Lemma \ref{tail_bound_Gn},
\begin{align*}
&P\left(\inf_{S_k^L} \mathbb{G}_n(s,t)+g_n(s,t)\le r\right)
\le P\left(\sup_{S_k^L} |\mathbb{G}_n(s,t)|\ge |r|\right)  \\
&\le \left\{6(k+1)[(k+1)^d/k^{-\delta}]+2\right\}^{2d}
\exp\left(-C\frac{|r|}{(k+1)^{-\delta/2}}\right)
\end{align*}
for
$(k+1)^{-\delta}\ge n^{-2\gamma/(d+2\gamma)}(1+\log n)^2
\Longleftrightarrow k+1\le n^{2\gamma/[\delta(d+2\gamma)]}
   (1+\log n)^{-2/\delta}$
so, if $\delta<2\gamma$, this will
hold eventually for all $(k+1)\le h_n^{-1}\eta$ (once
$h_n^{-1}\eta\le n^{2\gamma/[\delta(d+2\gamma)]}(1+\log n)^{-2/\delta}
\Longleftrightarrow \eta\le n^{2\gamma/[\delta(d+2\gamma)]}
  n^{-1/(d+2\gamma)}(1+\log n)^{-2/\delta}
  = n^{(2\gamma/\delta-1)/(d+2\gamma)} (1+\log n)^{-2/\delta}$).
The bound is summable over $k$ for any $\delta>0$.

Again following the proof of the previous theorem, for
${\underline \beta}<{\overline \beta}$,
define
$S_k^{{\underline \beta},{\overline \beta}}
  =S_k\cap \{(k+1)^{\underline \beta}
  <\prod_i t_i\le (k+1)^{\overline \beta}\}$.  We have, again using Lemma
\ref{tail_bound_Gn},
\begin{align*}
&P\left(\inf_{S_k^{{\underline \beta},{\overline \beta}}} \mathbb{G}_n(s,t)+g_n(s,t)\le r\right)
\le P\left(\sup_{S_k^{{\underline \beta},{\overline \beta}}} |\mathbb{G}_n(s,t)|
  \ge C_1 k^{\gamma+{\underline \beta}}\right)  \\
&\le \left\{6(k+1)[(k+1)^d/(k^{{\underline \beta}}\wedge 1)]+2\right\}^{2d}
\exp\left(-C\frac{C_1 k^{\gamma+{\underline \beta}}}
  {(k+1)^{{\overline \beta}/2}}\right)
\end{align*}
for $(k+1)^{\overline \beta}\ge n^{-2\gamma/(d+2\gamma)}(1+\log n)^2$
(which will hold once the same inequality holds for $\delta$ for
$-\delta<{\overline \beta}$) and $k+1\le h_n^{-1}\eta$.
The bound is summable over $k$ for any
${\underline \beta},{\overline \beta}$ with
$2\gamma+2{\underline \beta}-{\overline \beta}>0$.

Thus, noting as in the previous theorem that, for any
$-\delta<\beta_1<\beta_2<\ldots<\beta_{\ell-1}<\beta_\ell=d$,
$S_k=S_k^L\cup S_k^{-\delta,\beta_1}\cup
S_k^{\beta_1,\beta_2}\cup\ldots \cup
S_k^{\beta_{\ell-1},\beta_\ell}$,
if we choose $0<\delta<\gamma$, $\beta_1=0$, $\beta_2=\gamma$, and
$\beta_{i+1}=(2\beta_i)\wedge d$ for $i\ge 2$,
the arguments above will show that
the probability of the infimum being less than or equal to $r$ over the
sets indexed by $k$ for any $k\le h_n^{-1}\eta$ is bounded uniformly in
$n$ by a sequence that is summable over $k$ (once
$\eta\le n^{(2\gamma/\delta-1)/(d+2\gamma)} (1+\log n)^{-2/\delta}$).
Thus, if we choose $M$ such that the tail of this sum past $M$ is less
than $\varepsilon$ and let $N$ be large enough so that
$\eta\le N^{(2\gamma/\delta-1)/(d+2\gamma)} (1+\log N)^{-2/\delta}$,
we will have the desired result.

\end{proof}

Lemmas \ref{large_s_lemma1_multi} and \ref{large_s_lemma2_multi} hold as
stated with the same proofs, so the rest of the proof is the same as in
the $\gamma=2$ case.  The $n/(a_n \log n)$ rate for $Z_{n,2}$ is still
faster than the $n^{(d+\gamma)/(d+2\gamma)}$ rate for $a_n$ increasing
slowly enough.

The proof of Theorem \ref{abs_cont_thm} for the limiting process is the
same as before.  The only place the drift term is used is in ensuring that
the inequality $g_{P,x_m,j}(s_{i,k},t_k)\le K/k$ holds in the last display
in the proof of the theorem, which is still the case.

\subsection*{Testing Rate of Convergence Conditions: Subsampling}

First, I collect results on the rate estimate $\hat \beta$ defined in
(\ref{rate_est_formula}).  The next lemma bounds $\hat \beta$ when the
statistic may not converge at a polynomial rate.  Throughout the
following, $S_n$ is a statistic on $\mathbb{R}$ with cdf $J_n(x)$ and
quantile function $J_n^{-1}(t)$.  $L_{n,b}(x|\tau)$ and $\tilde
L_{n,b}(x|\tau)$ are defined as in the body of the paper, with
$S(T_n(\theta))$ replaced by $S_n$.

\begin{lemma}\label{subsamp_lower_lemma}
Let $S_n$ be a statistic such that, for some sequence $\tau_n$ and $x>0$,
$\tau_n J_n^{-1}(t)\ge x$ for
large enough $n$.  Then, if $\tau_bS_n\stackrel{p}{\to} 0$ and $b/n\to 0$,
we will have, for any $\varepsilon>0$,
$L_{n,b}^{-1}(t+\varepsilon|\tau)\ge x-\varepsilon$ with probability
approaching one.
\end{lemma}
\begin{proof}
It suffices to show $L_{n,b}(x-\varepsilon|\tau)\le t+\varepsilon$ with
probability approaching one.  On
the event $E_n\equiv\{|\tau_bS_{\mathcal{S}}|\le \varepsilon\}$, which
has probability approaching one,
$L_{n,b}(x-\varepsilon|\tau)\le \tilde L_{n,b}(x|\tau)$.  We also have
$E[L_{n,b}(x|\tau)]=P(\tau_bS_{\mathcal{S}}\le x)=J_b(x/\tau_b)\le t$ by
assumption.  Thus,
\begin{align*}
&P(L_{n,b}(x-\varepsilon|\tau)\le t+\varepsilon)
\ge P\left( \left\{\tilde L_{n,b}(x|\tau)\le t+\varepsilon\right\}
      \cap E_n\right)  \\
&\ge P\left( \left\{\tilde L_{n,b}(x|\tau)\le E[L_{n,b}(x|\tau)]
        +\varepsilon\right\}
      \cap E_n\right).
\end{align*}
This goes to one by standard arguments.
\end{proof}

\begin{lemma}\label{beta_a_lemma}
Let $\hat\beta_a$ be the estimator defined in Section
\ref{subsamp_rate_subsec}, of any other estimator such that
$\hat\beta_a=\frac{-\log L_{n,b_1}^{-1}(t|1)+\mathcal{O}_p(1)}
  {\log b_1-\mathcal{O}_p(1)}$.
Suppose that, for some $x_\ell>0$ and $\beta_u$, $x_u n^{\beta_u}\le
J_n^{-1}(t-\varepsilon)$ eventually and
$b_1^{\beta_u}S_n\stackrel{p}{\to} 0$.
Then,
for any $\varepsilon>0$, we will have
$\hat\beta_a\le\hat\beta_u+\varepsilon$ with probability approaching one.
\end{lemma}
\begin{proof}
We have
\begin{align*}
\hat\beta_a=-\frac{\log L_{n,b_1}^{-1}(t|1)}{\log b_1}+o_p(1)
=\frac{\beta_u \log b_1-\log L_{n,b_1}^{-1}(t|b^{\beta_u})}
   {\log b_1}+o_p(1)
\le \beta_u-\frac{\log (x_u/2)}{\log b_1}+o_p(1)
\stackrel{p}{\to} \beta_u
\end{align*}
where the inequality holds with probability approaching one by Lemma
\ref{subsamp_lower_lemma}.
\end{proof}

The following lemma shows that the asymptotic distribution of the KS
statistic is strictly increasing on its support, which is needed for the
estimates of the rate of convergence in \citet{politis_subsampling_1999}
to converge at a fast enough rate that they can be used in the subsampling
procedure.

\begin{lemma}\label{inc_cdf_lemma}
Under Assumptions \ref{smoothness_assump_multi}, \ref{bdd_y_assump_multi},
\ref{S_assump}, \ref{inv_mat_assump} and \ref{abs_cont_S_assump}
with part (ii) of Assumption \ref{smoothness_assump_multi} replaced
by Assumption \ref{smoothness_assump_alpha}, if $S$ is convex, then the
the asymptotic distribution $S(Z)$ in Theorem \ref{inf_dist_thm_alpha}
satisfies $P(S(Z)\in (a,\infty))=1$ for some $a$, and the cdf of
$S(Z)$ is strictly increasing on $(a,\infty)$.
\end{lemma}
\begin{proof}
First, note that, for any concave functions $f_1,\ldots,f_{d_Y}$,
$f_i:V_i\to\mathbb{R}$, for some vector space $V_i$,
$x\mapsto S(f_1(x_1),\ldots,f_{d_Y}(x_{d_Y}))$ is convex, since, for any
$\lambda\in(0,1)$,
\begin{align*}
&S(f_1(\lambda x_{a,1}+(1-\lambda) x_{b,1})
  ,\ldots,f_k(\lambda x_{a,{d_Y}}+(1-\lambda) x_{b,{d_Y}}))  \\
&\ge S(\lambda f_1(x_{a,1})+(1-\lambda)f_k(x_{b,1}),
  ,\ldots,\lambda f_k(x_{a,{d_Y}})+(1-\lambda)f_k(x_{b,{d_Y}}))  \\
&\ge \lambda S(f_1(x_{a,1}),\ldots,f_k(x_{a,{d_Y}}))
   + (1-\lambda) S(f_1(x_{b,1}),\ldots,f_k(x_{b,{d_Y}}))
\end{align*}
where the first inequality follows since $S$ is decreasing in each
argument and by concavity of the $f_k$s, and the second follows by
convexity of $S$.

$S(Z)$ can be written as, for some random processes
$\mathbb{H}_1(t),\ldots,\mathbb{H}_{d_Y}(t)$ with continuous sample
paths and $\mathbb{T}\equiv \mathbb{R}^{|\mathcal{X}_0|\cdot 2d_X}$,
$S(\inf_{t\in\mathbb{T}}\mathbb{H}_1(t),\ldots
  ,\inf_{t\in\mathbb{T}}\mathbb{H}_{d_Y}(t))$.
Since the infimum
of a real valued function is a concave functional, this is a convex
function of the sample paths of
$(\mathbb{H}_1(t),\ldots,\mathbb{H}_{d_Y}(t))$.  The result follows from
Theorem 11.1 in \citet{davydov_local_1998} as long as the vector of random
processes can be given a topology for which this function is lower
semi-continuous.  In fact, this step can be done away with by noting that,
for $\mathbb{T}_0$ a countable dense subset of $\mathbb{T}$ and
$\mathbb{T}_\ell$ the first $\ell$ elements of this subset,
$S(\inf_{t\in\mathbb{T}_\ell}\mathbb{H}_1(t),\ldots
  ,\inf_{t\in\mathbb{T}_\ell}\mathbb{H}_{d_Y}(t))
\stackrel{d}{\to}
S(\inf_{t\in\mathbb{R}^{2d}}\mathbb{H}_1(t),\ldots
  ,\inf_{t\in\mathbb{R}^{2d}}\mathbb{H}_{d_Y}(t))$ as $\ell\to\infty$,
so, letting $F_\ell$ be the cdf of
$S(\inf_{t\in\mathbb{T}_\ell}\mathbb{H}_1(t),\ldots
  ,\inf_{t\in\mathbb{T}_\ell}\mathbb{H}_{d_Y}(t))$, applying Proposition
11.3 of \citet{davydov_local_1998} for each $F_\ell$ shows that
$\Phi^{-1}(F_\ell(t))$ is concave for each $\ell$, so, by convergence in
distribution, this holds for $S(Z)$ as well.
\end{proof}

The same result in \citet{davydov_local_1998} could also be used in the
proof of Theorem \ref{abs_cont_thm} to show that the distribution of
$S(Z)$ is continuous except possibly at the infimum of its support, but an
additional argument would be needed to show that, if such an atom exists,
it would have to be at zero.  In the proof of Theorem \ref{abs_cont_thm},
this is handled by using the results of \citet{pitt_local_1979} instead.

We are now ready to prove Theorem \ref{rate_thm_alpha}.

\begin{proof}[proof of Theorem \ref{rate_thm_alpha}]
First, suppose that Assumption \ref{smoothness_assump_multi} holds with
part (ii) of Assumption \ref{smoothness_assump_multi} replaced by
Assumption \ref{smoothness_assump_alpha} for some
$\underline\gamma<\gamma<\overline\gamma$ and
$\mathcal{X}_0$ nonempty.  By Theorem \ref{inf_dist_thm_alpha},
$n^{(d_X+\gamma)/(d_X+2\gamma)}S(T_n(\theta))$ converges in distribution
to a continous distribution.  Thus, by Lemma \ref{beta_a_lemma},
$\hat\beta_a\stackrel{p}{\to}(d_X+\gamma)/(d_X+2\gamma)$, so
$\hat\beta_a>\underline\beta=(d_X+\overline\gamma)/(d_X+2\overline\gamma)$
with probability approaching
one.  On this event, the test uses the subsample estimate of the
$1-\alpha$ quantile with rate estimate $\hat\beta \wedge \overline\beta$.
By Theorem 8.2.1 in \citet{politis_subsampling_1999},
$\hat\beta \wedge \overline\beta=(d_X+\gamma)/(d_X+2\gamma)
  +o_p((\log n)^{-1})$
as long as the asymptotic distribution of
$n^{(d_X+\gamma)/(d_X+2\gamma)}S(T_n(\theta))$ is increasing on the
smallest interval $(k_0,k_1)$ on which the asymptotic distribution has
probability one.  This holds by Lemma \ref{inc_cdf_lemma}.  By Theorem
8.3.1 in \citet{politis_subsampling_1999}, the $o_p((\log n)^{-1})$ rate
of convergence for the rate estimate $\hat\beta \wedge \overline\beta$
implies that the probability of rejecting converges to $\alpha$.

Next, suppose that %
Assumption \ref{smoothness_assump_multi} holds with part
(ii) of Assumption \ref{smoothness_assump_multi} replaced by Assumption
\ref{smoothness_assump_alpha} for $\gamma=\overline\gamma$.
The test that compares $n^{1/2}S(T_n(\theta))$ to a positive critical
value will fail to reject with probability approaching one in this case,
so, on an event with probability approaching one, the test will reject
only if $\hat\beta_a\ge\underline\beta$ and
the subsampling test with rate $\hat\beta \wedge \overline\beta$ rejects.
Thus, the probability of rejecting is asymptotically no greater than the
probability of rejecting with the subsampling test with rate
$\hat\beta \wedge \overline\beta$, which has asymptotic level $\alpha$
under these conditions by the argument above.

Now, consider the case where, for some $x_0\in\mathcal{X}_0$ and
$B<\infty$, $\bar m_j(\theta,x)\le B\|x-x_0\|^{\gamma}$ for some
$\gamma>\bar \gamma$.  Let
$\tilde m_j(W_i,\theta)
  =m_j(W_i,\theta)+(B\|x-x_0\|^{\gamma}-\bar m_j(\theta,x))$.  Then
$\tilde m_j(W_i,\theta)\ge m_j(W_i,\theta)$, and $\tilde m_j(W_i,\theta)$
satisfies the assumptions of Theorems \ref{inf_dist_thm_alpha} and
\ref{abs_cont_thm}, so
\begin{align*}
n^{(d_X+\gamma)/(d_X+2\gamma)}S(T_n(\theta))
\ge n^{(d_X+\gamma)/(d_X+2\gamma)}
  S(0,\ldots,0,
  \inf_{s,t}E_n \tilde m_j(W_i,\theta)I(s<X_i<s+t),0,\ldots,0)
\end{align*}
and the latter quantity converges in distribution to a continuous random
variable that is positive with probability one.  Thus, by Lemma
\ref{beta_a_lemma}, for any $\varepsilon>0$,
$\hat\beta_a<(d_X+\gamma)/(d_X+2\gamma)+\varepsilon$
with probability approaching one.  For $\varepsilon$ small enough, this
means that
$\hat\beta_a<(d_X+\overline\gamma)/(d_X+2\overline\gamma)$ with
probability approaching one.  Thus, the procedure uses an asymptotically
level $\alpha$ test with probability approaching one.

The remaining case is where $\bar m_j(\theta,x)$ is bounded from below
away from zero.  If $m_j(W_i,\theta)\ge 0$ for all $j$ with probability
one, $S(T_n(\theta))$ and the estimated $1-\alpha$ quantile will both be
zero, so the probability of rejecting will be zero, so suppose that
$P(m_j(W_i,\theta)<0)>0$ for some $j$.  Then, for some $\eta>0$, we have
$nS(T_n(\theta))>\eta$ with
probability approaching one.  From Lemma \ref{subsamp_lower_lemma}
(applied with $t$ less that $1-\alpha$ and $\tau_b=b$), it
follows that
$L_{n,b}^{-1}(1-\alpha|b^{\hat\beta\wedge\overline\beta})
=b^{\hat\beta\wedge\overline\beta-1} L_{n,b}^{-1}(1-\alpha|b)
\ge b^{\hat\beta\wedge\overline\beta-1} \eta/2$ with probability
approaching one.
By Lemma \ref{large_s_lemma2_multi},
$S(T_n(\theta))$ will converge at a $n\log n$ rate, so that
$n^{\hat\beta\wedge\overline\beta}S(T_n(\theta))
  <n^{\hat\beta\wedge\overline\beta-1} (\log n)^2$ with probability
approaching one.  Thus, we will fail to reject with probability
approaching one as long as
$n^{\hat\beta\wedge\overline\beta-1} (\log n)^2
  \le b^{\hat\beta\wedge\overline\beta-1} \eta/2
  = n^{\chi_3(\hat\beta\wedge\overline\beta-1)} \eta/2$
for large enough $n$, and this holds since $\chi_3<1$.
A similar argument holds for
$\tilde L_{n,b}^{-1}(1-\alpha|b^{\hat\beta\wedge\overline\beta})$.

\end{proof}

\subsection*{Testing Rate of Convergence Conditions:
  Estimating the Second Derivative}

\begin{proof}[proof of Lemma \ref{min_set_lemma}]
Let $h(x)=\bar m_j(\theta,x)-\min_{x'\in D} \bar m_j(\theta,x)$ where
$\bar m_j(\theta,x)=E(m_j(W_i,\theta)|X_i=x)$ for a continuous version of
the conditional mean function.
First, note that
$\mathcal{X}_0^j$ is compact.
Since each $x\in\mathcal{X}_0^j$ is a local minimizer of $h(x)$ such that
the second derivative matrix is strictly positive definite at $x$,
there is an open set $A(x)$ containing
each $x\in\mathcal{X}_0^j$ such that $h(x)>0$ on $A(x)\backslash x$.
The sets $A(x)$ with $x$ ranging over $\mathcal{X}_0^j$ form a covering
of $\mathcal{X}_0^j$ with open sets.  Thus, there is a finite subcover
$A(x_1),\ldots A(x_\ell)$ of $\mathcal{X}_0^j$.  Since the only elements in
$A(x_1)\cup\cdots\cup A(x_\ell)$ that are also in $\mathcal{X}_0^j$ are
$x_1,\ldots, x_\ell$, this means that
$\mathcal{X}_0^j=\{x_1,\ldots, x_\ell\}$.
\end{proof}

\begin{proof}[proof of Theorem \ref{X0_est_thm}]
By the next lemma, we will have
$\mathcal{X}_0^j\subseteq \hat{\mathcal{X}}_0^j \subseteq
\cup_{k=1}^{\hat \ell_j}B_{\varepsilon_n}(\hat x_{j,k})$ and
$\mathcal{X}_0^j\subseteq \hat{\mathcal{X}}_0^j \subseteq
\cup_{k \text{ s.t. } j\in J(k)} B_{\varepsilon_n}(x_k)$
with probability approaching one.  When this holds, we will have
$\hat \ell \le |\{k|j\in J(k)\}|$ by construction and, once
$\varepsilon_n$ is less than the smallest distance between any two points
in $\mathcal{X}_0^j$, we will also have $\hat \ell_j = |\{k|j\in J(k)\}|$
and, for each $k$ from $1$ to $\hat \ell_j$, we will have, for some
function $r(j,k)$ such that $r(j,\cdot)$, is bijective from
$\{1,\ldots,\hat \ell_j\}$ to $\{k|j\in J(k)\}$,
$x_{r(j,k)}\in B_{\varepsilon_n}(\hat x_{j,k})$ for each $j,k$.  When this
holds, all of the $\hat x_{j,k}$s with $r(j,k)$ equal will be in the same
equivalence class, since the corresponding $\varepsilon_n$ neighborhoods
will intersect.  When $\varepsilon_n$ is small enough that
$\varepsilon_n$ neighborhoods containing $x_r$ and $\varepsilon_n$
neighborhoods containing $x_s$ do not intersect for $r\ne s$, there will
be exactly $\ell$ equivalence classes, each one corresponding to the
$(j,k)$ indices such that $r(j,k)$ is the same.  Let the labeling of the
$\tilde x_s$s be such that, for all $s$, $\tilde x_s=\hat x_{j,k}$ for
some $(j,k)$ such that $r(j,k)=s$.  Then, for each $s$, we have, for some
$(j,k)$ such that $r(j,k)=s$,
$x_s=x_{r(j,k)}\in B_{\varepsilon_n}(\hat x_{j,k})
=B_{\varepsilon_n}(\tilde x_s)$ with probability approaching one so that
$\tilde x_s\stackrel{p}{\to} x_s$.
To verify that $\hat J(s)=J(s)$ with probability approaching one, note
that, for $j\in J(s)$, we will have
$x_s\in \mathcal{X}_0^j\subseteq \cup_k B_{\varepsilon_n}(\hat x_{j,k})$
and $x_s\in B_{\varepsilon_n}(\tilde x_s)$
eventually, and, when this holds,
$\left[\cup_k B_{\varepsilon_n}(\hat x_{j,k})\right]
\cap B_{\varepsilon_n}(\tilde x_s) \ne \emptyset$
so that $j\in \hat J(s)$.  For $j\notin J(s)$, each $\hat x_{j,k}$ will
eventually be within $\varepsilon_n$ of some $x_r$ with $r\ne s$, while
indices $(j',k')$ in the equivalence class associated with $s$ will
eventually have $\hat x_{j',k'}$ within $2\varepsilon$ of $x_s$, so that
$(j,k)$ will not be in the equivalence class associated with $s$ for any
$k$, and $j\notin \hat J(s)$.
\end{proof}

\begin{lemma}
Suppose that
$\sup_{x\in D}\|\hat{\bar m}_j(\theta,x)-\bar m_j(\theta,x)\|
=\mathcal{O}(a_n)$ for some sequence $a_n\to 0$.  Then, under Assumption
\ref{smoothness_assump_multi},
for any sequence
$b_n\to \infty$ with $b_na_n\to 0$ and $\varepsilon_n$ with
$\varepsilon_n\to 0$ more slowly than $\sqrt{b_na_n}$, the set 
$\hat{\mathcal{X}}_0^j\equiv \{x|\hat{\bar m}_j(\theta,x)\le b_na_n\}$
satisfies
\begin{align*}
\mathcal{X}_0^j\subseteq \hat{\mathcal{X}}_0^j
  \subseteq \cup_{k \text{ s.t. } j\in J(k)} B_{\varepsilon_n}(x_k)
\end{align*}
\begin{proof}
We will have $\mathcal{X}_0^j\subseteq \hat{\mathcal{X}}_0^j$ as soon as
$\sup_{x\in D}\|\hat{\bar m}_j(\theta,x)-\bar m_j(\theta,x)\|\le b_n a_n$,
which happens with probability approaching one.  To show that
$\hat{\mathcal{X}}_0^j\subseteq
\cup_{k \text{ s.t. } j\in J(k)} B_{\varepsilon_n}(x_k)$ eventually,
suppose that, for some $\hat x\in \hat{\mathcal{X}}_0^j$, $\hat x\notin
B_{\varepsilon_n}(x_k)$ for any $k$.  Let $C$ and $\eta$ be such that
$\bar m_j(\theta,x)\ge C\min_k\|x-x_k\|^2$ when $\|x-x_k\|\le \eta$ for
some $k$ (such a $C$ and $\eta$ exist by Assumption
\ref{smoothness_assump_multi}).  Then, for any $\hat x$ such that $\hat{\bar
m}_j(\theta,\hat x) \le b_na_n$, we must have, with probability
approaching one,
\begin{align*}
C\min_k\|x-x_k\|^2\le \bar m_j(\theta,\hat x)
\le b_na_n + \bar m_j(\theta,\hat x)-\hat{\bar m}_j(\theta,\hat x)
\le 2b_na_n
\end{align*}
where the first inequality follows since $\hat{\mathcal{X}}_0^j$ is
contained in
$\{x|\|x-x_k\|\le \eta \text{ some } k \text{ s.t. } j\in J(k)\}$
eventually.  Since $\varepsilon_n\ge \sqrt{2b_na_n/C}$ eventually, the
first claim follows.
\end{proof}

\end{lemma}

\subsection*{Local Alternatives}

\begin{proof}[proof of Theorem \ref{local_alt_exact_thm}]
Everything is the same as in the proof of Theorem
\ref{inf_dist_thm_multi}, but with the following modifications.

First, in the proof of Theorem \ref{local_process_thm_multi}, we need to
show that, for all $j$,
\begin{align*}
\frac{\sqrt{n}}{\sqrt{h_n^d}}(E_n-E)[m_j(W_i,\theta_0+a_n)-m_j(W_i,\theta_0)]
  I(h_ns<X_i-x_k<h_n(s+t))
\end{align*}
converges to zero uniformly over $\|(s,t)\|<M$ for any fixed $M$.  By
Theorem 2.14.1 in \citet{van_der_vaart_weak_1996},
the $L^2$ norm of this is bounded up to a constant by
$J(1,\mathcal{F}_n,L_2)\frac{1}{h_n^d}\sqrt{E F_n(X_i,W_i)^2}$, where
$\mathcal{F}_n
=\{(x,w)\mapsto [m_j(w,\theta_0+a_n)-m_j(w,\theta_0)]
  I(h_ns<x-x_k<h_n(s+t))|(s,t)\in\mathbb{R}^{2d}\}$
and
$F_n(x,w)=|m_j(w,\theta_0+a_n)-m_j(w,\theta_0)|I(-h_nM\iota<x-x_k<2h_nM\iota)$
is an envelope function for this class (here $\iota$ is a vector of
ones).
The covering numbers of the $\mathcal{F}_n$s are uniformly bounded by a
polynomial,  so that we just need to
show that $\frac{1}{h_n^d}\sqrt{E F_n(X_i,W_i)^2}$ converges to zero.  We
have
\begin{align*}
&\frac{1}{\sqrt{h_n^d}}\sqrt{E F_n(X_i,W_i)^2}  \\
&=\frac{1}{\sqrt{h_n^d}}\sqrt{E
  E\{[m_j(W_i,\theta_0+a_n)-m_j(W_i,\theta_0)]^2|X_i\}I(-h_nM\iota<X_i-x_k<2h_nM\iota)}
\\
&\le \frac{1}{\sqrt{h_n^d}}\sqrt{E I(-h_nM\iota<X_i-x_k<2h_nM\iota)}
\sup_{\|x-x_k\|\le \eta}
E\{[m_j(W_i,\theta_0+a_n)-m_j(W_i,\theta_0)]^2|X_i=x\}
\end{align*}
where the first equality uses the law of iterated expectations and the
second holds eventually with $\eta$ chosen so that the convergence in
Assumption \ref{m2_assump} is uniform over $\|x-x_k\|<\eta$.  The first
term is bounded eventually by
$\overline f\int_{-M\iota<x<2M\iota} \, dx$
where $\overline f$ is a bound for the density of $X_i$ in a neighborhood
of $x_k$ (this follows from the same change of variables as in other parts
of the proof).  The second term converges to zero by Assumption
\ref{m2_assump}.

Next, in the proof of Theorem \ref{local_process_thm_multi}, we need to
show that
\begin{align*}
\frac{1}{h_n^{d+2}}E[\bar m_j(\theta_0+a_n,X_i)-\bar m_j(\theta_0,X_i)]
I(h_ns<X_i-x_k<h_n(s+t))
\to f_X(x_k)\bar m_{\theta,j}(\theta_0,x_k)a
\prod_i t_i
\end{align*}
uniformly in $\|(s,t)\|\le M$.  We have
\begin{align*}
&\frac{1}{h_n^{d+2}}E[\bar m_j(\theta_0+a_n,X_i)-\bar m_j(\theta_0,X_i)]
I(h_ns<X_i-x_k<h_n(s+t))
-f_X(x_k)\bar m_{\theta,j}(\theta_0,x_k)a
\prod_i t_i    \\
&=\frac{1}{h_n^{d+2}}\int_{h_ns<x-x_k<h_n(s+t)}
\left\{[\bar m_j(\theta_0+a_n,x)-\bar m_j(\theta_0,x)]f_X(x)
-h_n^2f_X(x_k)\bar m_{\theta,j}(\theta_0,x_k)a\right\} \, dx  \\
&=\int_{s<x<s+t}
\left\{h_n^{-2}[\bar m_j(\theta_0+a_n,h_n x +x_k)-\bar m_j(\theta_0,h_n x+ x_k)]
f_X(h_n x+ x_k)
-f_X(x_k)\bar m_{\theta,j}(\theta_0,x_k)a\right\} \, dx
\end{align*}
where the second equality comes from the change of variable $x\mapsto h_n
x+x_k$.
This will go to zero uniformly in $\|(s,t)\|\le M$ as long as
$\sup_{\|x\|\le   2M}\|f_X(h_n x+ x_k)-f_X(x_k)\|$ and
\begin{align*}
\sup_{\|x\|\le 2M}
\|h_n^{-2}[\bar m_j(\theta_0+a_n,h_n x +x_k)-\bar m_j(\theta_0,h_n x+ x_k)]
-\bar m_{\theta,j}(\theta_0,x_k)a\|
\end{align*}
both go to zero.  $\sup_{\|x\|\le   2M}\|f_X(h_n x+ x_k)-f_X(x_k)\|$ goes
to zero by continuity of $f_X$ at $x_k$.  As for the other expression,
since $a h_n^2=a_n$, the mean value theorem shows that this is equal to
$\bar m_{\theta,j}(\theta^*(a_n),h_nx+x_k)a
-\bar m_{\theta,j}(\theta_0,x_k)a$
for some $\theta^*(a_n)$ between $\theta_0$ and $\theta_0+a_n$.
This goes to
zero by Assumption \ref{diff_m_assump}.

In verifying the conditions of Lemma \ref{inf_dist_lemma_multi},
we need to make sure the bounds,
$g_{P,x_k,j,a}(s,t)\ge C\|(s,t)\|^2\prod_i t_i$
and
\begin{align*}
g_{n,x_k,j,a}(s,t)
\equiv \frac{1}{h_n^{d+2}} Em_j(W_i,\theta_0+a_n)I(h_ns<X_i<h_n(s+t))
\ge C\|(s,t)\|^2\prod_i t_i
\end{align*}
still hold for $\|(s,t)\|\ge M$ for $M$ large enough and, for the latter
function, $\|(s,t)\|\le h_n^{-1}\eta$ for some $\eta>0$ and $n$ greater
than some $N$ that does not depend on $M$.  We have
\begin{align*}
&g_{P,x_k,j,a}(s,t)
=g_{P,x_k,j}(s,t)+\bar m_{\theta,j}(\theta_0,x_k) a f_X(x_k)\prod_i t_i
\ge C\|(s,t)\|^2\prod_i t_i
+\bar m_{\theta,j}(\theta_0,x_k) a f_X(x_k)\prod_i t_i  \\
&=\|(s,t)\|^2
[C+\bar m_{\theta,j}(\theta_0,x_k) a f_X(x_k)/\|(s,t)\|^2]\prod_i t_i
\end{align*}
where the first inequality follows from the bound in the original proof.
For $\|(s,t)\|\ge M$ for $M$ large enough, this is greater than or equal
to $K\|(s,t)\|^2\prod_i t_i$ for
$K=C-|\bar m_{\theta,j}(\theta_0,x_k) a| f_X(x_k)/M^2>0$.  For
$g_{n,x_k,j,a}(s,t)$, we have
\begin{align*}
&\|g_{P,x_k,j,a}(s,t)-g_{P,x_k,j}(s,t)\|
=\|\frac{1}{h_n^{d+2}} E[m_j(W_i,\theta_0+a_n)-m_j(W_i,\theta_0)]
I(h_ns<X_i<h_n(s+t)) \|  \\
&\le \sup_{\|x-x_k\|\le \eta}
\|\frac{1}{h_n^2}[\bar m_j(\theta_0+a_n,x)-\bar m_j(\theta_0,x)]\|
\|\frac{1}{h_n^d} E I(h_ns<X_i<h_n(s+t)) \|.
\end{align*}
By the mean value theorem,
$\bar m_j(\theta_0+a_n,x)-\bar m_j(\theta_0,x)
=\bar m_{j,\theta}(\theta^*(a_n),x)a_n$ for some $\theta^*(a_n)$ between
$\theta_0$ and $\theta_0+a_n$.  By continuity of the derivative as a
function of $(\theta,x)$, for small enough $\eta$ and $n$ large enough,
$\bar m_{j,\theta}(\theta^*(a_n),x)$ is bounded from above, so that
$\|\frac{1}{h_n^2}[\bar m_j(\theta_0+a_n,x)-\bar m_j(\theta_0,x)]\|$ is
bounded by a constant times $\|a_n\|/h_n^2=\|a\|$.
By continuity of $f_X$ at $x_k$,
$\|\frac{1}{h_n^d} E I(h_ns<X_i<h_n(s+t)) \|$ is bounded by some
constant %
times $\prod_i t_i$ for $\|(s,t)\|\le h_n^{-1}\eta$.  Thus,
for $M\le \|(s,t)\|\le h_n^{-1}\eta$ for the appropriate $M$ and $\eta$,
we have, for some constant $C_1$,
\begin{align*}
&g_{P,x_k,j,a}(s,t)
\ge g_{P,x_k,j}(s,t)-C_1\prod_i t_i
\ge C\|(s,t)\|^2\prod_i t_i-C_1\prod_i t_i  \\
&= \|(s,t)\|^2[C-C_1/\|(s,t)\|^2]\prod_i t_i
\end{align*}
where the second inequality uses the bound from the original proof.  For
$M$ large enough, this gives the desired bound with the constant equal to
$C-C_1/M>0$.

In verifying the conditions of Lemma \ref{inf_dist_lemma_multi}, we also
need to make sure the argument in Lemma \ref{tail_bound_Gn} still
goes through when $m(W_i,\theta_0)$ is replaced by $m(W_i,\theta_0+a_n)$.
To get the lemma to hold (with the constant $C$ depending only on the
distribution of $X$ and the $\overline Y$ in Assumption
\ref{bdd_y_assump_local}), we can use the same proof, but with the classes
of functions $\mathcal{F}_n$
defined to be
$\mathcal{F}_n
=\{(x,w)\mapsto m_j(w,\theta_0+a_n)
  I(h_ns_0<x-x_k<h_n(s_0+t))|t\le t_0\}$
($J(1,\mathcal{F}_n,L^2)$ is bounded uniformly for these classes because
the covering number of each $\mathcal{F}_n$ is bounded by the same
polynomial),
and using the envelope function
$F_n(x,w)=\overline Y I(h_ns_0<x-x_k<h_n(s_0+t_0))$ when applying Theorem
2.14.1 in \citet{van_der_vaart_weak_1996}.

\end{proof}

\begin{proof}[proof of Theorem \ref{local_alt_degen_thm}]
First, note that, for any neighborhoods $B(x_k)$ of the elements of
$\mathcal{X}_0$,
$\sqrt{n} \inf_{s,t} E_nm_j(W_i,\theta_0+a_n)I(s<X<s+t)
=\sqrt{n} \inf_{(s,s+t)\in \cup_{k \text{ s.t. } j\in J(k)} B(x_k)}
E_nm_j(W_i,\theta_0+a_n)I(s<X_i<s+t)+o_{p}(1)$
since, if these neighborhoods are made small enough, we
will have, for any $(s,s+t)$ not in one of these neighborhoods,
$Em_j(W_i,\theta_0+a_n)I(s<X_i<s+t)
\ge \underline B P(s<X_i<s+t)$ by an argument similar to the one in Lemma
\ref{large_s_lemma1_multi}, so that an argument similar to the one in
Lemma \ref{large_s_lemma2_multi} will show that
$\inf_{(s,s+t)\in \cup_{k \text{ s.t. } j\notin J(k)} B(x_k)}
E_nm_j(W_i,\theta_0+a_n)I(s<X_i<s+t)$ converges to zero at a faster than
$\sqrt{n}$ rate (Assumption \ref{diff_m_assump} guarantees
that $E[m_j(W_i,\theta_0+a_n)|X]$ is eventually bounded away from zero
outside of any neighborhood of $\mathcal{X}_0$ so that a similar argument
applies).

Thus, the result will follow once we show that, for each $j$ and $k$ such
that $j\in J(k)$,
\begin{align*}
&\sqrt{n} \inf_{(s,s+t)\in B(x_k)} E_nm_j(W_i,\theta_0+a_n)I(s<X_i<s+t)
  \\
&\stackrel{p}{\to} \inf_{s,t}
f_X(x_k)\int_{s<x<s+t}
\left(\frac{1}{2}x'Vx
+\overline m_{\theta,j}(\theta_0,x_k)a\right) \, dx.
\end{align*}
With this in mind, fix $j$ and $k$ with $j\in J(k)$.

Let $(s^*_n,t^*_n)$ minimize $E_nm_j(W_i,\theta_0+a_n)I(s<X<s+t)$ over
$B(x_k)^2$ (and be chosen from the set of minimizers in a measurable way).
First, I show
that $\rho(0,(s^*_n,t^*_n))\stackrel{p}{\to} 0$ where $\rho$ is
the covariance semimetric
$\rho((s,t),(s',t'))
=var(m_j(W_i,\theta_0)I(s<x<s+t)-m_j(W_i,\theta_0)I(s'<x<s'+t'))$.
To show
this, note that, for any $\varepsilon>0$,
$Em_j(W_i,\theta_0+a_n)I(s<X_i<s+t)$ is
bounded from below away from zero for $\rho(0,(s,t))\ge \varepsilon$
for large enough $n$.  To see this, note that, for $\rho(0,(s,t))\ge
\varepsilon$, $\prod_i t_i\ge K$ for some constant $K$, so that
$\|(s,t)\|\ge K^{1/d}$ and, for some constant $C$ and a bound $\overline
f$ for $f_X$ on $B(x_k)$,
\begin{align*}
&Em_j(W_i,\theta_0+a_n)I(s<X_i<s+t)  \\
&=Em_j(W_i,\theta_0)I(s<X_i<s+t)
+E[\bar m_j(\theta_0+a_n,X_i)-\bar m_j(\theta_0,X_i)]I(s<X_i<s+t)  \\
&\ge C_1\|(s,t)\|^2\left(\prod_i t_i\right)
-\sup_{x\in B(x_k)}\|\bar m_j(\theta_0+a_n,x)-\bar m_j(\theta_0,x)\|
\overline f \left(\prod_i t_i\right)  \\
&\ge \left[C_1\|(s,t)\|^2
-\sup_{x\in B(x_k)}\|\bar m_j(\theta_0+a_n,x)-\bar m_j(\theta_0,x)\|
\overline f \right] K.
\end{align*}
By Assumption \ref{m2_assump}, $\sup_{x\in B(x_k)}\|\bar
m_j(\theta_0+a_n,x)-\bar m_j(\theta_0,x)\|$ converges to zero, so
the last term in this display will be positive and bounded away from zero
for large enough $n$.
Thus, we can write $\sqrt{n}E_nm_j(W_i,\theta_0+a_n)I(s<X_i<s+t)$ as the sum of
$\sqrt{n}(E_n-E)m_j(W_i,\theta_0+a_n)I(s<X_i<s+t)$, which is
$\mathcal{O}_p(1)$
uniformly in $(s,t)$, and $\sqrt{n}Em_j(W_i,\theta_0+a_n)I(s<X<s+t)$,
which is bounded from below uniformly in $\rho(0,(s,t))\ge \varepsilon$ by
a sequence of constants that go to infinity.  Thus,
$\inf_{\rho(0,(s,t))\ge \varepsilon}
\sqrt{n}E_nm_j(W_i,\theta_0+a_n)I(s<X<s+t)$ is greater than zero with
probability approaching one, so
$\rho(0,(s^*,t^*))\stackrel{p}{\to} 0$.

Thus, for some sequence of random variables
$\varepsilon_n\stackrel{p}{\to} 0$,
\begin{align*}
&\sqrt{n}\inf_{s,t} E_nm_j(W_i,\theta_0+a_n)I(s<X<s+t)  \\
&=\sqrt{n}\inf_{\rho(0,(s^*,t^*))\le \varepsilon_n, (s,s+t)\in B(x_k)}
E_nm_j(W_i,\theta_0+a_n)I(s<X<s+t).
\end{align*}
This is equal to $\sqrt{n}\inf_{\rho(0,(s^*,t^*))\le \varepsilon_n,
  (s,s+t)\in B(x_k)}
Em_j(W_i,\theta_0+a_n)I(s<X<s+t)$ plus a term that is bounded by
$\sqrt{n}\sup_{\rho(0,(s^*,t^*))\le \varepsilon_n, (s,s+t)\in B(x_k)}
|(E_n-E)E_nm_j(W_i,\theta_0+a_n)I(s<X<s+t)|$.
By Assumption \ref{m2_assump} and an argument using the maximal inequality
in Theorem 2.14.1 in \citet{van_der_vaart_weak_1996},
$\sqrt{n}\sup_{(s,s+t)\in B(x_k)}
|(E_n-E)[m_j(W_i,\theta_0+a_n)-m_j(W_i,\theta_0)]I(s<X_i<s+t)|$ converges
in probability to zero.
$\sqrt{n}(E_n-E)m_j(W_i,\theta_0)I(s<X_i<s+t)$ converges 
in distribution under the supremum norm to a mean zero
Gaussian process
$\mathbb{H}(s,t)$ with covariance kernel
$cov(\mathbb{H}(s,t),\mathbb{H}(s',t'))
=cov(m_j(W_i,\theta_0)I(s<X_i<s+t),m_j(W_i,\theta_0)I(s'<X_i<s'+t'))$
and almost sure $\rho$ continuous sample paths.  Since
$(z,\varepsilon)\mapsto \sup_{\rho(0,(s,t))\le \varepsilon} |z(s,t)|$ is
continuous in $C(\mathbb{R}^{2d_X},\rho)\times \mathbb{R}$ (where
$C(\mathbb{R}^{2d_X},\rho)$ is the space of $\rho$ continuous functions
on $\mathbb{R}^{2d}$) under the product norm of the supremum norm and
the Euclidean norm, by the continuous mapping theorem,
$\sup_{\rho(0,(s,t))\le \varepsilon_n}
|\sqrt{n}(E_n-E)m_j(W_i,\theta_0)I(s<X_i<s+t)|
\stackrel{d}{\to} \sup_{\rho(0,(s,t))\le 0} \mathbb{H}(s,t) = 0$
(the last step follows since $var(\mathbb{H}(s,t))=0$ whenever
$\rho(0,(s,t))=0$).

Thus,
\begin{align*}
&\sqrt{n}\inf_{(s,s+t)\in B(x_k)}E_nm_j(W_i,\theta_0+a_n)I(s<X_i<s+t)  \\
&=\sqrt{n}\inf_{\rho(0,(s,t))<\varepsilon_n, (s,s+t)\in B(x_k)}
Em_j(W_i,\theta_0+a_n)I(s<X_i<s+t)+o_p(1)  \\
&=\sqrt{n}\inf_{\rho(0,(s,t))<\varepsilon_n, (s,s+t)\in B(x_k)}
\int_{s<x<s+t} \bar m_j(\theta_0+a_n,x) f_X(x) \, dx +o_p(1).
\end{align*}
By Assumption \ref{diff_m_assump}, the integrand is positive eventually
for $\|(s-x_k,t)\|\ge\eta$ for any $\eta>0$, and once this holds, the
infimum will be achieved on $\|(s-x_k,t)\|<\eta$.  Using a first order
Taylor expansion in the first argument of $\bar m_j(\theta_0+a_n,x)$ and a
second order Taylor expansion in the second argument the integrand is
equal to
\begin{align*}
\left[\frac{1}{2}(x-x_k)V(x^*(x))(x-x_k)'
+\bar m_{\theta,j}(\theta^*(a_n),x)a_n\right] f_X(x)
\end{align*}
for some $x^*(x)$ between $x$ and $x_k$ and $\theta^*(a_n)$ between
$\theta_0$ and $\theta_0+a_n$.  For $\eta$ small enough, continuity of the
derivatives at $(\theta_0,x_k)$ guarantees that this is bounded from below
by $C_1\|x-x_k\|^2-C_2a_n$ for some constants $C_1$ and $C_2$, so the
integrand is positive for $x$ greater than $C \sqrt{\|a_n\|}$ for some large
$C$, so that the infimum will be taken on $\|(s,s+t)\|< C\sqrt{\|a_n\|}$.
Thus, we have
\begin{align*}
&\sqrt{n}\inf_{(s,s+t)\in B(x_k)}E_nm_j(W_i,\theta_0+a_n)I(s<X_i<s+t)  \\
&=\sqrt{n}\inf_{\rho(0,(s,t))<\varepsilon_n, \|(s-x_k,t)\|<C \sqrt{\|a_n\|}}
\int_{s<x<s+t} \bar m_j(\theta_0+a_n,x) f_X(x) \, dx +o_p(1).
\end{align*}
This will be equal up to $o(1)$ to the infimum of
\begin{align*}
\sqrt{n} \int_{s<x<s+t} \left[\frac{1}{2}(x-x_k)V_j(x_k)(x-x_k)'
+\bar m_{\theta,j}(\theta_0,x_k)a_n\right]
f_X(x_k) \, dx
\end{align*}
once we show that the difference between this expression and
$\sqrt{n}\int_{s<x<s+t} \bar m_j(\theta_0+a_n,x) f_X(x) \, dx$ goes to
zero uniformly over $\|(s-x_k,t)\|\le C \sqrt{\|a_n\|}$ (the infimum of this
last display will be taken at a sequence where $\|(s-x_k,t)\|\le C
\sqrt{\|a_n\|}$ anyway, so that the infimum can be taken over all of
$\mathbb{R}^{2d}$).

The difference between these terms is
\begin{align*}
&\sqrt{n} \int_{s<x<s+t} \left[\frac{1}{2}(x-x_k)V_j(x_k)(x-x_k)'
+\bar m_{\theta,j}(\theta_0,x_k)a_n\right]
[f_X(x)-f_X(x_k)] \, dx  \\
&+\sqrt{n} \int_{s<x<s+t} \frac{1}{2}\left[(x-x_k)V_j(x^*(x))(x-x_k)'
-(x-x_k)V_j(x_k)(x-x_k)'\right]
f_X(x) \, dx  \\
&+\sqrt{n} \int_{s<x<s+t}\left[\bar m_{\theta,j}(\theta^*(a_n),x)
-\bar m_{\theta,j}(\theta_0,x_k)\right]a_n
f_X(x) \, dx.
\end{align*}
These can all be bounded using the change of variables
$u=(x-x_k) n^{1/(2(d+2))}$
and the continuity of densities, conditional means, and their derivatives.
The first term is
\begin{align*}
&\sqrt{n} \int_{n^{1/(2(d+2))}(s-x_k)<u<(s+t-x_k)n^{1/(2(d+2))}}
\left[\frac{1}{2}uV_j(x_k)u'n^{-1/(d+2)}
+\bar m_{\theta,j}(\theta_0,x_k)an^{-1/(d+2)}\right]  \\
&\times[f_X(n^{-1/(2(d+2))}u+x_k)-f_X(x_k)] n^{-d/(2(d+2))}\, du  \\
&=\int_{n^{1/(2(d+2))}(s-x_k)<u<(s+t-x_k)n^{1/(2(d+2))}}
\left[\frac{1}{2}uV_j(x_k)u'
+\bar m_{\theta,j}(\theta_0,x_k)a\right]  \\
&\times[f_X(n^{-1/(2(d+2))}u+x_k)-f_X(x_k)]\, du.
\end{align*}
The integrand converges to zero uniformly over $u$ in any bounded set
by the continuity of $f_X$ at $x_k$, and the area of integration is
bounded by $\|u\|\le 2n^{1/(2(d+2))}\|(s-x_k,t)\|\le
2Cn^{1/(2(d+2))}\sqrt{\|a\|}n^{-1/(2(d+2))}=2C \sqrt{\|a\|}$ on
$\|(s-x_k,t)\|\le C\sqrt{\|a_n\|}$.
Using the same change of variables, the second term is bounded by the
integral of
\begin{align*}
\frac{1}{2}\left[u'V_j(x^*(n^{-1/(2(d+2))}u+x_k))u'
-uV_j(x_k)u'\right]
f_X(n^{-1/(2(d+2))}u+x_k)
\end{align*}
over a bounded region, and this converges to zero uniformly in any bounded
region by continuity of the second derivative matrix.  The last term is,
by the same change of variables, bounded by the integral of
\begin{align*}
\left[\bar m_{\theta,j}(\theta^*(a_n),n^{-1/(2(d+2))}u+x_k)
-\bar m_{\theta,j}(\theta_0,x_k)\right]a
f_X(n^{-1/(2(d+2))}u+x_k)
\end{align*}
over a bounded region, and this converges to zero by continuity of
$m_{\theta,j}(\theta,x)$ at $(\theta_0,x_k)$.

Thus,
\begin{align*}
&\sqrt{n}\inf_{(s,s+t)\in B(x_k)}E_nm_j(W_i,\theta_0+a_n)I(s<X_i<s+t)  \\
&=\inf_{\|(s-k,t)\|\le C \sqrt{\|a_n\|}} \sqrt{n} \int_{s<x<s+t}
  \left[\frac{1}{2}(x-x_k)V_j(x_k)(x-x_k)'
+\bar m_{\theta,j}(\theta_0,x_k)a_n\right]
f_X(x_k) \, dx + o_p(1)  \\
&=\inf_{\|(s-x_k,t)\|\le C \sqrt{\|a\|}}
\int_{(s-x_k) <u<(s-x_k+t)}
  \left[\frac{1}{2}uV_j(x_k)u'
+\bar m_{\theta,j}(\theta_0,x_k)a_n\right]
f_X(x_k) \, du + o_p(1)
\end{align*}
where the last equality follows from the same change of variables and a
change of coordinates in $(s,t)$.  The result follows since, for large
enough $C$, the unconstrained infimum is taken on $\|(s-x_k,t)\|\le C
\sqrt{\|a\|}$, and $C$ can be chosen arbitrarily large.
\end{proof}

\bibliography{library}

\newpage

\begin{figure}[h]
  \centering
  \caption{Case with faster than root-$n$ convergence of KS statistic}
  \includegraphics[height=2.8in]{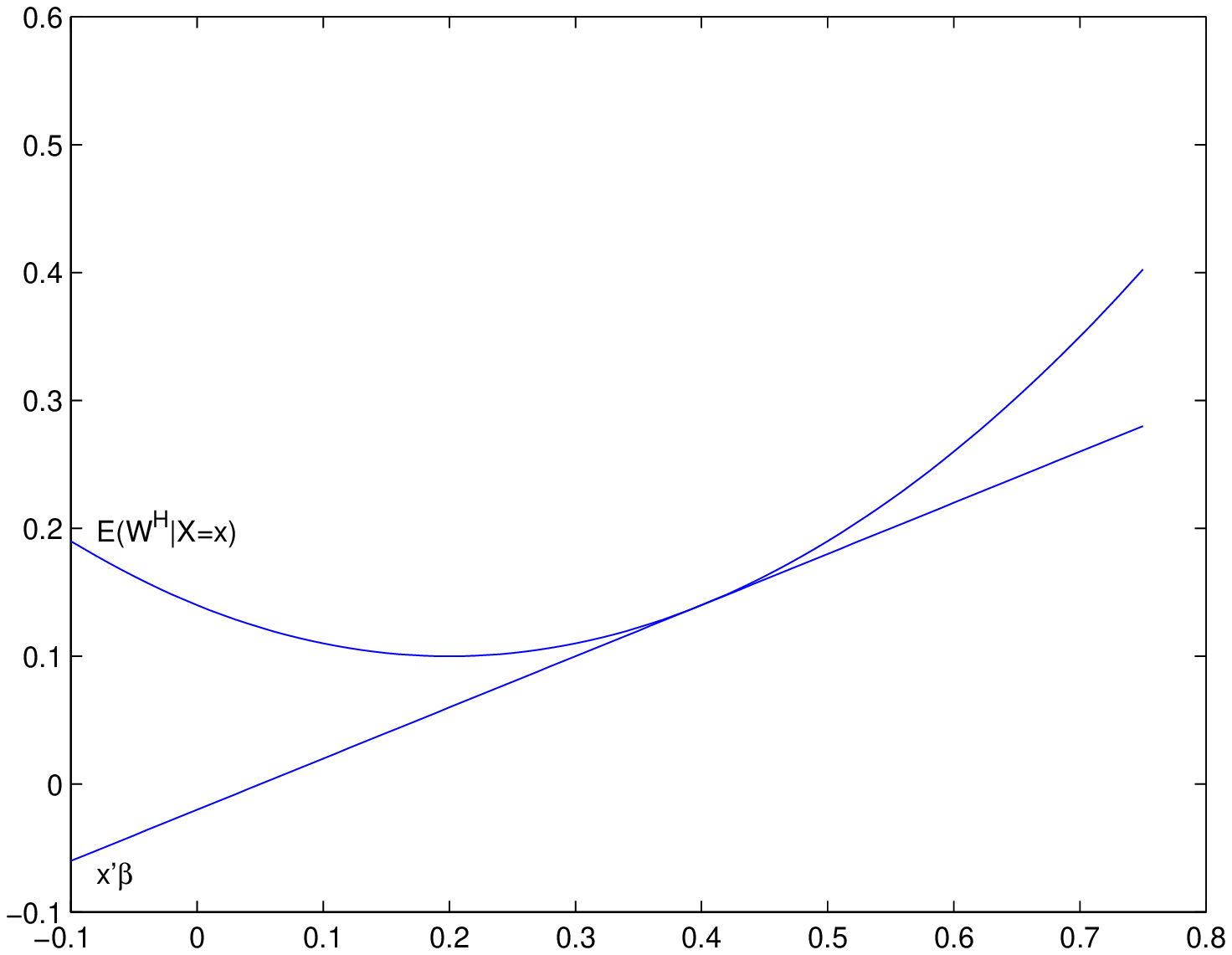}\label{int_reg_smooth_fig}
\end{figure}

\begin{figure}[h]
  \centering
  \caption{Cases with root-$n$ convergence of KS statistic ($\beta_1$) and
  faster rates ($\beta_2$)}
  \includegraphics[height=2.8in]{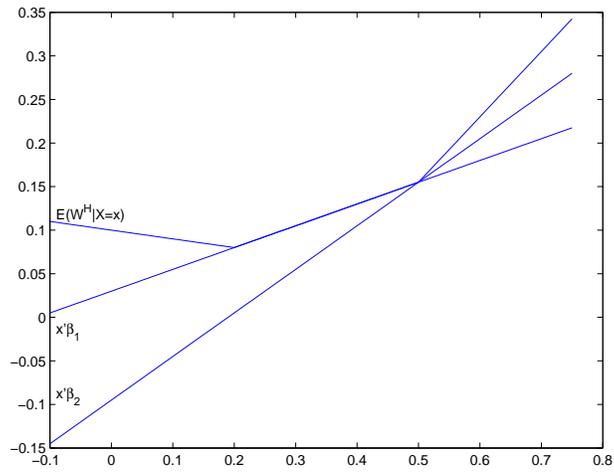}\label{int_reg_rootn_fig}
\end{figure}

\clearpage

\begin{figure}[h]
  \centering
  \caption{Conditional Means of $W_i^H$ and $W_i^L$ for Design 1}
  \includegraphics[height=2.8in]{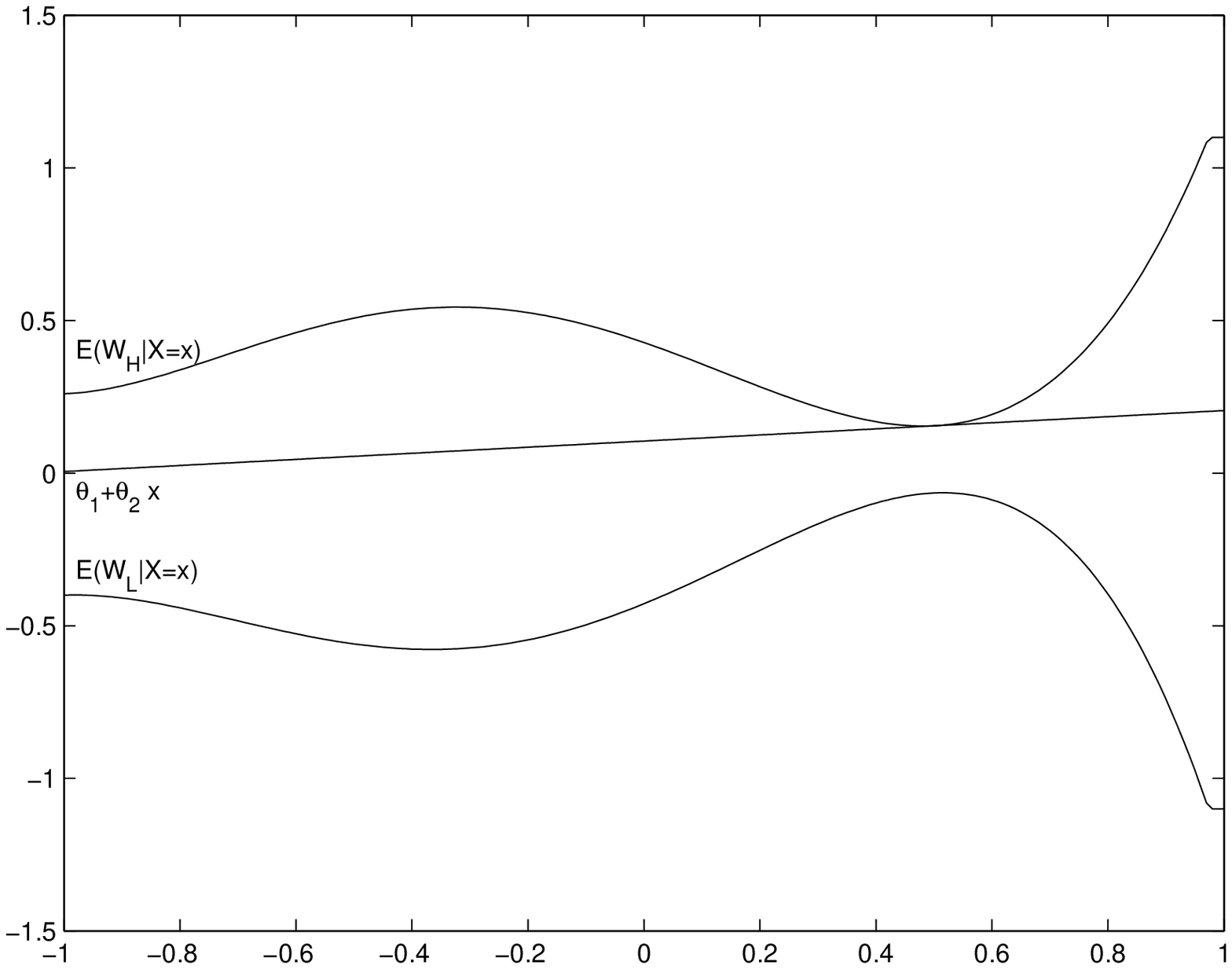}
    \label{cond_means_d1_fig}
\end{figure}

\begin{figure}[h]
  \centering
  \caption{Conditional Means of $W_i^H$ and $W_i^L$ for Design 2}
  \includegraphics[height=2.8in]{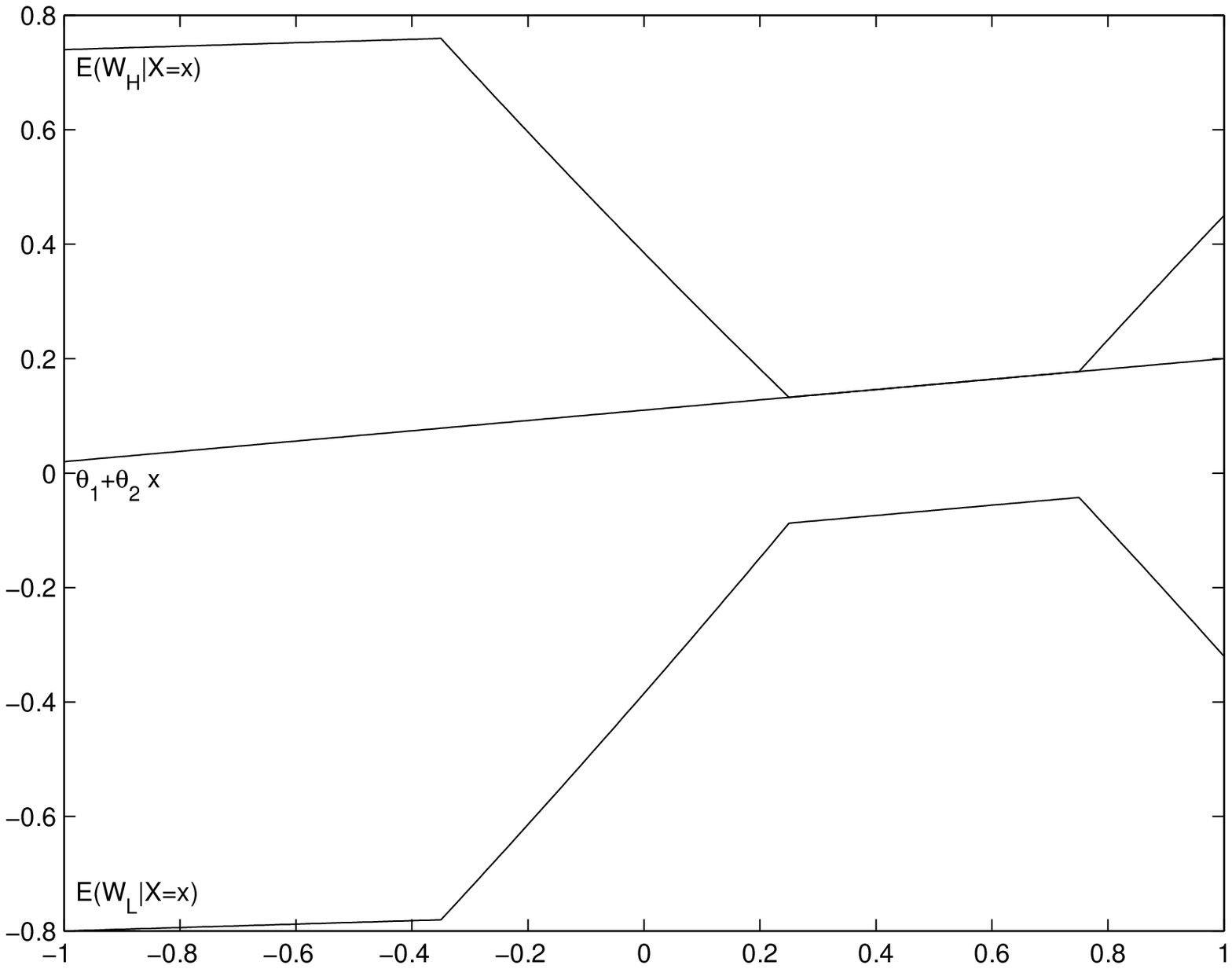}
    \label{cond_means_d3_fig}
\end{figure}

\begin{figure}[h]
  \caption{Histograms for $n^{3/5}S(T_n(\theta))$ for Design 1
    ($n^{3/5}$ Convergence)}
  \includegraphics[width=7in]
    {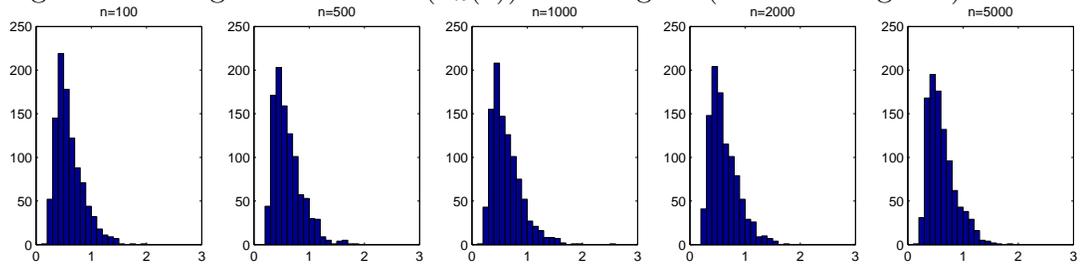}
    \label{ks_stat_hist_d1_fig}
\end{figure}

\begin{figure}[h]
  \caption{Histograms for $n^{1/2}S(T_n(\theta))$ for Design 2
    ($n^{1/2}$ Convergence)}
  \includegraphics[width=7in]
    {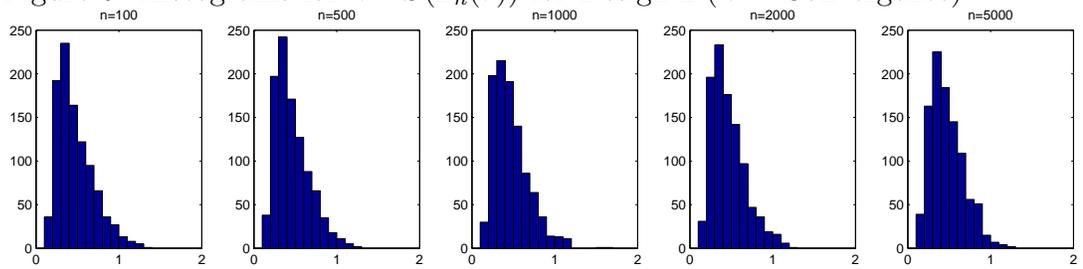}
    \label{ks_stat_hist_d3_fig}
\end{figure}

\begin{figure}[h]
  \centering
  \caption{Data for Empirical Illustration}
  \includegraphics{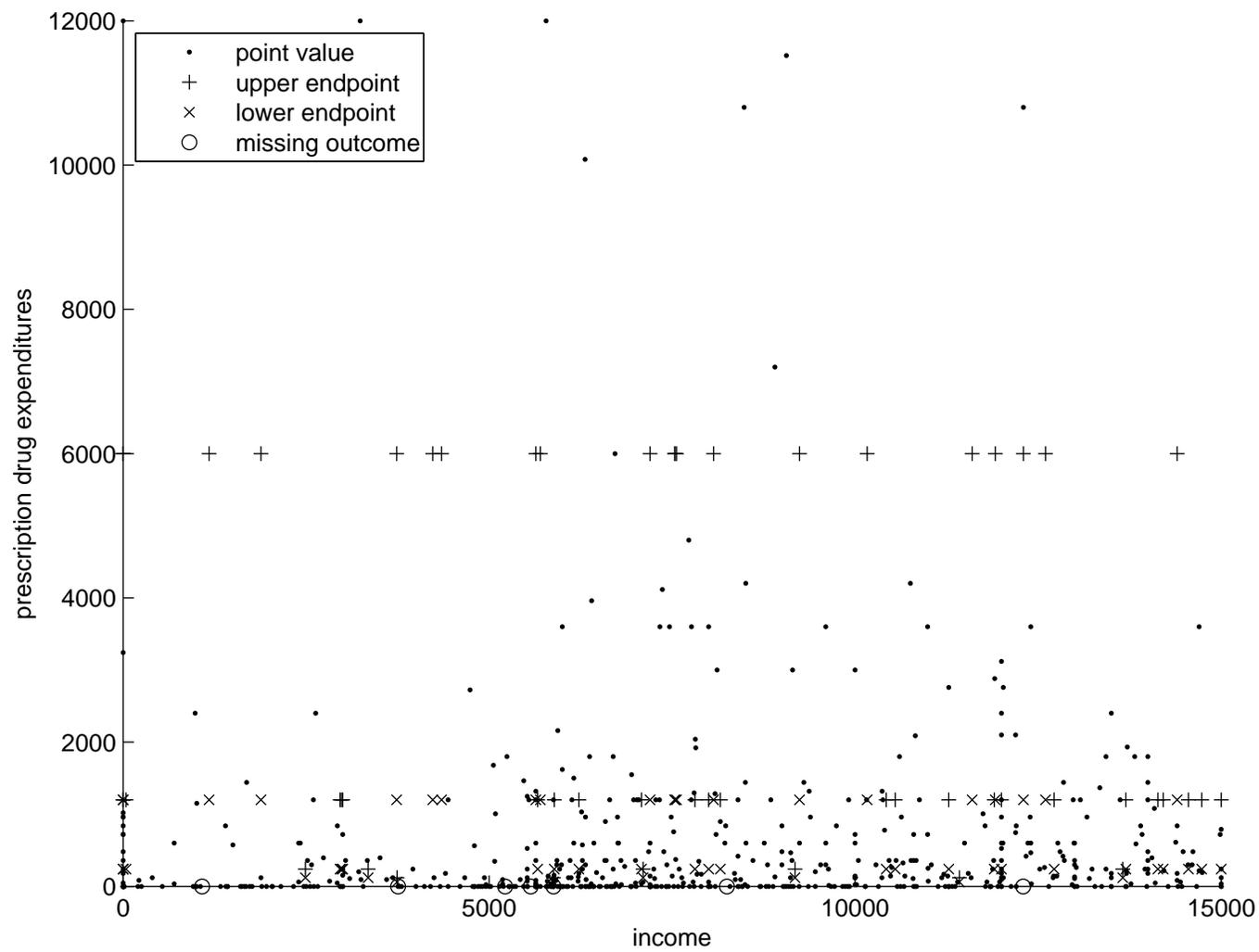}
    \label{data_plot_fig}
\end{figure}

\begin{figure}[h]
  \centering
  \caption{95\% Confidence Region Using Estimated Rate}
  \includegraphics[height=2.8in]{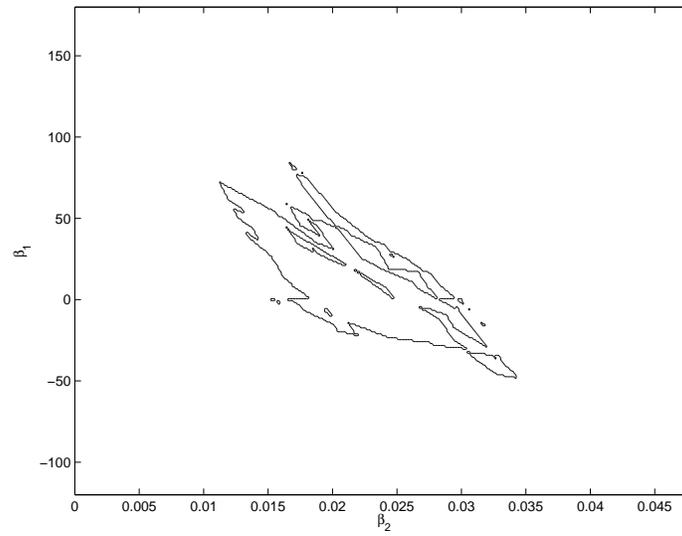}
    \label{cr95_est_fig}
\end{figure}

\begin{figure}[h]
  \centering
  \caption{95\% Confidence Region Using Conservative Rate}
  \includegraphics[height=2.8in]{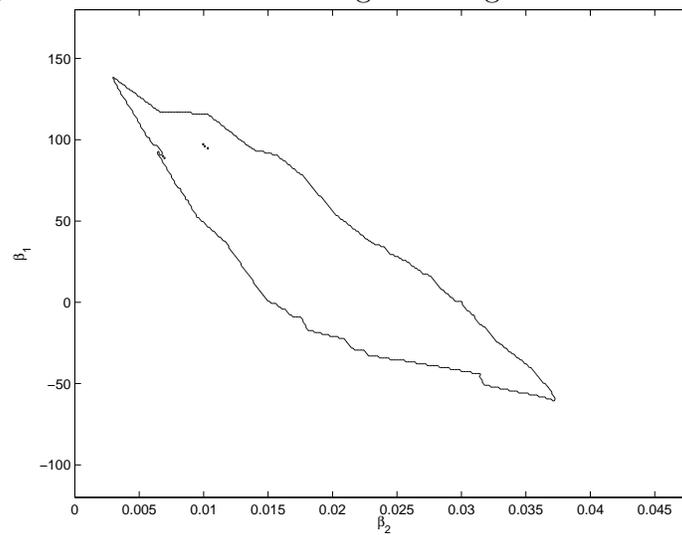}
    \label{cr95_cons_fig}
\end{figure}

\begin{figure}[h]
  \centering
  \caption{95\% Confidence Region Using LAD with Points}
  \includegraphics[height=2.8in]{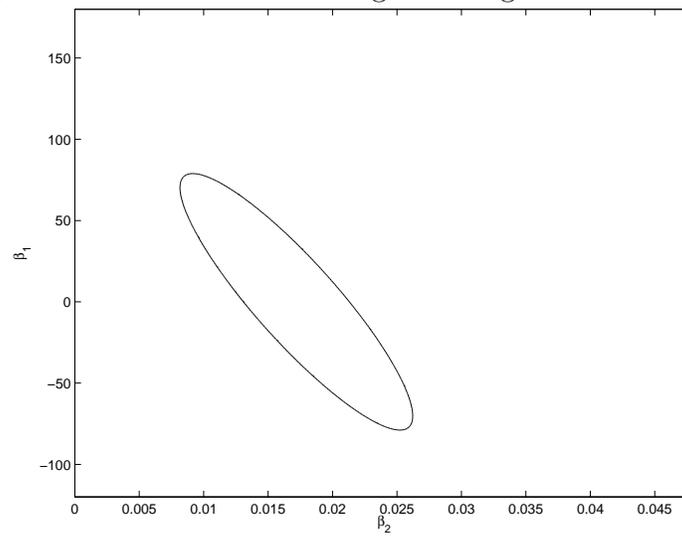}
    \label{cr95_wald_fig}
\end{figure}

\clearpage

\begin{table}[h]
\centering
\begin{tabular}{l|ccccc}
&$n=100$&$n=500$&$n=1000$&$n=2000$&$n=5000$\\\hline
\multicolumn{6}{c}{nominal 90\% coverage} \\\hline
estimated rate&0.873&0.890&0.897&0.889&0.879\\
conservative rate ($n^{1/2}$)&0.991&0.987&0.987&0.995&0.996\\
(infeasible) exact rate ($n^{3/5}$)&0.921&0.909&0.905&0.903&0.890\\\hline
\multicolumn{6}{c}{nominal 95\% coverage} \\\hline
estimated rate&0.940&0.943&0.954&0.947&0.934\\
conservative rate ($n^{1/2}$)&0.998&1.000&0.998&1.000&0.999\\
(infeasible) exact rate
  ($n^{3/5}$)&0.976&0.965&0.949&0.956&0.953\\\hline
\end{tabular}

\caption{Coverage Probabilities for Design 1}
\label{cov_prob_table_d1}
\end{table}

\begin{table}[h]
\centering
\begin{tabular}{l|ccccc}
&$n=100$&$n=500$&$n=1000$&$n=2000$&$n=5000$\\\hline
\multicolumn{6}{c}{nominal 90\% coverage} \\\hline
estimated rate&0.780&0.910&0.928&0.925&0.924\\
conservative rate ($n^{1/2}$)&0.949&0.947&0.938&0.932&0.924\\
(infeasible) exact rate ($n^{1/2}$)&0.949&0.947&0.938&0.932&0.924\\\hline
\multicolumn{6}{c}{nominal 95\% coverage} \\\hline
estimated rate&0.885&0.945&0.966&0.971&0.979\\
conservative rate ($n^{1/2}$)&0.991&0.982&0.975&0.974&0.979\\
(infeasible) exact rate ($n^{1/2}$)&0.991&0.982&0.975&0.974&0.979\\\hline
\end{tabular}

\caption{Coverage Probabilities for Design 2}
\label{cov_prob_table_d3}
\end{table}

\begin{table}[h]
\centering
\begin{tabular}{l|ccccc}
&$n=100$&$n=500$&$n=1000$&$n=2000$&$n=5000$\\\hline
\multicolumn{6}{c}{nominal 90\% coverage} \\\hline
estimated rate&0.26&0.13&0.08&0.06&0.03\\
conservative rate ($n^{1/2}$)&0.33&0.17&0.12&0.09&0.06\\
(infeasible) exact rate
  ($n^{3/5}$)&0.21&0.10&0.07&0.05&0.03\\\hline
\multicolumn{6}{c}{nominal 95\% coverage} \\\hline
estimated rate&0.35&0.17&0.11&0.07&0.05\\
conservative rate ($n^{1/2}$)&0.39&0.22&0.15&0.11&0.07\\
(infeasible) exact rate ($n^{3/5}$)&0.29&0.13&0.09&0.06&0.04\\\hline
\end{tabular}

\caption{Mean of $\hat u_{1-\alpha}-\theta_{1,D1}$ for Design 1}
\label{ci_length_table_d1}
\end{table}

\begin{table}[h]
\centering
\begin{tabular}{l|ccccc}
&$n=100$&$n=500$&$n=1000$&$n=2000$&$n=5000$\\\hline
\multicolumn{6}{c}{nominal 90\% coverage} \\\hline
estimated rate&0.11&0.08&0.06&0.04&0.02\\
conservative rate ($n^{1/2}$)&0.20&0.09&0.06&0.04&0.02\\
(infeasible) exact rate
  ($n^{1/2}$)&0.20&0.09&0.06&0.04&0.02\\\hline
\multicolumn{6}{c}{nominal 95\% coverage} \\\hline
estimated rate&0.18&0.10&0.07&0.05&0.03\\
conservative rate ($n^{1/2}$)&0.27&0.11&0.08&0.05&0.03\\
(infeasible) exact rate ($n^{1/2}$)&0.27&0.11&0.08&0.05&0.03\\\hline
\end{tabular}

\caption{Mean of $\hat u_{1-\alpha}-\theta_{2,D2}$ for Design 2}
\label{ci_length_table_d3}
\end{table}

\begin{table}
\centering
\begin{tabular}{c|c|c}
 & $\theta_1$ & $\theta_2$ \\\hline
Estimated Rate & $[-48, 84]$ & $[0.0113, 0.0342]$  \\
Conservative Rate & $[-60, 138]$ & $[0.0030, 0.0372]$  \\
LAD with Points & $[-63, 63]$ & $[0.0100, 0.0244]$  \\\hline
\end{tabular}
\caption{95\% Confidence Intervals for Components of $\theta$}
\label{ci_table}
\end{table}

\end{document}